\documentclass[aihp]{imsart}

\RequirePackage{amsthm,amsmath,amsfonts,amssymb}
\RequirePackage[numbers]{natbib}
\RequirePackage[colorlinks,citecolor=blue,urlcolor=blue]{hyperref}
\usepackage{csquotes}
\usepackage{dsfont}

\startlocaldefs
\theoremstyle{plain}
  \newtheorem{thm}{Theorem}[section]
  \newtheorem{lem}[thm]{Lemma}
  \newtheorem{prop}[thm]{Proposition}
  \newtheorem{cor}[thm]{Corollary} 
\theoremstyle{definition}
  \newtheorem{defn}[thm]{Definition}
  \newtheorem{rmk}[thm]{Remark}
  \newtheorem{ex}[thm]{Example}
  \newtheorem{asm}{Assumption}

\theoremstyle{plain}

\def\CD{\mathcal{D}}
\def\CF{\mathcal{F}}
\def\CG{\mathcal{G}}
\def\CH{\mathcal{H}}
\def\CI{\mathcal{I}}
\def\CJ{\mathcal{J}}

\def\1{\mathds{1}}

\def\BE{\mathbb{E}}

\def\BN{\mathbb{N}}
\def\BR{\mathbb{R}}
\def\BP{\mathbb{P}}

\def\Var{\text{Var}}
\def\Cov{\text{Cov}}

\newcommand{\indep}{\perp\!\!\!\!\!\!\perp} 

\endlocaldefs

\begin{document}

\begin{frontmatter}

\title{Consistency of Honest Decision Trees and Random Forests}
\runtitle{Consistency of Honest Decision Trees and Random Forests}

\begin{aug}
\author[A]{\inits{}\fnms{Martin}~\snm{Bladt}\ead[label=e1]{martinbladt@math.ku.dk}} \and
\author[A]{\inits{}\fnms{Rasmus Frigaard}~\snm{Lemvig}\ead[label=e2]{rfl@math.ku.dk}}
\address[A]{Department of Mathematical Sciences, University of Copenhagen, Copenhagen, 2100, DK\printead[presep={,\ }]{e1}, \printead[presep={,\ }]{e2}}

\end{aug}

\begin{abstract}
Various types of consistency of honest decision trees and random forests are considered in the regression setting. In contrast to related literature, the proofs are elementary and follow the classical arguments used for smoothing methods. Under mild regularity conditions on the regression function and data distribution, weak and almost sure convergence of honest trees and honest forest averages to the true regression function is established, and uniform convergence over compact covariate domains is obtained. The framework naturally accommodates ensemble variants based on subsampling and also a two-stage bootstrap sampling scheme. Existing analyses are synthesized and simplified, and hence existing gaps are indirectly bridged. The elementary nature of the arguments clarifies the close relationship between data-adaptive partitioning and kernel-type methods, providing an accessible approach to understanding the asymptotic behavior of tree-based methods.
\end{abstract}




\end{frontmatter}
\tableofcontents
\newpage

\section{Introduction}

Decision trees and their ensemble extensions are central tools in statistical machine learning due to their flexibility, interpretability, and excellent empirical performance. From a practical standpoint, methods such as CART \cite{CART} and Breiman's random forests \cite{BreimanRF,BreimanBagging} have become default choices across many applied fields. From
a theoretical viewpoint, however, providing simple and general consistency guarantees has been a persistent challenge because the
partitioning is highly data dependent and interacts with randomization and resampling steps in intricate ways. Early theoretical work established consistency for simplified or highly
structured variants of forests. For example, \cite{BiauDevroyeLugosi2008} proved universal consistency results for broad classes of averaging rules, \cite{LinJeon2006} connected certain forests to adaptive nearest-neighbour methods and \cite{Biau2012} analysed a centered random-forest model and showed adaptation to sparsity. $L^2$-consistency for non-adaptive splitting rules was proved in \cite{ScornetAsymp}, showing that some of the practical success of forests is theoretically tractable. Sharp bounds on the bias for many simple random forest models were derived in \cite{Klusowski}. Subsequent work further clarified the connection between forests and classical smoothers: \cite{ScornetKernel} recasts certain random-forest variants as kernel estimators and derives rates of consistency in that setup.
Other contributions include analyses of purely-random and centered
forests \cite{Genuer2012,Arlot2014}, investigations of the effect of
subsampling and tree depth \cite{RFSubsampling}, and uniform
consistency results for specialized forest forms such as random survival
forests \cite{RSF, RSFcon}.

A related and influential strand of research advocates for \emph{honest}
tree constructions, where the sample is split into disjoint parts used
separately for partitioning (split selection) and for estimation inside
leaves. Honesty simplifies bias analysis and underpins valid
inference procedures in causal and heterogeneous-effects contexts
\cite{WagerAthey,AtheyImbens2019}. Honesty also allows for combining conditioning arguments with classical statistical analysis, leading to more simple and transparent proofs. Consistency and inference for
honest forest variants have been developed in these works and in
follow-up studies addressing uncertainty quantification and inference
for heterogeneous effects \cite{WagerAthey}.

More recently, a number of authors obtained stronger and more general results for tree-based learners. As was shown in \cite{CattaneoChandakKlusowski}, convergence rates for oblique (multivariate-split) regression 
trees can be derived under structural assumptions, \cite{RFSubsampling} quantified
the interplay of subsampling and tree depth, and a body of work studied variance/bias trade-offs in various randomization schemes \cite{Genuer2012,Biau2012}. Nevertheless, many consistency proofs in the literature rely either on restrictive structural assumptions (e.g.\ additive models or sufficient-impurity-decrease conditions) or on technical arguments that obscure the connection to classical nonparametric smoothing. Furthermore, the inherent complexity of recursive, data-dependent partitioning has occasionally led to technical gaps or a reliance on heuristic arguments in existing derivations. Additionally, it is common in the literature to analyse only the entire forest instead of the estimators produced from individual trees. This makes theoretical properties of random forests difficult to translate to other types of estimators, especially ones that are functional in nature. Furthermore, Monte Carlo effects from bootstrapping are often not addressed, which makes it unclear what the effects of different types of bootstrapping can have on the final estimate. In the literature, this often takes place in the form of \emph{infinite forests}, which can be seen as a limit of the random forest estimate when the number of trees approaches infinity. The precise relationship between finite and infinite forests is analysed in \cite{ScornetAsymp} while \cite{MentchHooker} and \cite{WagerHastieEfron} investigate the effect of choosing a finite number of trees.

In this paper we present a streamlined analysis of honest regression trees and honest random forests that (i) uses elementary proofs modelled
on kernel-smoothing techniques, and remarkably our analysis takes place at the level of individual trees; (ii) yields transparent conditions for \emph{strong} (almost-sure) consistency, $L^p$ consistency and, crucially, \emph{uniform} convergence in
the covariate; and (iii) covers a broader class of ensemble schemes, including a double-bootstrap (two-stage subsampling/bootstrapping)
ensemble that, to our knowledge, has not been considered in consistency
analyses before. We aim to provide a robust theoretical foundation that bridges the gaps found in prior literature. Our assumptions are deliberately mild: honesty of the
tree estimator, control on minimal leaf mass, and modest regularity of the regression function (e.g.\ continuity or local Lipschitz conditions). Under these conditions we obtain almost-sure, uniform convergence to the true regression function under general conditions on the data generating process and splitting scheme. This contrasts with several prior works that either established only $L^2$ or pointwise consistency  \cite{ScornetAsymp}, or obtained rates under stronger technical
conditions \cite{CattaneoChandakKlusowski,RFSubsampling}. By recasting tree averages as adaptive kernel-like averages and transferring standard kernel arguments, our analysis both simplifies and generalises existing results: it recovers many classical consistency claims as special cases while providing a flexible framework that handles new sampling-based ensemble variants such as the double-bootstrap scheme. Our explicit handling of various bootstrap schemes also provides insight into the concrete assumptions needed for consistency to be retained for each tree in a random forest. 

An outline for the paper and our central results is as follows. Section~\ref{sec:singletree} concerns consistency for a single tree and starts by presenting the key assumptions under which the main consistency results are established. The subsequent paragraphs concern weak, strong and uniform consistency, respectively. The main results in this section are Theorems~\ref{thm:conspart1}, \ref{thm:strongconsistency} and \ref{thm:uniformconsistency} which establish weak, strong and uniform consistency under general conditions on the splitting rule. More tractable regularity conditions on the splits are subsequently discussed, culminating with Proposition~\ref{prop:regularIuniform} which establishes uniform strong consistency for such general schemes. The general results initially presented are further exemplified through two concrete non-adaptive splitting schemes. Section~\ref{sec:bootstrap} modifies the assumptions from Section~\ref{sec:singletree} to our proposed \enquote{double bootstrap} framework. The following subsections deal with weak and strong consistency for bootstrapped trees which immediately translate to consistency results for a whole forest. The main results are Theorems~\ref{thm:bootstraptreeconsistency} and \ref{thm:strongconsistencybootstrap}, establishing weak and strong convergence. Theorems~\ref{thm:weakconsistencybootstrap} and \ref{thm:regularstrongconsistency} then establish weak and strong consistency under more tractable regularity conditions, suitably adapted to handle the subsampling arising from bootstrapping. Throughout this section, all requirements on the bootstrap weights are exemplified through various schemes used in practical applications, including subsampling with or without replacement. Section~\ref{sec:proofs} contains the proofs of our results, while Appendix~\ref{sec:moreresults} contains additional results and background used in some of the proofs of Section~\ref{sec:proofs}. Finally, Section~\ref{sec:disc} concludes with a brief discussion and outlook.

\section{Consistency for a single tree} \label{sec:singletree}

In this section, we present an approach to establishing consistency of decision trees and random forests by a separate analysis of the node volume and the tree itself. To be more precise, we show that under an honesty condition, decision trees are consistent if the size of the leaf around the chosen prediction point shrinks to zero in a suitably regular fashion, and the number of observations used for growing the tree grows sufficiently fast compared to the rate that the volume of a leaf node shrinks. We start by presenting the assumptions, we need.

\subsection{Assumptions}

Throughout the paper, we assume that we are given iid training data $(X_i, Y_i)_{i=1}^n$ with $X_i \in \BR^d$ being $d$-dimensional features and $Y_i$ one-dimensional responses. We let $(X, Y)$ denote a generic observation and use $X^j$ to denote the $j$'th coordinate of $X$. Every random variable is defined on the background probability space $(\Omega, \CF, \BP)$. The goal is to estimate the conditional mean $\BE[Y \mid X = x]$ using decision trees and random forests. We use $\overset{\BP}{\to}$ to denote convergence in probability. All limits are understood to be for $n \to \infty$ unless otherwise stated.

The first assumption concerns the data generating process.

\begin{asm}[Covariates]\label{asm:covariate}
We assume that $X \in [0, 1]^d$ such that $X$ has continuous density $f$ bounded away from zero and infinity, $0 < \varepsilon \leq f \leq C < \infty$. We furthermore assume that $\BE[Y^2] < \infty$ with $\BE[Y^2 \mid X = x]$ continuous in $x \in [0, 1]^d$ and $\BE[Y^2 \mid X = x] \leq K$ for some constant $K$.
\end{asm}

In growing a tree, we have an auxiliary parameter $\theta$ independent of the data. This parameter determines the randomness injected into the splitting procedure. For multidimensional covariates, $\theta$ could control the chosen feature in a split, but $\theta$ could also play a role in determining the threshold after the feature was chosen. For given $x \in [0, 1]^d$, let $L(x, \theta)$ denote the leaf that $x$ falls into. The following honesty condition is the key assumption of the paper. 

\begin{asm}[Honesty]\label{asm:honesty}
Partition the data $(X_1, Y_1), ..., (X_n, Y_n)$ into two subsets $\CI$ and $\CJ$ with $\CI$ used only for prediction and $\CJ$ used for determining the splits. We assume that the number of samples in $\CI$ and $\CJ$ is chosen deterministically and that $n_\CI := |\CI| = \Theta(n)$ and $n_\CJ := |\CJ| = \Theta(n)$. 
\end{asm}

This honesty assumption differs from the one in~\cite{WagerAthey} presented for double-sample trees. In~\cite{WagerAthey} it is allowed that the features from $\CI$ play a role in placing splits. We only consider binary splits on single features such that a generic observation $X$ belongs to the left node if $X^j \leq c$ for some threshold $c$ in the $j$'th interval of the node and to the right node otherwise. Recent work, see~\cite{CattaneoChandakKlusowski}, studies asymptotic properties of trees grown using splits at linear combinations of variables, and extending our work in this direction is an interesting avenue for further research. 

By reordering the observations if necessary, we may assume that
\begin{equation*}
    \CI = \{(X_1, Y_1), ..., (X_{n_\CI}, Y_{n_\CI})\}.
\end{equation*}
Let
\begin{equation*}
	N_L(x, \theta) = \sum_{i = 1}^{n_\CI} \1_{\{X_i \in L(x, \theta)\}}
\end{equation*}
denote the number of observations among the prediction data $\CI$ which falls into the leaf $L(x, \theta)$. We emphasise that $L(x, \theta)$ is a function of $\CJ$ and $\theta$ only, but that $N_L(x, \theta)$ depends on both $\theta$ and the entire dataset. The tree predictor can then be written as
\begin{equation*}
	T(x, \theta) = \sum_{i = 1}^{n_\CI} \frac{\1_{\{X_i \in L(x, \theta)\}}}{N_L(x, \theta)}Y_i.
\end{equation*}

It is convenient to introduce more notation related to the leaf $L(x, \theta)$. We write
\begin{equation*}
    L(x, \theta) = (a_{n, 1}^x, b_{n, 1}^x] \times \cdots \times (a_{n, d}^x, b_{n, d}^x]
\end{equation*}
with, slightly abusing notation, $a_{n, j}^x := a_{n, j}^x(\theta, \CJ, s(x))$ and $b_{n, j}^x := b_{n, j}^x(\theta, \CJ, s(x))$ where $s(x)$ denotes the number of splits leading to the leaf $L(x, \theta)$. We furthermore let $\lambda$ denote Lebesgue measure, so that
\begin{equation*}
    \lambda(L(x, \theta)) = \prod_{j = 1}^d (b_{n, j}^x - a_{n, j}^x)
\end{equation*}
denotes the volume of the leaf $L(x, \theta)$. We can now formulate a property on the splitting rule which leads to consistency, as we shall see below.

\begin{asm}[{Consistency}]\label{asm:leaf}
We assume the following properties of the splitting rule of the tree:
\begin{enumerate}
    \item[(I)] $b_{n, j}^x - a_{n, j}^x \overset{\BP}{\to} 0$ for all $j = 1, ..., d$.
    \item[(II)] $n_\CI \lambda(L(x, \theta)) \overset{\BP}{\to} \infty$.
\end{enumerate}
\end{asm}

Points (I) and (II) are the main conditions of interest, and we later verify them for a large class of splitting schemes as well as some concrete ones. 

\subsection{Weak consistency of decision trees via honesty}

The goal of this subsection is to prove weak consistency of the tree predictor $T(x, \theta)$. Our approach is inspired by the proof of weak consistency for the Nadaraya--Watson kernel estimator, see~\cite{Ghosh} for background on kernel methods. The strategy is as follows. We introduce the estimator 
\begin{equation*}
	\widehat{f}(x, \theta) := \sum_{i = 1}^{n_\CI} \frac{\1_{\{X_i \in L(x, \theta)\}}}{n_\CI \lambda(L(x, \theta))}
\end{equation*}
for the density $f$. We prove that this estimator is indeed consistent for $f$, and then we prove consistency of $\widehat{m}(x, \theta) := T(x, \theta) \widehat{f}(x, \theta)$, the analogue of the Priestley--Chao estimator of $m(x) := \BE[Y \mid X = x]f(x)$. Combining these results yields consistency of $T(x, \theta)$. In order to show weak consistency of these components, we use a conditioning argument which relies on the following technical lemma which states that we may replace an $n$-dependent $\sigma$-algebra with a fixed one. In order to state the lemma, recall that we suppress the dependence on $n$ when we write $\CI$ and $\CJ$. Writing $\CI_n$ and $\CJ_n$ instead, we can think of these two sets as being built separately as $n$ grows, and we may consider
\begin{equation*}
    \CJ_\infty := \bigcup_{n = 1}^\infty \CJ_n.
\end{equation*}
\begin{lem}[Conditioning]\label{lem:convergenceconditioning}
Fix $\delta > 0$. It holds that
\begin{equation*}
    \BP(|\widehat{f}(x, \theta) - \BE[\widehat{f}(x, \theta) \mid \theta, \CJ]| > \delta \mid \theta, \CJ) = \BP(|\widehat{f}(x, \theta) - \BE[\widehat{f}(x, \theta) \mid \theta, \CJ]| > \delta \mid \theta, \CJ_\infty)
\end{equation*}
and that
\begin{equation*}
    \BP(|\widehat{m}(x, \theta) - \BE[\widehat{m}(x, \theta) \mid \theta, \CJ]| > \delta \mid \theta, \CJ) = \BP(|\widehat{m}(x, \theta) - \BE[\widehat{m}(x, \theta) \mid \theta, \CJ]| > \delta \mid \theta, \CJ_\infty).
\end{equation*}
\end{lem}

\begin{prop}[Weak consistency of the density estimator]\label{prop:densityestimatorconsistent}
Under Assumptions \ref{asm:covariate}, \ref{asm:honesty} and \ref{asm:leaf}, it holds that
\begin{equation*}
	\widehat{f}(x, \theta) \overset{\BP}{\to} f(x).
\end{equation*}
\end{prop}

The following result is one of the fundamental contributions of the paper.

\begin{thm}[Weak consistency of decision trees]\label{thm:conspart1}
Under Assumptions \ref{asm:covariate}, \ref{asm:honesty} and \ref{asm:leaf}, we have
\begin{equation*}
	T(x, \theta) \overset{\BP}{\to} \BE[Y \mid X = x].
\end{equation*}
\end{thm}

It is of interest to extend this result to stronger forms of convergence. The following result uses the result of the previous theorem to establish $L^p$-consistency.

\begin{thm}[$L_1$ and $L_p$ convergence of decision trees]\label{thm:L1consistency}
Under Assumptions \ref{asm:covariate}, \ref{asm:honesty} and \ref{asm:leaf}, it holds that
\begin{equation*}
    T(x, \theta) \overset{L^1}{\to} \BE[Y \mid X = x].
\end{equation*}
More generally, if we strengthen the moment condition of Assumption~\ref{asm:covariate} to $\BE[|Y|^{p + \delta} \mid X = x] \leq K$ for some constants $p \geq 1$ and $\delta, K > 0$, it holds that
\begin{equation*}
    T(x, \theta) \overset{L^p}{\to} \BE[Y \mid X = x].
\end{equation*}
\end{thm}

\subsection{Strong consistency}

In this subsection, we investigate conditions for strong consistency of $T(x, \theta)$. The strategy is again to consider the components $\widehat{f}(x, \theta)$ and $\widehat{m}(x, \theta)$ separately and using the general conditioning results in the appendix.

\begin{prop}[Strong consistency of the density estimator]
\label{prop:densityestimatorstrongconvergence}
Assume that Assumptions \ref{asm:covariate}, \ref{asm:honesty} and \ref{asm:leaf} hold where convergence in probability in the latter has been replaced with a.s. convergence. If in addition,
\begin{equation*}
    \sum_{n_\CI = 1}^\infty \exp(-\delta n_\CI \lambda(L(x, \theta))) < \infty \quad \BP\text{-a.s.}
\end{equation*}
for any $\delta > 0$, then
\begin{equation*}
    \widehat{f}(x, \theta) \to f(x) \quad \BP\text{-a.s.}
\end{equation*}
\end{prop}

In the case where the response variable $Y$ is bounded, one may essentially copy the proof of Proposition~\ref{prop:densityestimatorstrongconvergence} to show that $\widehat{m}(x, \theta) \to m(x)$ a.s. and thus that $T(x, \theta) \to \BE[Y \mid X = x]$ a.s. The following theorem establishes the same result under significantly weaker assumptions.

\begin{thm}[Strong consistency of decision trees]\label{thm:strongconsistency}
Assume that Assumptions \ref{asm:covariate}, \ref{asm:honesty} and \ref{asm:leaf} hold, where convergence in probability has been replaced by a.s. convergence. If one of the following sets of conditions hold, we have
\begin{equation*}
    T(x, \theta) \to \BE[Y \mid X = x] \quad \BP\text{-a.s.}
\end{equation*}
\begin{enumerate}
    \item[(i)] $\BE[Y^4] < \infty$ and 
    \begin{equation*}
        \sum_{n_\CI = 1}^\infty \frac{1}{n_\CI^2 \lambda(L(x, \theta))^2} < \infty \quad \BP\text{-a.s.}
    \end{equation*}
    \item[(ii)] The moment generating function
    \begin{equation*}
        \kappa_{Y\mid X= x}(t) := \BE[e^{tY} \mid X = x]
    \end{equation*}
    is finite in a neighbourhood $(-c, c)$ of zero which does not depend on $x \in [0, 1]^d$, $x \mapsto \kappa_{Y \mid X =x}(t)$ is continuous for all $t \in (-c, c)$ and
    \begin{equation*}
        \sum_{n_\CI = 1}^\infty \exp(-\delta n_\CI \lambda(L(x, \theta))) < \infty \quad \BP\text{-a.s.}
    \end{equation*}
    for any $\delta > 0$.
\end{enumerate}
\end{thm}

\subsection{Uniform consistency}

It is of interest to establish uniform consistency in the test value $x \in [0,1]^d$. The strategy is to establish uniform consistency in $\widehat{m}$ and $\widehat{f}$, which then translates into uniform consistency of $T$ due to compactness of the covariate domain. The following theorem contains both a weak and a strong uniform convergence statement for the components. In the theorem, the infimum over all possible leaf volumes,
\begin{equation*}
    \underline{v}_n := \inf_{x \in [0, 1]^d} \lambda(L(x, \theta))
\end{equation*}
is an essential quantity, and describing the exact behaviour of this infimum is essential for establishing uniform convergence results for concrete splitting schemes. We make the remark that, despite the supremum running over an uncountable set, this infimum is actually measurable, since we only have a finite number of leaves.

The only difference in establishing weak and strong uniform convergence is the requirement on the asymptotic behaviour of the quantity $\underline{v}_n$ and an auxiliary sequence $\{M_n\}$ which is relevant only for the case where $Y$ is unbounded, see also the corollary below.

\begin{thm}[Uniform consistency of decision trees]\label{thm:uniformconsistency}

Let Assumptions \ref{asm:covariate} and \ref{asm:honesty} hold, and assume furthermore that $f$ and $m$ are Lipschitz continuous and that the moment generating function $\kappa_Y$ of $Y$ exists in a neighbourhood of zero.
\begin{enumerate}
    \item[(i)] (Weak uniform consistency): If
    \begin{equation*}
        \sup_{x \in [0, 1]^d} \{b_{n, j}^x - a_{n, j}^x\} \overset{\BP}{\to} 0
    \end{equation*}
    for all $j = 1, \dots, d$ and there exists a sequence $\{M_n\}$ such that
    \begin{equation*}
        \frac{e^{-tM_n}}{\underline{v}_n} \overset{\BP}{\to} 0 \quad \text{and} \quad \max\Big\{\Big(\frac{\underline{v}_n \sqrt{n_\CI}}{M_n} \Big)^{4d} , 1\Big\} e^{-\delta \underline{v}_n^2 n_\CI/M_n^2} \overset{\BP}{\to} 0 
    \end{equation*}
    for any $t, \delta > 0$, then the tree estimator $T$ is weakly uniformly consistent in $x$,
    \begin{equation*}
        \sup_{x \in [0, 1]^d} |T(x, \theta) - \BE[Y \mid X = x]| \overset{\BP}{\to} 0.
    \end{equation*}
    \item[(ii)] (Strong uniform consistency): If
    \begin{equation*}
        \sup_{x \in [0, 1]^d} \{b_{n, j}^x - a_{n, j}^x\} \to 0 \quad \BP\text{-a.s.}
    \end{equation*}
    for all $j = 1, \dots, d$ and there exists a sequence $\{M_n\}$ such that
    \begin{equation*}
        \sum_{n = 1}^\infty \frac{e^{-tM_n}}{\underline{v}_n} < \infty \quad \text{and} \quad \sum_{n = 1}^\infty \max\Big\{\Big(\frac{\underline{v}_n \sqrt{n_\CI}}{M_n} \Big)^{4d} , 1\Big\} e^{-\delta \underline{v}_n^2 n_\CI/M_n^2} < \infty \quad \BP\text{-a.s.}
    \end{equation*}
    for any $t, \delta > 0$, then the tree estimator $T$ is strongly uniformly consistent in $x$:
    \begin{equation*}
        \sup_{x \in [0, 1]^d} |T(x, \theta) - \BE[Y \mid X = x]| \to 0 \quad \BP\text{-a.s.}
    \end{equation*}
\end{enumerate}
\end{thm}

The following corollary follows immediately from the proof of the previous theorem.

\begin{cor}[Uniform consistency of decision trees with bounded responses]
Let Assumptions \ref{asm:covariate} and \ref{asm:honesty} hold, and assume furthermore that $f$ and $m$ are Lipschitz continuous and that the response $Y$ is bounded.
\begin{enumerate}
    \item (Weak uniform consistency): If
    \begin{equation*}
        \sup_{x \in [0, 1]^d} \{b_{n, j}^x - a_{n, j}^x\} \overset{\BP}{\to} 0
    \end{equation*}
    for all $j = 1, \dots, d$ and
    \begin{equation*}
        \max\{(\underline{v}_n \sqrt{n_\CI})^{4d}, 1\}e^{-\delta \underline{v}_n^2 n_\CI} \overset{\BP}{\to} 0
    \end{equation*}
    for any $\delta > 0$, then the tree estimator $T$ is weakly uniformly consistent in $x$,
    \begin{equation*}
        \sup_{x \in [0, 1]^d} |T(x, \theta) - \BE[Y \mid X = x]| \overset{\BP}{\to} 0.
    \end{equation*}
    \item (Strong uniform consistency): If
    \begin{equation*}
        \sup_{x \in [0, 1]^d} \{b_{n, j}^x - a_{n, j}^x\} \to 0 \quad \BP\text{-a.s}
    \end{equation*}
    for all $j = 1, \dots, d$ and
    \begin{equation*}
        \sum_{n = 1}^\infty \max\{(\underline{v}_n \sqrt{n_\CI})^{4d}, 1\}e^{-\delta \underline{v}_n^2 n_\CI} < \infty \quad \BP\text{-a.s.}
    \end{equation*}
    for any $\delta > 0$, then the tree estimator $T$ is strongly uniformly consistent in $x$,
    \begin{equation*}
        \sup_{x \in [0, 1]^d} |T(x, \theta) - \BE[Y \mid X = x]| \to 0 \quad \BP\text{-a.s.}
    \end{equation*}
\end{enumerate}
\end{cor}


\subsection{General splitting criteria}

We now investigate criteria on the tree-growing procedure which ensure that the conditions (I) and (II) of Assumption~\ref{asm:leaf} hold. We start by considering general restrictions on the splitting rule and later consider more concrete splitting schemes. 



Our approach is inspired by the regularity assumptions of Definitions 3 and 4 in~\cite{WagerAthey}. As we derive stronger convergence results, we gradually make more restrictive assumptions. The following proposition gives an equivalent criterion for condition (I) of Assumption~\ref{asm:leaf} to hold. We let
\begin{equation*}
    N_L^j(x, \theta) = \sum_{i = 1}^{n_\CI} \1_{\{X_i^j \in (a_{n, j}^x, b_{n, j}^x]\}}
\end{equation*}
denote the number of $\CI$-observations whose $j$'th coordinate falls into the $j$'th interval of the leaf $L(x, \theta)$. 

\begin{prop}[Consistency under balance condition]\label{prop:regularI}
Under Assumptions \ref{asm:covariate} and \ref{asm:honesty}, we have that condition (I) of Assumption~\ref{asm:leaf} holds if and only if 
\begin{equation}\label{eq:jproportion}
    \frac{N_L^j(x, \theta)}{n_\CI} \overset{\BP}{\to} 0 \quad \text{for all } j = 1, \dots, d.
\end{equation}
The equivalence also holds with convergence in probability replaced with almost sure convergence. 
\end{prop}

Later we present stronger conditions which imply~\eqref{eq:jproportion}. Intuitively, condition (I) of Assumption~\ref{asm:leaf} ensures that the bias decreases as $n \to \infty$, while (II) ensures that the variance of the predictions also decrease sufficiently fast. In order for (II) to hold, we need to ensure that the leaves do not shrink too quickly. The idea is to require, with probability increasing to one, a minimal number of observations $k_n$ in a node. More precisely, we consider the following set of assumptions.
\begin{asm}[Minimal node size]\label{asm:minimalnodesize}
Let $(k_n)$ denote a deterministic sequence. Define $a_n(x) := \BP(N_L(x, \theta) \geq k_n)$ and $a_n := \inf_{x \in [0, 1]^d}a_n(x)$. We then consider the following set of conditions on $k_n$, $a_n(x)$ and $a_n$:
\begin{enumerate}
    \item[(i)] $a_n(x) \to 1$,
    \item[(ii)] $\sum_{n = 1}^\infty (1 - a_n(x)) < \infty$,
    \item[(iii)] $a_n \to 1$,
    \item[(iv)] $\sum_{n = 1}^\infty n_\CI(1 - a_n) < \infty$.
\end{enumerate}
\end{asm}
The following result shows that with the proper restriction on $k_n$, we can ensure that (II) holds.


\begin{prop}[Consistency under minimal node size]\label{prop:regularII}
If a tree is honest in the sense of Assumption~\ref{asm:honesty} and satisfies (i) of Assumption \ref{asm:minimalnodesize} where $k_n$ satisfies
\begin{equation*}
    n_\CI^{4d}e^{-k_n^2/2n_\CI} \to 0,
\end{equation*}
then (II) of Assumption~\ref{asm:leaf} holds. If furthermore (ii) of Assumption \ref{asm:minimalnodesize} holds and $k_n$ satisfies
\begin{equation*}
    \sum_{n_\CI = 1}^\infty \frac{n_\CI^{4d}}{e^{k_n^2/2n_\CI}} < \infty,
\end{equation*}
then (II) of Assumption~\ref{asm:leaf} holds almost surely.
\end{prop}

The observations made in the proof of Proposition~\ref{prop:regularII} can be used to formulate the following result for strong consistency.

\begin{prop}[Strong consistency under minimal node size]\label{prop:regularIIstrong}
Consider an honest tree satisfying~\eqref{eq:jproportion} almost surely and (ii) of Assumption \ref{asm:minimalnodesize} with
\begin{equation*}
    \sum_{n_\CI = 1}^\infty \frac{n_\CI^{4d}}{e^{k_n^2/2n_\CI}} < \infty.
\end{equation*}
The tree is strongly consistent if one of the following sets of conditions apply.
\begin{enumerate}
    \item[(i)] We have $\BE[Y^4] < \infty$ and
    \begin{equation*}
        \sum_{n_\CI = 1}^\infty \frac{1}{k_n^2} < \infty.
    \end{equation*}
    \item[(ii)] The moment generating function of $Y \mid X = x$ exists and satisfies the requirements in Theorem~\ref{thm:strongconsistency} and
    \begin{equation*}
        \sum_{n_\CI = 1}^\infty e^{-\delta k_n} < \infty
    \end{equation*}
    for every $\delta > 0$.
\end{enumerate}
\end{prop}

Several choices for $k_n$ exist that satisfy the requirements for weak or strong consistency. Some are listed in the following corollary.

\begin{cor}[Concrete minimal node size sequences]\label{cor:regularconsistent}
The following minimal node size sequences satisfy all summability requirements in Proposition \ref{prop:regularIIstrong}.
\begin{enumerate}
    \item[(1)] $k_n \sim  n_\CI^\beta$ for $\beta \in (1/2, 1)$.
    \item[(2)] $k_n \sim  \sqrt{n_\CI \log(n_\CI)^\beta}$ for $\beta > 1$.
\end{enumerate}
\end{cor}

This choice even suffices for strong uniform consistency as long as stronger assumptions on $N_L^j$ are satisfied which take uniformity into account. Proposition~\ref{prop:regularI} also holds uniformly.

\begin{prop}[Uniform consistency under equivalent shrinking condition]\label{prop:regularIuniform}
Under Assumptions \ref{asm:covariate} and \ref{asm:honesty}, we have that
\begin{equation*}
    \sup_{x \in [0, 1]^d} \{b_{n, j}^x - a_{n, j}^x\} \overset{\BP}{\to} 0 \quad \text{for all } j = 1, \dots, d
\end{equation*}
if and only if
\begin{equation}\label{eq:jproportionsuniform}
    \frac{\sup_{x \in [0, 1]^d} N_L^j(x, \theta)}{n_\CI} \overset{\BP}{\to} 0 \quad \text{ for all } j = 1, \dots, d.
\end{equation}
The same equivalence holds with convergence in probability replaced with convergence almost surely. 
\end{prop}

\begin{prop}[Strong uniform consistency with polynomial minimal node size]\label{prop:regularIIstronguniform}
Consider an honest tree satisfying~\eqref{eq:jproportionsuniform} and (iv) of Assumption \ref{asm:minimalnodesize} with $k_n \sim n_\CI^\beta$ for $\beta \in (1/2, 1)$. If $f$ and $m$ are Lipschitz continuous and the moment generating function of $Y$ exists in a neighbourhood around zero, the tree estimator $T(x, \theta)$ is strongly uniformly consistent in $x$.
\end{prop}

It is natural to ask what assumptions imply the balance conditions~\eqref{eq:jproportion} and~\eqref{eq:jproportionsuniform}. One such assumption is a global regularity assumption and a corresponding covariate-wise regularity condition. These assumptions are reminiscent of those presented in (see~\cite{WagerAthey}), but are presented more generally as probability statements. For the following discussion, let $s_n(x) := s_n(x, \theta, \CJ)$ denote the number of splits in the tree leading to the leaf $L(x, \theta)$ and let $s_{n, j}(x) := s_{n, j}(x, \theta, \CJ)$ be the number of splits along covariate $j$ leading to $L(x, \theta)$.

\begin{defn}[Regularity]
A decision tree is called $(\alpha, k_n)$-\emph{regular} for $\alpha \in (0, 1/2]$ if, with probability tending to one, at least a fraction $\alpha$ of the $\CI$-observations in the parent node are in each daughter node, and if the number of observations in a leaf satisfies $k_n \leq N_L(x, \theta) \leq 2k_n - 1$. To be precise, $(\alpha, k_n)$-regularity means that 
\begin{align*}
    \overline{a}_n(x) &:= \BP(k_n \leq N_L(x, \theta) \leq 2k_n - 1) \to 1, \\
    b_n(x) &:= \BP(n_\CI\alpha^{s_n(x)} \leq N_L(x, \theta)\leq n_\CI(1 - \alpha)^{s_n(x)}) \to 1.
\end{align*}
We call a tree \emph{covariate-wise regular} with parameters $\alpha_j \in (0, 1/2]$ for $j = 1, \dots, d$ if, with probability tending to one, at every split on covariate $j$, at least a fraction $\alpha_j$ of the $\CI$-observations in the parent node have $j$'th coordinate in the $j$-interval for each of the two daughter nodes, that is, if
\begin{equation*}
    b_n^j(x) := \BP(n_\CI\alpha_j^{s_{n, j}(x)} \leq N_L^j(x, \theta) \leq n_\CI(1 - \alpha_j)^{s_{n, j}(x)}) \to 1 \quad \text{for all } j = 1, \dots, d.
\end{equation*}
\end{defn}

Intuitively, if the leaf sizes are (with high probability) not allowed to grow too large compared to $n_\CI$, the total number of splits should grow to infinity. This turns out to be the case. In order to also conclude that the number of splits along each coordinate goes to infinity, the following assumption is key, see also~\cite{WagerAthey}.

\begin{defn}[Random-split tree]
A decision tree is called a \emph{random-split} tree if at each split, the probability of splitting on feature $j$ is bounded from below by a constant $\overline{\pi}_j > 0$
\end{defn}

\begin{lem}[Number of splits divergence]\label{lem:numsplits}
Under the assumptions of $(\alpha, k_n)$-regularity and $k_n/n \to 0$ as $n \to \infty$, it holds that the total number of splits $s_n(x)$ leading to the leaf $L(x, \theta)$ satisfies $s_n(x) \to \infty$ in probability. Furthermore, if the tree satisfies the random-split assumption, $s_{n, j}(x) \to \infty$ in probability as well and in particular,~\eqref{eq:jproportion} holds.
\end{lem}

\begin{cor}[Consistency of regular trees]
An honest $(\alpha, k_n)$-regular random-split tree satisfying the covariate-wise regularity condition satisfies (I) of Assumption~\ref{asm:leaf}.
\end{cor}

The following lemma indicates how the regularity assumptions translate into results on uniform consistency.

\begin{lem}[Uniform divergence of number of splits for regular trees]\label{lem:regularsidelength}
Consider an honest covariate-wise regular decision tree and let Assumption~\ref{asm:covariate} hold. Then if
\begin{equation*}
    \inf_{x \in [0, 1]^d} s_{n, j}(x) \overset{\BP}{\to} \infty
\end{equation*}
for all $j$, it holds that
\begin{equation*}
    \sup_{x \in [0, 1]^d}\{b_{n, j}^x - a_{n, j}^x\} \overset{\BP}{\to} 0
\end{equation*}
for all $j$.
\end{lem}

We find that the assumption
\begin{equation*}
    \inf_{x \in [0, 1]^d} s_{n, j}(x) \overset{}{\to} \infty
\end{equation*}
in some mode of convergence is difficult to verify without further assumptions. A way to ensure this condition is to enforce a split on covariate $j$ every time the total number of splits leading to $x$ (namely $s_n(x)$) increases by a certain large but fixed number $N_j$ ($\geq d$). From the proof of Lemma~\ref{lem:numsplits}, we then have
\begin{equation*}
    s_{n, j}(x) \geq \frac{s_n(x)}{N_j} \geq \frac{\log(n_\CI/(2k_n - 1))}{N_j \log(\alpha^{-1})}
\end{equation*}
with high probability. 

\subsection{Concrete splitting criteria}

In this subsection, we investigate some common splitting schemes found in the literature and show how they are integrated into our framework. The presented splitting schemes are totally non-adaptive, that is, the splits are not data-driven. Hence the honesty condition is superfluous to establish consistency, but we include these schemes to exemplify our findings.

\subsubsection{Uniform splitting}

We start by considering uniform splitting, adapted from the uniform forest presented in~\cite{ScornetAsymp}. The scheme works as follows:

\begin{enumerate}
    \item[(i)] In a node, a splitting variable is chosen uniformly among $\{1, ..., d\}$.
    \item[(ii)] For any feature, a threshold is chosen uniformly in the range of the node. 
    \item[(iii)] The tree is grown to be balanced such that the algorithm terminates when every terminal node has been split $s_n$ times. 
\end{enumerate}

To see that $b_{n, j}^x - a_{n, j}^x \to 0$ $\BP$-a.s., let $U_k$ denote the proportion of the node volume in split $k$ retained from split $k - 1$. Then $\{U_k\}_k$ is an iid $U[0, 1]$ sequence. Note that this sequence can be assumed to be independent of $\theta$ so that $\theta$ only governs the feature being split upon. We then have
\begin{equation*}
    \lambda(L(x, \theta)) \overset{d}{=} \prod_{k = 1}^{s_n} U_k.
\end{equation*}
Let $s_{n, j}(\theta)$ denote the number of splits along feature $j$. Then similarly,
\begin{equation*}
    b_{n, j}^x - a_{n, j}^x \overset{d}{=} \prod_{k = 1}^{s_{n, j}(\theta)} U_{k, j}
\end{equation*}
with $\{U_{k, j}\}_k$ iid $U[0, 1]$ as well. The following two results establish (I) and (II) of Assumption~\ref{asm:leaf} for proper conditions on the number of splits $s_n$. 

\begin{lem}[Consistency shrinkage condition under uniform splitting]\label{lem:uniformconsistency}
Under uniform splitting, assuming that $s_n \uparrow \infty$, then (I) of Assumption~\ref{asm:leaf} holds.
\end{lem}

\begin{prop}[Consistency growth condition under uniform splitting]\label{prop:uniformconsistency}
Under uniform splitting, if we have that $s_n \to \infty$ and that $\lim_{n \to \infty} {\log(n_\CI)}/{s_n}$ exists and is strictly larger than one, then (II) of Assumption~\ref{asm:leaf} holds almost surely.
\end{prop}

\begin{rmk}\label{rmk:uniform}
In the literature, see e.g.~\cite{ScornetAsymp}, the uniform random forest is shown to be consistent whenever $s_n \to \infty$ and $2^{s_n}/n_\CI \to 0$. Say $\lim_{n \to \infty} {\log(n_\CI)}/{s_n}$ exists. Then
\begin{align*}
    \lim_{n \to \infty }\frac{2^{s_n}}{n_\CI} = 0 
    \Leftrightarrow \quad \lim_{n \to \infty} \frac{\log(n_\CI)}{s_n} > \log(2).
\end{align*}
This shows that the requirement in the above proposition is similar (with a different constant) to  $2^{s_n}/n_\CI \to 0$. However, our results provide the more delicate consistency of a single tree, without needing to average across a whole forest. 
\end{rmk}

A concrete choice of $s_n$ could be
\begin{equation*}
    s_n = \Big\lfloor \frac{\log(n_\CI)}{1 + \varepsilon} \Big\rfloor
\end{equation*}
for some $\varepsilon > 0$, as then $\log(n_\CI)/s_n \geq 1 + \varepsilon$, and we obviously have $s_n \to \infty$. Hence for this choice of $s_n$, we get weak consistency of a decision tree grown using uniform splitting. 

\subsubsection{Centered splitting}

Centered splitting (see e.g.~\cite{Klusowski}) is another example of a non-adaptive splitting scheme and works as follows:

\begin{enumerate}
    \item[(i)] At every split, select a splitting feature $j$ at random with probability $p_j > 0$.
    \item[(ii)] Split the node in the midpoint of the interval of the chosen feature. 
    \item[(iii)] The tree is grown to be balanced with a total of $2^{s_n}$ leaf nodes, implying that each leaf has been split exactly $s_n$ times. 
\end{enumerate}

In this case, we obviously have
\begin{equation*}
    \lambda(L(x, \theta)) = 2^{-s_n}.
\end{equation*}
Choosing $s_n$ in a way such that $2^{s_n}/n_\CI \to 0$ for $n \to \infty$, we ensure that (II) holds. If also $s_n \to \infty$ and $p_j > 0$ for all $j = 1, ..., d$, a split on feature $j$ happens infinitely often with probability one (formally, this holds due to the Borel--Cantelli lemma), and so $b_{n, j}^x - a_{n, j}^x \to 0$ $\BP$-a.s. for every $j$. To see why, we make the observation
\begin{equation*}
    b_{n, j}^x - a_{n, j}^x = 2^{-s_{n, j}(x)}
\end{equation*}
with $s_{n, j}(x)$ the number of splits along feature $j$ leading to $x$. Since $s_{n, j}(x) \sim \text{Bin}(s_n, p_j)$ and $s_n \to \infty$, it holds that $s_{n, j}(x) \to \infty$ $\BP$-a.s. and hence (I) also holds a.s. A concrete choice of $s_n$ is
\begin{equation*}
    s_n = \lceil \log_2 n_\CI^{1 - \beta} \rceil
\end{equation*}
for some $\beta \in (0, 1)$. With this choice, we have $s_n \to \infty$ and 
\begin{equation*}
    n_\CI\lambda(L(x, \theta)) = \frac{n_\CI}{2^{s_n}} > \frac{n_\CI}{2^{\log_2 n_\CI^{1 - \beta} + 1}} = \frac{n_\CI^\beta}{2}.
\end{equation*}
This shows that (II) holds a.s. Hence a decision tree grown using centered splitting is weakly consistent with this choice of $s_n$. As for strong consistency, note that for this choice of $s_n$,
\begin{equation*}
    \exp(-\delta n_\CI \lambda(L(x, \theta)) \leq \exp(-\delta n_\CI^\beta / 2)
\end{equation*}
which is summable by the integral test as
\begin{equation*}
    \int_1^\infty e^{-\delta t^\beta}\mathrm{d}t = \frac{1}{\beta} \int_1^\infty u^{1/\beta - 1} e^{-\delta u}\mathrm{d}u < \infty.
\end{equation*}
We may thus conclude that choosing $s_n = \lceil \log_2 n_\CI^{1 - \beta} \rceil$, a decision tree grown using centered splitting is strongly consistent. \\

Establishing uniform consistency is, as we can tell, not possible through Theorem~\ref{thm:uniformconsistency} without some modification to the choice of splitting variable. Clearly,
\begin{equation*}
    \sup_{x \in [0, 1]^d} \{b_{n, j}^x - a_{n, j}^x\} = 2^{-\inf_{x \in [0, 1]^d} s_{n, j}(x)},
\end{equation*}
and so we see that
\begin{equation*}
    \sup_{x \in [0, 1]^d} \{b_{n, j}^x - a_{n, j}^x\} \to 0 \quad \Leftrightarrow \quad \inf_{x \in [0, 1]^d} s_{n, j}(x) \to \infty
\end{equation*}
where the convergences can be either in probability or almost surely. The following proposition gives a negative result in this direction.

\begin{prop}[Non-diverging uniform number of splits for centered splitting]\label{prop:centerednonuniform}
For a decision tree grown using centered splitting and a feature $j$ with $p_j \in (0, 1/2)$, it holds that
\begin{equation*}
    \lim_{s_n \to \infty}\BP\Big(\inf_{x \in [0, 1]^d}s_{n, j}(x) = 0 \Big) = \frac{1 - 2p_j}{1 - p_j} > 0
\end{equation*}
so that in particular, $\sup_{x \in [0, 1]^d} \{b_{n, j}^x - a_{n, j}^x\}$ does not go to zero in probability.
\end{prop}

This result tells us that even if $s_{n, j}(x) \to \infty$ $\BP$-a.s. for every $x \in [0, 1]^d$, as the number of possible paths to a leaf grows, the number of possibilities for a path to contain no splits on $j$ grows so quickly that the probability of having at least one such path is non-negligible, at least for the usual situation where $p_j < 1/2$. Hence, a modification of the choice of splitting variable in each node is required in order to apply Theorem~\ref{thm:uniformconsistency}. Inspired by the elementary approach from the more general splitting criteria above, we may assume that at least every $N_j$'th split in a given path is on variable $j$. Then
\begin{equation*}
    \inf_{x \in [0, 1]^d}s_{n, j}(x) \geq \frac{s_n}{N_j} \to \infty.
\end{equation*}
Ensuring that the infimum goes to infinity, we can verify the remaining requirements for strong uniform consistency with a proper choice of $s_n$ as the following theorem shows.

\begin{thm}[Strong uniform consistency for modified centered splitting]\label{thm:centereduniform}
Assume that a decision tree is grown using centered splitting with part (ii) of the above algorithm modified to ensure that
\begin{equation*}
    \inf_{x \in [0, 1]^d} s_{n, j}(x) \to \infty \quad \BP\text{-a.s.} 
\end{equation*}
for all $j$. Then if we choose $s_n = \lceil \log_2 n_\CI^{1 - \beta}\rceil$ with $\beta \in (1/2, 1)$, the tree estimator $T(\cdot, \theta)$ is strongly uniformly consistent in $x$.
\end{thm}

\section{Bootstrapping and random forests}\label{sec:bootstrap}

\subsection{Bootstrap assumptions}

Tree-based methods are primarily used for prediction and, in that case, a single tree is seldom used in practice. One instead grows a full forest of some fixed number of trees, $B$, where each tree is grown on a resampled version of the original dataset. Usually, the $B$ datasets are sampled with replacement (bootstrapped) and are of size $n$ as presented in~\cite{BreimanBagging}, an idea that goes back to Efron~\cite{Efron}. In our setup, it turns out that bootstrapped decision trees are still consistent under general bootstrap schemes, as long as the splitting rule still behaves well under the bootstrapped data. 

Let us make it precise what we mean by bootstrapping in our setup. Recall that we work under the honesty condition given in Assumption~\ref{asm:honesty} where data is partitioned into $\CI$ (prediction) and $\CJ$ (growing). If we apply bootstrapping with replacement before partitioning the data, it is possible that $\CI$ and $\CJ$ overlap, interfering with the above consistency results. On the other hand, if we subsample without replacement before partitioning, we clearly still obtain consistency as long as the subsample size grows to infinity with $n$ and as long as parameters varying with $n$ are reframed in terms of the subsample size. We therefore consider an alternative setup, where we first partition the original observations into $\CI$ and $\CJ$ and then bootstrap these disjoint subsets separately. 

To be precise, and now using the weighted version of the bootstrap (see~\cite{VaartWellner}), let $\boldsymbol{W}^\CI = (W_1^\CI, ..., W_{n_\CI}^\CI)$ denote the bootstrap weights associated with $\CI$ and $\boldsymbol{W}^\CJ = (W_1^\CJ, ..., W_{n_\CJ}^\CJ)$ the weights associated with $\CJ$. We use an * to denote quantities formed by bootstrapping. For example,
\begin{equation*}
    L^*(x, \theta) = \prod_{j = 1}^d (a_{n, j}^{x, *}, b_{n, j}^{x, *}]
\end{equation*}
is the leaf where $x$ lies, when the tree is grown on the bootstrapped $\CJ$-sample. 

\begin{asm}[Bootstrap weights]\label{asm:bootstrap}
We require the following:
\begin{enumerate}
    \item[(i)] The weights are non-negative and exchangeable for both $\CI$ and $\CJ$.
    \item[(ii)] The $\CI$-weights and $\CJ$-weights are independent and both sets of weights are independent of the original observations and $\theta$. 
    \item[(iii)] On the $\CI$-weights we assume $\sup_{n \in \BN} \BE[W_1^\CI] < \infty$, that
    \begin{equation*}
        \lim_{n \to \infty} \frac{\BE[W_1^\CI W_2^\CI]}{\BE[W_1^I]^2} = 1
    \end{equation*}
    and that the limit
    \begin{equation*}
        L_{2, 1} := \lim_{n \to \infty} \frac{\BE[(W_1^\CI)^2]}{\BE[W_1^\CI]}
    \end{equation*}
    exists and is finite. 
    \item[(iv)] We assume the following modification of Assumption~\ref{asm:leaf}:
    \begin{enumerate}
        \item[(I*)] $b_{n,j}^{x,*} - a_{n, j}^{x, *} \overset{\BP}{\to} 0$.
        \item[(II*)] $n_\CI \BE[W_1^\CI] \lambda(L^*(x, \theta)) \overset{\BP}{\to} \infty$.
    \end{enumerate}
\end{enumerate}
\end{asm}

We emphasise that we do not assume that the weights are integer-valued. The first limit condition in (iii) states that the weights are asymptotically uncorrelated, which is necessary to ensure that the variance of our estimators tends to zero. The moment conditions in (iii) are quite natural and are satisfied by several standard bootstrap schemes. We now consider some important examples.

\begin{ex}[{Multinomial subsampling}]
If $(W_1^\CI, ..., W_{n_\CI}^\CI)$ are multinomial with $m_n$ trials and probabilities $(1/n_\CI, ..., 1/n_\CI)$, we have
\begin{align*}
    \BE[W_1^\CI] &= \frac{m_n}{n_\CI}, \quad \BE[(W_1^\CI)^2] = \frac{m_n}{n_\CI}\Big(1 - \frac{1}{n_\CI}\Big) + \frac{m_n^2}{n_\CI^2} = \frac{m_n}{n_\CI}\Big(1 - \frac{1}{n_\CI} + \frac{m_n}{n_\CI}\Big) \\
    \BE[W_1^\CI W_2^\CI] &= \frac{m_n(m_n - 1)}{n_\CI^2}
\end{align*}
and so
\begin{equation*}
    \lim_{n \to \infty} \frac{\BE[W_1^\CI W_2^\CI]}{\BE[W_1^\CI]} = \lim_{n \to \infty}\frac{m_n(m_n - 1)/n_\CI^2}{m_n^2/n_\CI^2} = 1 -\lim_{n \to \infty} \frac{1}{m_n} = 1
\end{equation*}
as long as $m_n \to \infty$. We also see that $L_{2, 1} = 1 + \lim_{n \to \infty} \frac{m_n}{n_\CI}$ depending on whether the limit of $m_n/n_\CI$ exists. If $m_n = n_\CI$ (standard Efron bootstrap), $L_{2, 1} = 2$. 
\end{ex}

\begin{ex}[{Subsampling without replacement}]
When subsampling without replacement, the weights $W_i^\CI$ are binary. Let $m_n$ denote the subsample size. Obviously,
\begin{equation*}
    \BE[W_1^\CI] = \BE[(W_1^\CI)^2] = \BP(W_1^\CI = 1) = \frac{m_n}{n_\CI}
\end{equation*}
and
\begin{equation*}
    \BE[W_1^\CI W_2^\CI] = \BP(W_1^\CI = 1 \mid W_2^\CI = 1)\BP(W_2^\CI = 1) = \frac{m_n(m_n - 1)}{n_\CI(n_\CI - 1)}.
\end{equation*}
Thus, $L_{2, 1} = 1$ and
\begin{equation*}
    \frac{\BE[W_1^\CI W_2^\CI]}{\BE[W_1^\CI]^2} = \frac{n_\CI}{n_\CI - 1} \frac{m_n - 1}{m_n} \to 1
\end{equation*}
as long as $m_n \to \infty$.
\end{ex}

\begin{ex}[{Wild bootstrap}]
The wild bootstrap, originally proposed by Wu,~\cite{Wu}, is obtained by choosing $W_1^\CI, ..., W_n^\CI$ iid with finite second moment. In this case, we trivially have that
\begin{equation*}
    \frac{\BE[W_1^\CI W_2^\CI]}{\BE[W_1^\CI]^2} = 1
\end{equation*}
and $L_{2,1} = \BE[(W_1^\CI)^2]/\BE[W_1^\CI]$ since the moments do not vary with $n$.
\end{ex}

Before presenting the results on consistency of a bootstrap decision tree, we introduce some notation. We let
\begin{equation*}
    \widehat{f}^*(x, \theta) := \sum_{i = 1}^{n_\CI} \frac{W_i^\CI \1_{\{X_i \in L^*(x, \theta)\}}}{n_\CI \BE[W_1^\CI] \lambda(L^*(x, \theta))}
\end{equation*}
denote the bootstrap version of $\widehat{f}(x, \theta)$. Similarly, we let
\begin{equation*}
    T^*(x, \theta) := \sum_{i = 1}^{n_\CI} \frac{W_i^\CI \1_{\{X_i \in L^*(x, \theta)\}}}{N_L^*(x, \theta)} Y_i, \quad \text{where} \quad N_L^*(x, \theta) = \sum_{i = 1}^{n_\CI} W_i^\CI \1_{\{X_i \in L^*(x, \theta)\}}
\end{equation*}
and $\widehat{m}^*(x, \theta) = T^*(x, \theta)\widehat{f}^*(x, \theta)$. 

\begin{rmk}
We comment on the above definitions. The expected value $\BE[W_1^\CI]$ controls the subsampling rate. In an $\CI$-sample of size $n_\CI$, the \enquote{expected number of observations} in a bootstrap sample is $n_\CI \BE[W_1^\CI]$. And indeed for multinomial subsampling and subsampling without replacement as presented in the above examples, $n_\CI \BE[W_1^\CI] = m_n$.
\end{rmk}

We now state the main results of this subsection, beginning with weak consistency.

\subsection{Weak consistency}

The following technical result is a bootstrap version of Lemma~\ref{lem:convergenceconditioning}. Let
\begin{equation*}
    \boldsymbol{W}^\CJ_\infty := \bigcup_{n = 1}^\infty \boldsymbol{W}_n^\CJ.
\end{equation*}
\begin{lem}[Bootstrap conditioning]\label{lem:convergenceconditioningbootstrap}
Fix $\delta > 0$. It holds that
\begin{align*}
    \BP(|\widehat{f}^*(x, \theta)& - \BE[\widehat{f}^*(x, \theta) \mid \theta, \CJ, \boldsymbol{W}^\CJ]| > \delta \mid \theta, \CJ, \boldsymbol{W}^\CJ) \\
    &= \BP(|\widehat{f}^*(x, \theta) - \BE[\widehat{f}^*(x, \theta) \mid \theta, \CJ, \boldsymbol{W}^\CJ]| > \delta \mid \theta, \CJ_\infty, \boldsymbol{W}^\CJ_\infty)
\end{align*}
and
\begin{align*}
    \BP(|\widehat{m}^*(x, \theta)& - \BE[\widehat{m}^*(x, \theta) \mid \theta, \CJ, \boldsymbol{W}^\CJ]| > \delta \mid \theta, \CJ, \boldsymbol{W}^\CJ) \\
    &= \BP(|\widehat{m}^*(x, \theta) - \BE[\widehat{m}^*(x, \theta) \mid \theta, \CJ, \boldsymbol{W}^\CJ]| > \delta \mid \theta, \CJ_\infty, \boldsymbol{W}^\CJ_\infty).
\end{align*}
\end{lem}

The main result regarding weak consistency of bootstrap trees is the following. 

\begin{thm}[Consistency of a bootstrap decision tree]\label{thm:bootstraptreeconsistency}
Under Assumptions \ref{asm:covariate}, \ref{asm:honesty} and \ref{asm:bootstrap}, it holds that
\begin{equation*}
    T^*(x, \theta) \overset{\BP}{\to} \BE[Y \mid X = x].
\end{equation*}
\end{thm}

Weak consistency of the random forest with $B$ trees, 
\begin{equation*}
    \text{RF}(x, B) := \frac{1}{B}\sum_{b = 1}^B T_b^*(x, \theta_b)
\end{equation*}
for $T_b^*$ the $b$'th bootstrap tree and iid $\{\theta_b\}$ follows immediately if Assumptions \ref{asm:covariate}, \ref{asm:honesty} and \ref{asm:bootstrap} hold for each tree. Notice that different weights and growing schemes may be considered to each individual tree. As for (iv) of Assumption~\ref{asm:bootstrap}, it is easily seen that for a totally non-adaptive tree (where the splits are placed independently of the data), there is no distinction between (iv) of Assumption~\ref{asm:bootstrap} and Assumption~\ref{asm:leaf}. For an adaptive splitting scheme, the situation is more complicated, but we can suitably adapt the notions and results of the subsection above on general splitting criteria to the bootstrap setup. Write
\begin{equation*}
    N_L^{j, *}(x, \theta) = \sum_{i = 1}^{n_\CI} \1_{\{X_i^j \in (a_{n, j}^{x, *}, b_{n, j}^{x, *}]\}}W_i^\CI
\end{equation*}
for the weighted number of $\CI$-observations in the $j$'th interval of the leaf of $x$. The following result is a bootstrap version of Proposition~\ref{prop:regularI}.

\begin{prop}[Equivalent node size shrinkage condition]\label{prop:regularIbootstrap}
Under Assumptions \ref{asm:covariate},  \ref{asm:honesty}, \ref{asm:bootstrap} (i), (ii) and (iii) as well as $n_\CI \BE[W_1^\CI] \to \infty$, condition (I*) of Assumption~\ref{asm:bootstrap} (iv) holds if and only if
\begin{equation}\label{eq:jproportionbootstrap}
\frac{N_L^{j, *}(x, \theta)}{n_\CI \BE[W_1^\CI]} \overset{\BP}{\to} 0 \quad \text{for all } j = 1, \dots, d.
\end{equation}
\end{prop}



The intuition behind conditions (I*) and (II*) is the same as in the case without bootstrapping. We introduce an analogous minimal node size condition by letting
\begin{equation}\label{eq:minnodesizebootstrap}
    a_n^*(x) = \BP(N_L^*(x, \theta) \geq k_n)
\end{equation}
where $(k_n)$ is again a deterministic sequence. With proper assumptions on $k_n$ and $a_n^*(x)$, we  obtain the following analogue of Proposition~\ref{prop:regularII}. Note how the choice of $k_n$ takes the subsampling rate into account. 

\begin{prop}[Node size growth condition under minimal node size]\label{prop:regularIIbootstrap}
Consider an honest bootstrap tree with $a_n^*(x)$ as defined in (\ref{eq:minnodesizebootstrap}) satisfying $a_n^*(x) \to 1$ and
\begin{equation*}
    \frac{k_n^2}{n_\CI\BE[W_1^\CI]} - 4d\log n_\CI \to \infty.
\end{equation*}
Define the adjusted bootstrap weights by $\widetilde{W}_i^\CI := W_i^\CI/\BE[W_1^\CI]$ and assume
\begin{equation}\label{eq:adjustedbootstrapweakcon}
    \BE[(\widetilde{W}_1^\CI - 1)(\widetilde{W}_2^\CI - 1)] \leq 0 \quad \text{or} \quad \BE[(\widetilde{W}_1^\CI - 1)(\widetilde{W}_2^\CI - 1)] = O(1/n_\CI). 
\end{equation}
Then (II$^*$) of Assumption~\ref{asm:bootstrap} (iv) holds. In particular, it holds for $k_n \sim (n_\CI \BE[W_1^\CI])^\beta$ for $\beta \in (1/2, 1)$.
\end{prop}

The following theorem follows from the two previous propositions.

\begin{thm}[Consistency of a bootstrap tree under minimal node size]\label{thm:weakconsistencybootstrap}
Suppose Assumption~\ref{asm:covariate} holds and let the $\CI$-weights satisfy~\eqref{eq:adjustedbootstrapweakcon},  $n_\CI\BE[W_1^\CI] \to \infty$ and Assumption~\ref{asm:bootstrap} (i), (ii) and (iii). Consider a bootstrap tree satisfying Assumption~\ref{asm:honesty},~\eqref{eq:jproportionbootstrap} and $a_n^*(x) \to 1$ with $a_n^*(x)$ defined in ~\eqref{eq:minnodesizebootstrap} and $k_n \sim (n_\CI \BE[W_1^\CI])^\beta$ for $\beta \in (1/2, 1)$. Then the tree is weakly consistent.
\end{thm}

%
%
%
%
The condition on the mixed moment $\BE[(\widetilde{W}_1^\CI - 1)(\widetilde{W}_2^\CI - 1)] = O(1/n_\CI)$ holds for the previously considered bootstrap schemes. Indeed, for the multinomial subsampling where $(W_1^\CI, ..., W_{n_\CI}^\CI)$ is multinomial with $m_n$ trials and probabilities $(1/n_\CI, ..., 1/n_\CI)$, we have
\begin{equation*}
    \BE[(\widetilde{W}_1^\CI - 1)(\widetilde{W}_2^\CI - 1)] = \frac{\BE[(W_1^\CI - \BE[W_1^\CI])(W_2^\CI - \BE[W_2^\CI])]}{\BE[W_1^\CI]^2} = -\frac{m_n/n_\CI^2}{m_n^2/n_\CI^2} = -\frac{1}{m_n}
\end{equation*}
which is always negative, and hence the condition is trivially satisfied. The same happens for the subsampling without replacement. There we have
%
%
%
\begin{equation*}
    \Cov(W_1^\CI, W_2^\CI) = 
    \frac{m_n}{n_\CI^2(n_\CI - 1)}(m_n - n_\CI) \leq 0
\end{equation*}
and so the requirement is trivially satisfied. For the wild bootstrap, the mixed moment is zero by independence. 

\subsection{Strong consistency}

Proving strong consistency again consists of two parts, namely considering the estimators $\widehat{f}^*(x, \theta)$ and $\widehat{m}^*(x, \theta)$ separately. We follow the exposition without the bootstrap above and present a result for the density estimator first, relying on a finite moment generating function for the bootstrap weights. This is the analogue of Proposition \ref{prop:densityestimatorstrongconvergence} above. Continuing along this route requires a specific interplay between the response and the bootstrap weights which is hard to verify for concrete bootstrap schemes, an exception being subsampling without replacement. If one forgoes the assumptions on the moment generating functions, one needs specific behaviour of the fourth moments along with the requirement that the bootstrap sample size is deterministic. This is however not a strong limitation in practice, since subsampling schemes remain default choices in many implementations. 

\begin{prop}[Strong consistency of a bootstrap decision tree density estimator]\label{prop:densityestimatorstrongconvergencebootstrap}
Let Assumptions \ref{asm:covariate}, \ref{asm:honesty} and \ref{asm:bootstrap} hold, where convergence in probability has been replaced by a.s. convergence. Assume that the $\CI$-weights $\boldsymbol{W}^\CI$ satisfy
\begin{equation*}
    \sup_{n \in \BN} \frac{\kappa_{W_1^\CI}(t) - 1}{\BE[W_1^\CI]} < \infty
\end{equation*}
where $\kappa_{W_1^\CI}$ is the moment generating function of $W_1^\CI$, assumed to exist in some neighbourhood of zero. If 
\begin{equation*}
    \sum_{n_\CI = 1}^\infty \exp(-\delta n_\CI \BE[W_1^\CI] \lambda(L^*(x, \theta))) < \infty,
\end{equation*}
for any $\delta > 0$, then $\widehat{f}^*(x, \theta)$ is strongly consistent,
\begin{equation*}
    \widehat{f}^*(x, \theta) \to f(x) \quad \BP\text{-a.s.}
\end{equation*}
\end{prop}

\begin{thm}[Strong consistency of a bootstrap decision tree]\label{thm:strongconsistencybootstrap}
Let Assumptions \ref{asm:covariate}, \ref{asm:honesty} and \ref{asm:bootstrap} hold, where convergence in probability has been replaced with a.s. convergence. The bootstrap tree is strongly consistent,
\begin{equation*}
    T^*(x, \theta) \to \BE[Y \mid X = x] \quad \BP\text{-a.s.}
\end{equation*}
if one of the following sets of assumptions hold.
\begin{enumerate}
    \item[(i)] Assume $\BE[Y^4], \BE[(W_1^\CI)^4] < \infty$, that $\sum_{i = 1}^{n_\CI} \widetilde{W}_i^\CI = n_\CI$ and
    \begin{equation*}
        \sum_{n_\CI = 1}^\infty \left(\frac{\BE[(\widetilde{W}_1^\CI)^4]}{n_\CI^3 \lambda(L^*(x, \theta))^3} + \frac{\BE[(\widetilde{W}_1^\CI)^2 (\widetilde{W}_2^\CI)^2]}{n_\CI^2 \lambda(L^*(x, \theta))^2} \right) < \infty \quad \BP\text{-a.s.}
    \end{equation*}
    \item[(ii)] Assume that the moment generating functions of $W_1^\CI$ and $W_1^\CI Y_1 \mid X_1 = x$ exist in some neighbourhood around zero (where the latter neighbourhood is independent of $x$) and satisfy the condition in Proposition~\ref{prop:densityestimatorstrongconvergencebootstrap} and that a.s.,
    \begin{equation}\label{eq:mixedmgfcondition}
        \frac{1}{\lambda(L^*(x, \theta))}\int_{L^*(x, \theta)} \frac{\BE[e^{tW_1^\CI Y_1} \mid X_1 = \widetilde{x}] - 1}{\BE[W_1^\CI]}f(\widetilde{x})\mathrm{d}\widetilde{x} = \frac{\BE[e^{tW_1^\CI Y_1} \mid X_1 = x] - 1}{\BE[W_1^\CI]}f(x) + o(1).
    \end{equation}
    and that
    \begin{equation*}
        \sum_{n_\CI = 1}^\infty \exp(-\delta n_\CI \BE[W_1^\CI] \lambda(L^*(x, \theta))) < \infty \quad \BP\text{-a.s.}
    \end{equation*}
    for any $\delta > 0$.
\end{enumerate}
%
%
\end{thm}

Before we consider more tractable assumptions on our trees, let us consider the requirements for $W_1^\CI$ in the different bootstrap schemes.

\begin{ex}[{Subsampling without replacement}]
For subsampling without replacement with $m_n$ trials, we have
\begin{equation*}
    \sum_{i = 1}^{n_\CI} \widetilde{W}_i^\CI = \frac{n_\CI}{m_n} \sum_{i = 1}^{n_\CI} W_i^\CI = n_\CI
\end{equation*}
and
\begin{equation*}
    \BE[(\widetilde{W}_1^\CI)^4] = \frac{n_\CI^3}{m_n^3}, \quad \BE[(\widetilde{W}_1^\CI)^2(\widetilde{W}_2^\CI)^2] = \frac{\BP(W_1^\CI = 1, W_2^\CI = 1)}{m_n^4/n_\CI^4} = O\Big(\frac{n_\CI^2}{m_n^2}\Big).
\end{equation*}
Hence, the summability assumption in (i) of Theorem \ref{thm:strongconsistencybootstrap} becomes
\begin{equation*}
    \sum_{n_\CI = 1}^\infty \left(\frac{1}{m_n^3\lambda(L^*(x, \theta))^3} + \frac{1}{m_n^2\lambda(L^*(x, \theta))^2} \right) < \infty \quad \BP\text{-a.s.}
\end{equation*}
Also,
\begin{equation*}
    \frac{\kappa_{W_1^\CI}(t) - 1}{\BE[W_1^\CI]} = \frac{\frac{m_n}{n_\CI}e^t + 1 - \frac{m_n}{n_\CI} - 1}{\frac{m_n}{n_\CI}} = e^t - 1
\end{equation*}
which is trivially bounded in $n$. Also, assuming that the moment generating function $\kappa_{Y \mid X = x}$ of $Y \mid X = x$ exists and is continuous in $x$, the law of total expectation yields together with independence that
\begin{equation*}
    \frac{\BE[e^{tW_1^\CI Y_1} \mid X_1 = x] - 1}{\BE[W_1^\CI]} = \frac{1 - \frac{m_n}{n_\CI} + \frac{m_n}{n_\CI}\kappa_{Y \mid X = x}(t) - 1}{\frac{m_n}{n_\CI}} = \kappa_{Y \mid X = x}(t) - 1
\end{equation*}
so that~\eqref{eq:mixedmgfcondition} simply holds.
\end{ex}

\begin{ex}[{Subsampling with replacement}]
If $(W_1^\CI, \ldots, W_{n_\CI}^\CI)$ is multinomial with $m_n$ trials, one can easily check that $\sum_{i = 1}^{n_\CI}\widetilde{W}_i^\CI = n_\CI$ and that the summability requirement in (i) of Theorem \ref{thm:strongconsistencybootstrap} is identical to the one for subsampling without replacement. As for the requirements in (ii), a necessary condition is that $\kappa_{Y \mid X = x}(t) < \infty$ for all $t > 0$. Indeed, we have in general that
\begin{equation*}
    \BE[e^{tW_1^\CI Y_1} \mid X_1 = x] = \int_0^\infty \kappa_{Y \mid X = x}(tw)\mathrm{d}F_{W_1^\CI}(w),
\end{equation*}
so if $\kappa_{Y \mid X = x}(t') = \infty$ for some $t' > 0$, then $\kappa_{Y \mid X = x}(tw) = \infty$ for $w > t'/t$. As for the simpler requirement in Proposition~\ref{prop:densityestimatorstrongconvergencebootstrap}, this indeed holds for this bootstrap scheme as the following lemma shows.

\begin{lem}[Multinomial weight MGF condition]\label{lem:multinomialMGF}
If $(W_1^\CI, \ldots, W_{n_\CI}^\CI)$ is multinomial with $m_n$ trials and probabilities $(1/n_\CI, \ldots, 1/n_\CI)$, the requirement on the moment generating function in Proposition~\ref{prop:densityestimatorstrongconvergencebootstrap} holds.
\end{lem}
\end{ex}

\begin{ex}[{Wild bootstrap}]
For the wild bootstrap, it suffices in Proposition~\ref{prop:densityestimatorstrongconvergencebootstrap} to assume that the moment generating function exists in a neighbourhood around zero, since the distribution of $W_1^\CI$ does not vary with $n$. If we also assume that the moment generating function of $W_1^\CI Y_1 \mid X_1 = x$ exists in a neighbourhood around zero independent of $x$ and is continuous in $x$, then~\eqref{eq:mixedmgfcondition} also holds, again because the moment generating function does not depend on $n$. 
\end{ex}

Strong consistency of the random forest
\begin{equation*}
    \text{RF}(x, B) = \frac{1}{B}\sum_{b = 1}^B T^*_b(x, \theta_b)
\end{equation*}
again follows if the conditions in Theorem~\ref{thm:strongconsistencybootstrap} are satisfied for each tree. Notice that different weights and growing schemes can be applied to each tree.\\

We now show strong consistency of trees satisfying general conditions. The proof follows the same line of argument as  Proposition~\ref{prop:regularIIstrong}. The key difference is that the preliminary bounds derived in the proofs of Propositions~\ref{prop:regularIbootstrap} and \ref{prop:regularIIbootstrap} based on the second order Markov inequality are not strong enough to obtain summability. The proof therefore relies on fourth order Markov inequalities instead. 

\begin{lem}[Equivalent bootstrap shrinking condition]\label{lem:regularstrongconsistencyI}
Assume that the $\CI$-weights satisfy
\begin{equation}\label{eq:regularstrongconsistency}
    \sum_{n_\CI = 1}^\infty \Big( \frac{\BE[(\widetilde{W}_1^\CI)^4]}{n_\CI^3} + \frac{\BE[(\widetilde{W}_1^\CI)^2]^2}{n_\CI^2} \Big) < \infty
\end{equation}
where $\widetilde{W}_i^\CI = W_i^\CI/\BE[W_1^\CI]$. Then (I$^*$) of Assumption~\ref{asm:bootstrap} holds almost surely if and only if
\begin{equation}\label{eq:jproportionbootstrapstrong}
  \frac{N_L^{j, *}(x, \theta)}{n_\CI \BE[W_1^\CI]} \to 0 \quad \BP\text{-a.s. \ for all } j = 1, \dots, d.
\end{equation}
\end{lem}

\begin{lem}[Sufficient bootstrap growing condition]\label{lem:regularstrongconsistencyII}
For a bootstrap honest tree with $\sum_{n = 1}^\infty (1 - a_n^*(x)) < \infty$ for $a_n^*(x)$ defined in ~\eqref{eq:minnodesizebootstrap} and
\begin{align*}
    &\sum_{n_\CI = 1}^\infty \frac{n_\CI^{4d}}{e^{k_n^2/(8n_\CI \BE[W_1^\CI])}} < \infty, \quad \sum_{n_\CI = 1}^\infty \frac{\BE[(W_1^\CI - \BE[W_1^\CI])^4]}{k_n^4} n_\CI < \infty, \\ &\sum_{n_\CI = 1}^\infty \frac{\BE[(W_1^\CI - \BE[W_1^\CI])^2]^2}{k_n^4}n_\CI^2 < \infty,
\end{align*}
we have that (II$^*$) of Assumption~\ref{asm:bootstrap} holds almost surely.
\end{lem}

\begin{ex}[{Multinomial subsampling}]
By using the results of~\cite{Ouimet}, it is readily verified that for this bootstrap scheme,
\begin{align*}
    &\frac{\BE[(\widetilde{W}_1^\CI)^4]}{n_\CI^3} = O\Big(\frac{1}{m_n^3}\Big), \quad \frac{\BE[(\widetilde{W}_1^\CI)^2]^2}{n_\CI^2} = O\Big(\frac{1}{m_n^2}\Big), \\
    &\frac{\BE[(W_1^\CI - \BE[W_1^\CI])^4]n_\CI}{k_n^4} = O\Big(\frac{m_n}{k_n^4}\Big), \quad \frac{\BE[(W_1^\CI - \BE[W_1^\CI])^2]^2 n_\CI^2}{k_n^4} = O\Big(\frac{m_n^2}{k_n^4}\Big).
\end{align*}
Hence we see that choosing $m_n = n_\CI^\gamma$ for $\gamma \in (1/2, 1)$, the summability condition in Lemma~\ref{lem:regularstrongconsistencyI} is satisfied. Also, the summability conditions in Lemma~\ref{lem:regularstrongconsistencyII} hold when $k_n \sim  n_\CI^\beta $ for $\beta \in (1/2, 1)$ and $\gamma < 2\beta - 1/2$. For example, choosing $\gamma = 3/5$ and $\beta = 2/3$ would suffice for all requirements in these two lemmata to be satisfied.
\end{ex}

\begin{ex}[{Subsampling without replacement}]
For this bootstrap scheme, one can verify that all the $O$-terms become the same as for multinomial subsampling and thus the same choices of $k_n$ and $m_n$ work.
\end{ex}

\begin{ex}[{Wild bootstrap}]
For the wild bootstrap, all moments are constant. So here~\eqref{eq:regularstrongconsistency} always holds. And it is easily verified that choosing $k_n \sim n_\CI^\beta$ for $\beta \in (3/4, 1)$ suffices to make all summability requirements hold. 
\end{ex}

\begin{thm}[Strong consistency of a bootstrap tree under minimal node size]\label{thm:regularstrongconsistency}
Consider a bootstrap honest tree satisfying $\sum_{n = 1}^\infty (1 - a_n^*(x)) < \infty$ with $a_n^*(x)$ defined in  ~\eqref{eq:minnodesizebootstrap},~\eqref{eq:jproportionbootstrapstrong} and all the requirements on $k_n$ and the $W_i^\CI$ from Assumption~\ref{asm:bootstrap} (i), (ii), (iii), Lemma~\ref{lem:regularstrongconsistencyI} and \ref{lem:regularstrongconsistencyII}. If one of the following sets of conditions hold, the bootstrap tree is strongly consistent.
\begin{enumerate}
    \item[(i)] We have $\BE[Y^4], \BE[(W_1^\CI)^4] < \infty$, $\sum_{i = 1}^{n_\CI} \widetilde{W}_i^\CI = n_\CI$ and
    \begin{equation*}
        \sum_{n_\CI = 1}^\infty \left(\frac{\BE[(W_1^\CI)^4]}{k_n^3 \BE[W_1^\CI]} + \frac{\BE[(W_1^\CI)^2 (W_2^\CI)^2]}{k_n^2 \BE[W_1^\CI]^2} \right) < \infty.
    \end{equation*}
    \item[(ii)] The moment generating functions of $W_1^\CI$ and $W_1^\CI Y_1 \mid X_1 = x$ satisfy the requirements in Theorem~\ref{thm:strongconsistencybootstrap} and
    \begin{equation*}
        \sum_{n_\CI = 1}^\infty e^{-\delta k_n} < \infty
    \end{equation*}
    for all $\delta > 0$.
\end{enumerate}
\end{thm}

To summarise, if we have subsampling with or without replacement, all the requirements in this proposition for the splitting scheme hold for $k_n \sim  n_\CI^\beta $ and $m_n =  n_\CI^\gamma $ where $\beta, \gamma \in (1/2, 1)$ and $\gamma < 2\beta - 1/2$. In fact, for both subsampling with or without replacement, one can verify that $\BE[(W_1^\CI)^4]/\BE[W_1^\CI] = O(1)$ and $\BE[(W_1^\CI)^2 (W_2^\CI)^2]/\BE[W_1^\CI]^2 = O(1)$. \\

Concerning $L^p$-consistency for the bootstrap tree, the proof of Theorem~\ref{thm:L1consistency} cannot be transferred to the bootstrap setup. One reason is that the weights for $\CI$ are not necessarily integer-valued, causing the conditioning argument to break down. If one could verify that each tree is $L^p$ consistent for some $p \geq 1$, convexity of the function $x \mapsto |x|^p$ guarantees $L^p$ consistency of $\text{RF}(x, B)$ as well. Investigations into $L^p$- and strong uniform consistency for honest bootstrap trees as presented here are possible avenues for further study.

\section{Proofs}\label{sec:proofs}

\subsection{Proofs for Section \ref{sec:singletree}}

\begin{proof}[Proof of Lemma~\ref{lem:convergenceconditioning}]
Consider the sub-$\sigma$-algebras $\CD := \sigma(\theta, (X_1, Y_1), \ldots, (X_n, Y_n))$, $\CG := \sigma(\theta, \CJ)$ and $\CH := \sigma(\CJ_\infty \setminus \CJ)$. By definition, $\CG \subseteq \CD$ and $\CD \indep \CH$. The random variables
\begin{equation*}
    \1_{\{|\widehat{f}(x, \theta) - \BE[\widehat{f}(x, \theta) \mid \theta, \CJ]| > \delta\}} \quad \text{and} \quad \1_{\{|\widehat{m}(x, \theta) - \BE[\widehat{m}(x, \theta) \mid \theta, \CJ]| > \delta\}}
\end{equation*}
are both $\CD$-measurable, and since $\sigma(\theta, \CJ_\infty) = \CG \lor \CH$, the assertion of the lemma follows immediately by Lemma~\ref{lem:condmean}. 
\end{proof}

\begin{proof}[Proof of Proposition~\ref{prop:densityestimatorconsistent}]
We note that
\begin{align*}
    \BE[\widehat{f}(x, \theta) \mid \theta, \CJ] &= \frac{1}{n_\CI \lambda(L(x, \theta))} \sum_{i = 1}^{n_\CI}\BE\Big[\1_{\{X_i \in L(x, \theta) \}} \mid \theta, \CJ \Big] \\
    &= \frac{1}{\lambda(L(x, \theta))}\int_{L(x, \theta)} f(\widetilde{x})\mathrm{d}\widetilde{x},
\end{align*}
and this converges in probability to $f(x)$ by condition (I)\footnote{Due to continuity, every $x \in [0, 1]^d$ is a Lebesgue point of $f$. We refer to~\cite{Rudin} for background.}. Hence we are done once we show that $\widehat{f}(x, \theta) - \BE[\widehat{f}(x, \theta) \mid \theta, \CJ] \overset{\BP}{\to} 0$. By Markov's inequality, for any $\delta > 0$,
\begin{equation*}
    \BP(|\widehat{f}(x, \theta) - \BE[\widehat{f}(x, \theta) \mid \theta, \CJ]| > \delta \mid \theta, \CJ) \leq \frac{\Var[\widehat{f}(x, \theta) \mid \theta, \CJ]}{\delta^2}
\end{equation*}
and
\begin{equation*}
    \Var[\widehat{f}(x, \theta) \mid \theta, \CJ] = \BE[\widehat{f}(x, \theta)^2 \mid \theta, \CJ] - \Big(\frac{1}{\lambda(L(x, \theta))}\int_{L(x, \theta)} f(\widetilde{x})\mathrm{d}\widetilde{x} \Big)^2,
\end{equation*}
so we need to show that $\BE[\widehat{f}(x, \theta)^2 \mid \theta, \CJ] \overset{\BP}{\to} f(x)^2$. We compute
\begin{align*}
    &\BE[\widehat{f}(x, \theta)^2 \mid \theta, \CJ]\\
    &= \frac{1}{n_\CI^2 \lambda(L(x, \theta))^2} \sum_{i = 1}^{n_\CI}\sum_{j = 1}^{n_\CI} \BE[\1_{\{X_i \in L(x, \theta)\}}\1_{\{X_j \in L(x, \theta)\}} \mid \theta, \CJ] \\
    &= \frac{1}{n_\CI^2 \lambda(L(x, \theta))^2}(n_\CI \BP(X_1 \in L(x, \theta) \mid \theta, \CJ) + n_\CI(n_\CI - 1) \BP(X_1 \in L(x, \theta) \mid \theta, \CJ)^2) \\
    &= \frac{1}{n_\CI \lambda(L(x, \theta))} \frac{\int_{L(x, \theta)} f(\widetilde{x})\mathrm{d}\widetilde{x}}{\lambda(L(x, \theta))} + \frac{n_\CI - 1}{n_\CI}\Big(\frac{1}{\lambda(L(x, \theta))} \int_{L(x, \theta)}f(\widetilde{x})\mathrm{d}\widetilde{x} \Big)^2 \\
    &\overset{\BP}{\to} 0 \cdot f(x) + 1 \cdot f(x)^2 = f(x)^2
\end{align*}
which shows that
\begin{equation*}
    \BP(|\widehat{f}(x, \theta) - \BE[\widehat{f}(x, \theta) \mid \theta, \CJ]| > \delta \mid \theta, \CJ) \overset{\BP}{\to} 0.
\end{equation*}
By Lemma~\ref{lem:convergenceconditioning}, we can replace $\theta, \CJ$ in the probability by $\theta, \CJ_\infty$. Take expectations and apply the tower property along with dominated convergence to conclude the desired statement. 
\end{proof}

\begin{proof}[Proof of Theorem~\ref{thm:conspart1}]
Just like in the proof of Proposition~\ref{prop:densityestimatorconsistent}, we have
\begin{equation*}
    \BE[\widehat{m}(x, \theta) \mid \theta, \CJ] = \frac{1}{\lambda(L(x, \theta))}\int_{L(x, \theta)} m(\widetilde{x})\mathrm{d}\widetilde{x} \overset{\BP}{\to} m(x).
\end{equation*}
To see that $\widehat{m}(x, \theta) - \BE[\widehat{m}(x, \theta) \mid \theta, \CJ] \overset{\BP}{\to} 0$, we again apply a conditional Chebyshev inequality. We compute
\begin{align*}
    \BE[\widehat{m}(x, \theta)^2 \mid \theta, \CJ] &= \frac{1}{n_\CI^2 \lambda(L(x, \theta))^2} \sum_{i = 1}^{n_\CI} \sum_{j = 1}^{n_\CI} \BE[\1_{\{X_i \in L(x, \theta)\}}\1_{\{X_j \in L(x, \theta)\}} Y_i Y_j \mid \theta, \CJ] \\
    &= \frac{1}{n_\CI \lambda(L(x, \theta))} \frac{\int_{L(x, \theta)} \BE[Y^2 \mid X = \widetilde{x}]f(\widetilde{x})\mathrm{d}\widetilde{x}}{\lambda(L(x, \theta))} \\
    &\quad+ \frac{n_\CI - 1}{n_\CI} \Big(\frac{1}{\lambda(L(x, \theta))}\int_{L(x, \theta)}m(\widetilde{x})\mathrm{d}\widetilde{x} \Big)^2
\end{align*}
which converges to $m(x)^2$ in probability, and so the conditional variance again converges to zero in probability. The rest of the argument is identical to the one for the density estimator. We have thus shown that $\widehat{f}(x, \theta) \overset{\BP}{\to} f(x)$ and $\widehat{m}(x, \theta) \overset{\BP}{\to} m(x)$. Now apply the continuous mapping theorem to complete the proof. 
\end{proof}

\begin{proof}[Proof of Theorem~\ref{thm:L1consistency}]
By letting, say, $p = \delta = 1$, the first assertion of the theorem follows from the second. To prove the second assertion, we apply the Lebesgue--Vitali Theorem (see Theorem 4.5.4 in~\cite{Bogachev}) which states that $L^p$-convergence of $T(x, \theta)$ is equivalent to convergence in probability and uniform integrability of the sequence $\{|T(x, \theta)|^p\}$. Convergence in probability was already established in Theorem~\ref{thm:conspart1}. To show uniform integrability, it suffices to show that 
\begin{equation*}
    \sup_n \{\BE[|T(x, \theta)|^{p + \delta}]\} < \infty.
\end{equation*}
We consider the conditional moment $\BE[|T(x, \theta)|^{p + \delta} \mid N_L(x, \theta) = m]$ for some $m = 1, ..., n_\CI$. If we know that exactly $m$ of the $X_i$ fall into $L(x, \theta)$, we can use exchangeability of the $(X_i, Y_i)$ to conclude that $\BE[|T(x, \theta)|^{p + \delta} \mid N_L(x, \theta) = m]$ equals
\begin{equation*}
    \BE\Big[\Big|\sum_{i = 1}^{n_\CI} \frac{\1_{\{X_i \in L(x, \theta)\}}Y_i}{m} \Big|^{p + \delta} \mid X_1, ..., X_m \in L(x, \theta), X_{m + 1}, ..., X_{n_\CI} \notin L(x, \theta) \Big]
\end{equation*}
which simplifies to
\begin{equation*}
    \BE\Big[\Big|\sum_{i = 1}^{m} \frac{{Y_i}}{m} \Big|^{p + \delta} \mid X_1, ..., X_m \in L(x, \theta) \Big].
\end{equation*}
We know that the function $x \mapsto |x|^{p + \delta}$ is convex since $p + \delta > 1$ by assumption. Thus,
\begin{equation*}
    \Big|\sum_{i = 1}^{m} \frac{{Y_i}}{m} \Big|^{p + \delta} \leq \frac{1}{m} \sum_{i = 1}^m |Y_i|^{p + \delta},
\end{equation*}
whence by the iid assumption on the observations,
\begin{align*}
    \BE[|T(x, \theta)|^{p + \delta} \mid N_L(x, \theta) = m] &\leq \frac{1}{m}\sum_{i = 1}^m \BE[|Y_i|^{p + \delta} \mid X_i \in L(x, \theta)] \\
    &= \BE[|Y_1|^{p + \delta} \mid X_1 \in L(x, \theta)].
\end{align*}
Using the tower property, we can bound this as follows:
\begin{align*}
    \BE[|Y_1|^{p + \delta} \mid X_1 \in L(x, \theta)] &= \frac{\BE[|Y_1|^{p + \delta} \1_{\{X_1 \in L(x, \theta)\}}]}{\BP(X_1 \in L(x, \theta))}\\
    &= \frac{\BE[\1_{\{X_1 \in L(x, \theta)\}} \BE[|Y_1|^{p + \delta} \mid X_1, \theta, \CJ]]}{\BP(X_1 \in L(x, \theta))} \\
    &= \frac{\BE[\1_{\{X_1 \in L(x, \theta)\}} \BE[|Y_1|^{p + \delta} \mid X_1]]}{\BP(X_1 \in L(x, \theta))} \leq K.
\end{align*}
In conclusion, we have shown that $\BE[|T(x, \theta)|^{p + \delta} \mid N_L(x, \theta)] \leq K$ which implies that $\BE[|T(x, \theta)|^{p + \delta}] \leq K$ again by the tower property, and the proof is complete.
\end{proof}

\begin{proof}[Proof of Proposition~\ref{prop:densityestimatorstrongconvergence}]
Make the decomposition
\begin{equation*}
    \widehat{f}(x, \theta) - f(x) = \widehat{f}(x, \theta) - \BE[\widehat{f}(x, \theta) \mid \theta, \CJ] + \BE[\widehat{f}(x, \theta) \mid \theta, \CJ] - f(x)
\end{equation*}
and note that
\begin{equation*}
    \BE[\widehat{f}(x, \theta) \mid \theta, \CJ] = \frac{1}{\lambda(L(x, \theta))} \int_{L(x, \theta)} f(\widetilde{x})\mathrm{d}\widetilde{x} \to f(x) \quad \BP\text{-a.s.}
\end{equation*}
by the almost sure strengthening of Assumption~\ref{asm:leaf}. Hence we only need to show that $\widehat{f}(x, \theta) - \BE[\widehat{f}(x, \theta) \mid \theta, \CJ] \to 0$ $\BP\text{-a.s.}$ Note that
\begin{equation*}
    \widehat{f}(x, \theta) - \BE[\widehat{f}(x, \theta) \mid \theta, \CJ] = \frac{1}{n_\CI \lambda(L(x, \theta))}\sum_{i = 1}^{n_\CI} (\1_{\{X_i \in L(x, \theta)\}} - p(\theta, \CJ)),
\end{equation*}
where $p(\theta, \CJ) := \BP(X_1 \in L(x, \theta) \mid \theta, \CJ)$, and that the random variables in the sum have mean zero conditional on $\theta, \CJ$ and are bounded in absolute value by 2. Hence by Bernstein's inequality, for any $\delta > 0$,
\begin{align*}
    &\BP\Big(\sum_{i = 1}^{n_\CI}(\1_{\{X_i \in L(x, \theta)\}} - p(\theta, \CJ)) > \delta \mid \theta, \CJ \Big)\\
    &\leq \exp\Big(-\frac{\delta^2}{2(n_\CI p(\theta, \CJ))(1 - p(\theta, \CJ) + 2\delta/3)} \Big) \\
    &\leq \exp\Big(-\frac{\delta^2}{2Cn_\CI \lambda(L(x, \theta)) + 4\delta/3} \Big)
\end{align*}
so that
\begin{equation*}
    \BP(|\widehat{f}(x, \theta) - \BE[\widehat{f}(x, \theta) \mid \theta, \CJ]| > \delta \mid \theta, \CJ) \leq 2\exp\Big(-\frac{\delta^2 n_\CI \lambda(L(x, \theta))}{2C + 4\delta/3} \Big).
\end{equation*}
By Lemma~\ref{lem:convergenceconditioning}, we may replace the conditioning in the probability on the left hand side by $\theta, \CJ_\infty$ which is independent of $n$. It now follows immediately from Corollary~\ref{cor:conditionalBC} that $\widehat{f}(x, \theta) - \BE[\widehat{f}(x, \theta) \mid \theta, \CJ] \to 0$ $\BP$-a.s. and the proof is complete.
\end{proof}

\begin{proof}[Proof of Theorem~\ref{thm:strongconsistency}]
Start by making the same decomposition as for $\widehat{f}(x, \theta)$,
\begin{equation*}
    \widehat{m}(x, \theta) - m(x) = \widehat{m}(x, \theta) - \BE[\widehat{m}(x, \theta) \mid \theta, \CJ] + \BE[\widehat{m}(x, \theta) \mid \theta, \CJ] - m(x).
\end{equation*}
Again,
\begin{equation*}
    \BE[\widehat{m}(x, \theta) \mid \theta, \CJ] = \frac{1}{\lambda(L(x, \theta))} \int_{L(x, \theta)} m(\widetilde{x})\mathrm{d}\widetilde{x} \to m(x) \quad \BP\text{-a.s.},
\end{equation*}
so we consider
\begin{equation*}
    \widehat{m}(x, \theta) - \BE[\widehat{m}(x, \theta) \mid \theta, \CJ] = \frac{1}{n_\CI \lambda(L(x, \theta))}\sum_{i = 1}^{n_\CI} (\1_{\{X_i \in L(x, \theta)\}}Y_i - \BE[\1_{\{X_1 \in L(x, \theta)\}}Y_1 \mid \theta, \CJ]).
\end{equation*}
This is where the proofs of the two cases differ. First consider the set of assumptions in (i). We compute the fourth moment of $\widehat{m}(x, \theta) - \BE[\widehat{m}(x, \theta) \mid \theta, \CJ]$ as follows. Disregarding the division with $n_\CI \lambda(L(x, \theta))$ in front, we have a sum of iid terms with mean zero and hence
\begin{align*}
    &\BE\Big[\Big(\sum_{i = 1}^{n_\CI} (\1_{\{X_i \in L(x, \theta)\}}Y_i - \BE[\1_{\{X_1 \in L(x, \theta)\}}Y_1 \mid \theta, \CJ]) \Big)^4 \mid \theta, \CJ \Big] = \\
    &n_\CI \BE[(\1_{\{X_1 \in L(x, \theta)\}}Y_1 - \BE[\1_{\{X_1 \in L(x, \theta)\}}Y_1])^4 \mid \theta, \CJ] + 3n_\CI(n_\CI - 1)\Var(\1_{\{X_1 \in L(x, \theta)\}}Y_1 \mid \theta, \CJ)^2.
\end{align*}
Expanding the conditional central fourth moment of $\1_{\{X_1 \in L(x, \theta)\}} Y_1$ and dividing by $\lambda(L(x, \theta))^4$, we see that
\begin{equation*}
    \frac{\BE[(\1_{\{X_1 \in L(x, \theta)\}}Y_1 - \BE[\1_{\{X_1 \in L(x, \theta)\}}Y_1])^4 \mid \theta, \CJ]}{\lambda(L(x, \theta))^4} = \frac{1}{\lambda(L(x, \theta))^3}(\BE[Y^4 \mid X = x] + o_\BP(1)).
\end{equation*}
Similarly,
\begin{equation*}
    \frac{\Var(\1_{\{X_1 \in L(x, \theta)\}}Y_1 \mid \theta, \CJ)^2}{\lambda(L(x, \theta))^4} = \frac{1}{\lambda(L(x, \theta))^2}(\BE[Y^2 \mid X = x]^2 + o_\BP(1)).
\end{equation*}
All in all,
\begin{align*}
    \BE[(\widehat{m}(x, \theta) - \BE[\widehat{m}(x, \theta) \mid \theta, \CJ])^4 \mid \theta, \CJ] &= \frac{\BE[Y^4 \mid X = x] + o_\BP(1)}{n_\CI^3 \lambda(L(x, \theta))^3} + 3\frac{n_\CI - 1}{n_\CI} \frac{\BE[Y^2 \mid X = x] + o_\BP(1)}{n_\CI^2 \lambda(L(x, \theta))^2}
\end{align*}
which is summable by assumption. Replacing $\theta, \CJ$ by $\theta, \CJ_\infty$ as per Lemma \ref{lem:convergenceconditioning}, we have $\widehat{m}(x, \theta) - \BE[\widehat{m}(x, \theta) \mid \theta, \CJ] \to 0$ $\BP$-a.s. Letting $Y \equiv 1$ yields almost sure convergence of $\widehat{f}(x, \theta)$ and the continuous mapping theorem completes the proof in the (i) case. \\

Now consider the set of assumptions in (ii). Under these conditions, we already know from  Proposition~\ref{prop:densityestimatorstrongconvergence} that $\widehat{f}(x, \theta) \to f(x)$ a.s. It remains to show that $\widehat{m}(x, \theta) \to m(x)$ a.s. Let $\Lambda_{n_\CI}(\cdot \mid \theta, \CJ)$ denote the cumulant generating function of the sum conditional on $\theta, \CJ$. Then for $t \in (-c, c)$,
\begin{align*}
    &\Lambda_{n_\CI}(t \mid \theta, \CJ)\\
    &:= \log\BE\Big[\exp\Big(t\sum_{i = 1}^{n_\CI} (\1_{\{X_i \in L(x, \theta)\}}Y_i - \BE[\1_{\{X_1 \in L(x, \theta)\}}Y_1 \mid \theta, \CJ]) \Big) \mid \theta, \CJ \Big] \\
    &= n_\CI \log \BE\Big[\exp\Big(t(\1_{\{X_1 \in L(x, \theta)\}}Y_1 - \BE[(\1_{\{X_1 \in L(x, \theta)\}}Y_1 \mid \theta, \CJ])\Big) \mid \theta, \CJ \Big] \\
    &= -tn_\CI \BE[(\1_{\{X_1 \in L(x, \theta)\}}Y_1 \mid \theta, \CJ] + n_\CI\log \BE\Big[\exp\Big(t \1_{\{X_1 \in L(x, \theta)\}}Y_1 \Big) \mid \theta, \CJ\Big]
\end{align*}
The expectation in the second term equals
\begin{align*}
    &\BE\Big[\exp\Big(t \1_{\{X_1 \in L(x, \theta)\}}Y_1 \Big) \mid \theta, \CJ\Big]\\
    &= 1 - p(\theta, \CJ) + \BE\Big[\1_{\{X_1 \in L(x, \theta)\}} e^{tY_1} \mid \theta, \CJ\Big] \\
    &= 1 - p(\theta, \CJ) + \int_{L(x, \theta)} \kappa_{Y \mid X = \widetilde{x}}(t)f(\widetilde{x})\mathrm{d}\widetilde{x} \\
    &= 1 + \lambda(L(x, \theta))f(x)(\kappa_{Y \mid X = x}(t) - 1 + o(1)) \quad \BP\text{-a.s.}
\end{align*}
Recalling that $\BE[\1_{\{X_1 \in L(x, \theta)\}}Y_1 \mid \theta, \CJ] = \lambda(L(x, \theta))(m(x) + o(1))$ a.s. and using the Taylor expansion $\log(1 + s) = s + o(s)$ in the limit $s \to 0$, we have a.s. that
\begin{align*}
    \Lambda_{n_\CI}(t \mid \theta, \CJ) = n_\CI\lambda(L(x, \theta))(-tm(x) + f(x)(\kappa_{Y \mid X = x}(t) - 1) + o(1)).
\end{align*}
We can now apply a two-sided Chernoff bound to obtain
\begin{align*}
    &\BP(|\widehat{m}(x, \theta) - \BE[\widehat{m}(x, \theta) \mid \theta, \CJ]| > \delta \mid \theta, \CJ)\\
    &\leq e^{-\delta tn_\CI \lambda(L(x, \theta))}\Big(e^{\Lambda_{n_\CI}(t \mid \theta, \CJ)} + e^{\Lambda_{n_\CI}(-t \mid \theta, \CJ)} \Big). 
\end{align*}
This leaves us with a sum of two exponentials. The exponents in the two terms without the $o(1)$ terms are
\begin{align*}
    E_1(t) &:= n_\CI \lambda(L(x, \theta))(-\delta t - tm(x) + f(x)(\kappa_{Y \mid X = x}(t) - 1)), \\
    E_2(t) &:= n_\CI \lambda(L(x, \theta))(-\delta t + tm(x) + f(x)(\kappa_{Y \mid X = x}(-t) - 1)).
\end{align*}
We are done once we show that $t$ can be chosen in a way that makes both exponents negative. Note that $E_1(0) = E_2(0) = 0$ and that $E_1'(0) = E_2'(0) = -\delta < 0$ using properties of moment generating functions. Thus there must exist a $t \in (0, c)$ such that $E_1(t) < 0$ and $E_2(t) < 0$. This shows that for every $\omega$ in a set with probability one, there exists a $\delta'(\omega)$ such that for large enough $n_\CI$,
\begin{equation*}
    \BP(|\widehat{m}(x, \theta) - \BE[\widehat{m}(x, \theta) \mid \theta, \CJ]| > \delta \mid \theta(\omega), \CJ(\omega)) \leq 2\exp\Big(-\delta'(\omega) n_\CI \lambda(L(x, \theta)) \Big).
\end{equation*}
Now replace the conditioning in the probability with $\theta, \CJ_\infty$, as is allowed by Lemma~\ref{lem:convergenceconditioning}, and apply Corollary~\ref{cor:conditionalBC} as earlier to conclude the proof of $\widehat{m}(x, \theta) \to m(x)$, $\BP$-a.s. Finally, apply the continuous mapping theorem to conclude $T(x, \theta) \to \BE[Y \mid X = x]$ $\BP$-a.s.
\end{proof}

\begin{proof}[Proof of Theorem~\ref{thm:uniformconsistency}]
Both weak and strong uniform consistency are proved in the same way. The only difference is the application of the conditional Borel--Cantelli criterion for strong convergence, Corollary~\ref{cor:conditionalBC} for the case of strong uniform consistency. Also, showing uniform consistency of $\widehat{f}(x, \theta)$ is a special case of the proof for $\widehat{m}(x, \theta)$. Simply let $Y \equiv 1$. Thus we only provide the proof for $\widehat{m}(x, \theta)$. To show uniform consistency for the component $\widehat{m}(x, \theta)$, we make the usual decomposition
\begin{equation*}
    \widehat{m}(x, \theta) - m(x) = \widehat{m}(x, \theta) - \BE[\widehat{m}(x, \theta) \mid \theta, \CJ] + \BE[\widehat{m}(x, \theta) \mid \theta, \CJ] - m(x).
\end{equation*}
To start with, consider the second term. Using Lipschitz continuity of $m$, there exists a constant $C_m$ such that
\begin{align*}
    |\BE[\widehat{m}(x, \theta) \mid \theta, \CJ] - m(x)| &\leq \frac{1}{\lambda(L(x, \theta))} \int_{L(x, \theta)} |m(\widetilde{x}) - m(x)|\mathrm{d}\widetilde{x} \\ &\leq \frac{C_m}{\lambda(L(x, \theta))} \int_{L(x, \theta)} \| \widetilde{x} - x\| \mathrm{d}\widetilde{x} \\ 
    &\leq C_m \sqrt{d}\max_{j = 1, ..., d}\{b_{n, j}^x - a_{n, j}^x\}.
\end{align*}
Hence
\begin{equation*}
    \sup_{x \in [0, 1]^d} |\BE[\widehat{m}(x, \theta) \mid \theta, \CJ] - m(x)| \leq C_m \sqrt{d} \max_{j = 1, ..., d} \sup_{x \in [0, 1]^d} \{b_{n, j}^x - a_{n, j}^x\},
\end{equation*}
and this converges to zero in the desired way by assumption. It remains to consider the first term, for which it holds that
\begin{align*}
    &\widehat{m}(x, \theta) - \BE[\widehat{m}(x, \theta) \mid \theta, \CJ]\\
    &= \frac{1}{n_\CI \lambda(L(x, \theta))} \sum_{i = 1}^{n_\CI}( \1_{\{X_i \in L(x, \theta)\}}Y_i - \BE[\1_{\{X_1 \in L(x, \theta)\}}Y_1 \mid \theta, \CJ]).
\end{align*}
Thus, for any $\delta > 0$,
\begin{align*}
    &\BP\Big(\sup_{x \in [0, 1]^d}|\widehat{m}(x, \theta) - \BE[\widehat{m}(x, \theta) \mid \theta, \CJ] > \delta \mid \theta, \CJ\Big) \leq \\
    &\BP\Big(\sup_{x \in [0, 1]^d}\Big|\frac{1}{n_\CI}\sum_{i = 1}^{n_\CI}( \1_{\{X_i \in L(x, \theta)\}}Y_i - \BE[\1_{\{X_1 \in L(x, \theta)\}}Y_1 \mid \theta, \CJ]) \Big| > \delta \underline{v}_n \mid \theta, \CJ \Big),
\end{align*}
where we recall that $\underline{v}_n := \inf_{x \in [0, 1]^d}\lambda(L(x, \theta))$. A decision tree creates a partition $\mathcal{P}$ of the feature space $[0, 1]^d$ which is contained in $\mathcal{A}$, the set of all half-open rectangles of $[0, 1]^d$, 
\begin{equation*}
    \mathcal{A} = \{(a_1, b_1] \times \cdots (a_d, b_d] : a_i, b_i \in [0, 1], a_i < b_i\}.
\end{equation*}
From this observation, it follows that
\begin{align*}
    &\BP\Big(\sup_{x \in [0, 1]^d}\Big|\frac{1}{n_\CI}\sum_{i = 1}^{n_\CI}( \1_{\{X_i \in L(x, \theta)\}}Y_i - \BE[\1_{\{X_1 \in L(x, \theta)\}}Y_1 \mid \theta, \CJ]) \Big| > \delta \underline{v}_n \mid \theta, \CJ \Big) \leq \\
    &\quad\BP\Big(\sup_{A \in \mathcal{A}}\Big|\frac{1}{n_\CI}\sum_{i = 1}^{n_\CI} (\1_{\{X_i \in A\}}Y_i - \BE[\1_{\{X_1 \in A\}}Y_1]) \Big| > \delta \underline{v}_n \mid \theta, \CJ \Big).
\end{align*}
The rest of the proof consists in bounding this probability. Assume that $Y \in [0, 1]$ a.s. and consider the function class
\begin{equation*}
    \CF := \{f_A:[0, 1]^d \times [0, 1] \to [0, 1], f_A(x, y) = \1_A(x)y : A \in \mathcal{A}\}.
\end{equation*}
By considering the countable subclass $\widetilde{\CF}$ obtained by restricting the endpoints of the rectangles in $\mathcal{A}$ to rational numbers, it is easy to see that $\CF$ is pointwise separable in the sense of~\cite{VaartWellner}. Also, $\CF$ is a VC-class (see again~\cite{VaartWellner} for background). Indeed, for a function $f_A \in \CF$, the subgraph equals
\begin{align*}
    &\{(x, y, t) \in [0, 1]^d \times [0, 1] \times \BR : t < \1_A(x)y\}\\
    &= [0, 1]^d \times [0, 1] \times (-\infty, 0] \\
    &\quad\cup \{(x, y, t) \in A \times [0, 1] \times (0, \infty) : t < y\} \\
    &= [0, 1]^d \times [0, 1] \times (-\infty, 0] \\
    &\quad\cup A \times \{(y, t) \in [0, 1] \times (0, \infty) : t < y\}.
\end{align*}
Hence the subgraphs of the functions in $\CF$ are of the form of a fixed set union a fixed set times a set from $\mathcal{A}$. From this observation, it is easy to see that $\CF$ is a VC-class of functions with the same VC-dimension as $\mathcal{A}$, namely $2d$ (see Appendix~\ref{sec:moreresults} for background). We may now apply a combination of Theorems 2.6.7 and 2.14.28 in~\cite{VaartWellner} (see also \cite{Talagrand1994}) to conclude that
\begin{align*}
    &\BP\Big(\sup_{A \in \mathcal{A}}\Big|\frac{1}{n_\CI}\sum_{i = 1}^{n_\CI} (\1_{\{X_i \in A\}}Y_i - \BE[\1_{\{X_1 \in A\}}Y_1]) \Big| > \delta \underline{v}_n \mid \theta, \CJ \Big) \\
    &\leq\Big(\frac{\max\{D(d) \delta \underline{v}_n \sqrt{n_\CI}, \sqrt{4d}\}}{\sqrt{4d}} \Big)^{4d}e^{-2 \delta^2 \underline{v}_n^2 n_\CI}
\end{align*}
for $D(d)$ a constant that depends on the VC-dimension of $\CF$ only, which is a function of $d$. This result is easily extended to bounded responses satisfying $Y \in [0, M]$ by considering the scaled function class
\begin{equation*}
    \CF_M := \{f_A:[0, 1]^d \times [0, M] \to [0, 1], f_A(x, y) = \1_A(x)y/M : A \in \mathcal{A}\},
\end{equation*}
which yields the same tail bound, but with $\delta$ replaced by $\delta/M$. To extend to $Y \in [-M, M]$ a.s., simply decompose $Y = Y \1_{\{Y \in [0, M]\}} + Y \1_{\{Y \in [-M, 0]\}}$ and make a union bond. To summarise, when $Y \in [-M, M]$ a.s., we have the bound
\begin{align*}
    &\BP\Big(\sup_{A \in \mathcal{A}}\Big|\frac{1}{n_\CI}\sum_{i = 1}^{n_\CI} (\1_{\{X_i \in A\}}Y_i - \BE[\1_{\{X_1 \in A\}}Y_1]) \Big| > \delta \underline{v}_n \mid \theta, \CJ \Big)\\
    &\leq 2\left(\frac{\max\Big\{\frac{D(d) \delta \underline{v}_n \sqrt{n_\CI}}{M}, \sqrt{4d}\Big\}}{\sqrt{4d}} \right)^{4d}e^{-2 \delta^2 \underline{v}_n^2 n_\CI/M^2}.
\end{align*}
Now assume that $Y$ may be unbounded but with finite moment generating function $\kappa_Y$ in the neighbourhood $(-c, c)$ of zero. Then we decompose
\begin{align*}
    &\BP\Big(\sup_{A \in \mathcal{A}}\Big|\frac{1}{n_\CI}\sum_{i = 1}^{n_\CI} \1_{\{X_i \in A\}}Y_i - \BE[\1_{\{X_1 \in A\}}Y_1] \Big| > \delta \underline{v}_n \mid \theta, \CJ \Big) \leq \\
    &\BP\Big(\sup_{A \in \mathcal{A}}\Big|\frac{1}{n_\CI}\sum_{i = 1}^{n_\CI} \1_{\{X_i \in A\}}Y_i\1_{\{|Y_i| \leq M\}} - \BE[\1_{\{X_1 \in A\}}Y_1\1_{\{|Y_1| \leq M\}}] \Big| > \delta \underline{v}_n \mid \theta, \CJ \Big) + \\
    &\BP\Big(\sup_{A \in \mathcal{A}}\Big|\frac{1}{n_\CI}\sum_{i = 1}^{n_\CI} \1_{\{X_i \in A\}}Y_i\1_{\{|Y_i| > M\}} - \BE[\1_{\{X_1 \in A\}}Y_1\1_{\{|Y_1| > M\}}] \Big| > \delta \underline{v}_n \mid \theta, \CJ \Big) \\
    &=: (1) + (2).
\end{align*}
Here (1) has the same exponential bound as before. As for (2), we combine a Markov inequality with a Cauchy--Schwarz inequality to obtain
\begin{align*}
    (2) &\leq \frac{1}{\delta \underline{v}_n} \BE\Big[\sup_{A \in \mathcal{A}}\Big|\frac{1}{n_\CI}\sum_{i = 1}^{n_\CI} \1_{\{X_i \in A\}}Y_i \1_{\{|Y_i| > M\}} - \BE[\1_{\{X_1 \in A\}}Y_1 \1_{\{|Y_1| > M\}}] \Big| \Big] \\
    &\leq \frac{2}{\delta \underline{v}_n} \BE[|Y| \1_{\{|Y| > M\}}] \leq 2\sqrt{\BE[Y^2]}\frac{\sqrt{\BP(|Y| > M)}}{\delta \underline{v}_n}.
\end{align*}
To bound the tail probability of $|Y|$, simply apply a Markov inequality with the function $x \mapsto e^{tx}$ for some $t \in (0, c)$ to obtain
\begin{align*}
    \BP(|Y| > M) &= \BP(Y > M) + \BP(-Y > M) = \BP(e^{tY}> e^{tM}) + \BP(e^{-tY} > e^{tM}) \\
    &\leq e^{-tM}(\kappa_Y(t) + \kappa_Y(-t)). 
\end{align*}
All in all, we have shown
\begin{equation*}
    (2) \leq \frac{ 2\sqrt{\BE[Y^2]}(\kappa_Y(t) + \kappa_Y(-t))}{\delta} \cdot\frac{e^{-tM/2}}{\underline{v}_n}.
\end{equation*}
We are free to choose $M$ as a function of $n_\CI$, so choose $M_n$ as stated in the theorem. Combine with Lemma~\ref{lem:convergenceconditioning} to complete the proof of uniform consistency $\widehat{m}$. It only remains to transfer uniform consistency of $\widehat{m}$ and $\widehat{f}$ to uniform consistency of $T$. Start by writing
\begin{align*}
    |T(x, \theta) - \BE[Y \mid X = x]| &= \left|\frac{\widehat{m}(x, \theta)}{\widehat{f}(x, \theta)} - \frac{m(x)}{f(x)} \right| = \left|\frac{\widehat{m}(x, \theta) f(x) - f(x)m(x) + f(x)m(x) - m(x)\widehat{f}(x, \theta)}{f(x)\widehat{f}(x, \theta)} \right| \\
    &\leq \frac{|\widehat{m}(x, \theta) - m(x)|}{\widehat{f}(x, \theta)} + \frac{|m(x)||\widehat{f}(x, \theta) - f(x)|}{f(x)\widehat{f}(x, \theta)}.
\end{align*}
From Assumption \ref{asm:covariate} and Jensen's inequality, $|m(x)| \leq \sqrt{K}$. Also, for $n$ large enough, $|\widehat{f}(x, \theta) - f(x)| < \varepsilon/2$, and so by the reverse triangle inequality, for large enough $n$,
\begin{equation*}
    \widehat{f}(x, \theta) = f(x) - |\widehat{f}(x, \theta) - f(x)| > f(x) - \frac{\varepsilon}{2} \geq \frac{\varepsilon}{2}.
\end{equation*}
But then for large enough $n$, we have the bound
\begin{equation*}
    \sup_{x \in [0, 1]^d}|T(x, \theta) - \BE[Y \mid X = x]| \leq \frac{2}{\varepsilon}\sup_{x \in [0, 1]^d}|\widehat{m}(x, \theta) - m(x)| + \frac{2\sqrt{K}}{\varepsilon^2}\sup_{x \in [0, 1]^d}|\widehat{f}(x, \theta) - f(x)|
\end{equation*}
and the proof is complete.
\end{proof}

\begin{proof}[Proof of Proposition~\ref{prop:regularI}]
We claim that
\begin{equation*}
     \frac{N_L^j(x, \theta)}{n_\CI} - \BP(X_1^j \in (a_{n, j}^x, b_{n, j}^x] \mid \theta, \CJ) \to 0 \quad \BP\text{-a.s.}
\end{equation*}
Fix $\delta > 0$. Then by a fourth order Markov inequality and the fact that (using honesty) conditional on $\theta, \CJ$, $N_L^j(x, \theta)$ is binomial distributed with parameters $n_\CI$ and $p_j(\theta, \CJ) := \BP(X_1^j \in (a_{n, j}^x, b_{n, j}^x] \mid \theta, \CJ)$,
\begin{align*}
    &\BP\Big(\Big| \frac{N_L^j(x, \theta)}{n_\CI} - p_j(\theta, \CJ) \Big| > \delta \mid \theta, \CJ \Big) \\
    &\leq \frac{n_\CI p_j(\theta, \CJ)(1 - p_j(\theta, \CJ))(1 + (3n_\CI - 6)p_j(\theta, \CJ)(1 - p_j(\theta, \CJ)))}{n_\CI^4 \delta^4} \\
    &\leq \frac{3}{\delta^4 n_\CI^2}
\end{align*}
which is summable and deterministic, proving the claim. We have from Assumption~\ref{asm:covariate} that
\begin{equation*}
    p_j(\theta, \CJ) = \int_{a_{n, j}^x}^{b_{n, j}^x} f^j(\widetilde{x})\mathrm{d}\widetilde{x} \in [\varepsilon(b_{n, j}^x - a_{n, j}^x), C(b_{n, j}^x - a_{n, j}^x)],
\end{equation*}
where $f^j$ is the density of $X_1^j$. The equivalence stated in the proposition now follows immediately. 
\end{proof}

\begin{proof}[Proof of Proposition~\ref{prop:regularII}]
As in the proof of Theorem~\ref{thm:uniformconsistency}, we use that the collection of all possible leaves is contained in the class $\mathcal{A}$ of all half-open rectangles in $[0, 1]^d$. Letting $\BP^{(n_\CI)}$ denote the empirical measure of $X_1, ..., X_{n_\CI}$, we get for any $\delta \in (0, 1]$ that
\begin{align*}
    &\BP\Big(\Big|\frac{N_L(x, \theta)}{n_\CI} - \BP(X_1 \in L(x, \theta) \mid \theta, \CJ) \Big| > \delta \mid \theta, \CJ \Big)\\
    & \leq \BP\Big(\sup_{A \in \mathcal{A}}\Big| \BP^{(n_\CI)}(A) - \BP(X_1 \in A) \Big| > \delta\Big)\\
    &\leq c n_\CI^{4d} e^{-2n_\CI\delta^2},
\end{align*}
for a constant $c \leq e^{4\delta(1 + \delta)}$ where we have used Lemma~\ref{lem:Devroyebound} from the appendix. By using the tower property, we can remove the conditioning in the first probability. Here we assume $d > 1$. We return to the case $d = 1$ below. It follows that with probability at least $1 - c n_\CI^{4d} e^{-2n_\CI\delta^2}$, we have
\begin{equation*}
    \BP(X_1 \in L(x, \theta) \mid \theta, \CJ) \geq \frac{N_L(x, \theta)}{n_\CI} - \delta
\end{equation*}
Now choose $\delta$ to be dependent on $n$ via $\delta = k_n/2n_\CI$. Define
\begin{equation*}
    b_n(x) := \BP\Big(\BP(X_1 \in L(x, \theta) \mid \theta, \CJ) \geq \frac{N_L(x, \theta)}{n_\CI} - \frac{k_n}{2n_\CI}\Big).
\end{equation*}
Since
\begin{equation*}
    \BP(X_1 \in L(x, \theta) \mid \theta, \CJ) = \int_{L(x, \theta)} f(\widetilde{x})\mathrm{d} \widetilde{x} \leq C \lambda(L(x, \theta)),
\end{equation*}
we have
\begin{align*}
    \BP\Big(n_\CI \lambda(L(x, \theta)) &\geq \frac{k_n}{2C}\Big) \geq \BP\Big(\BP(X_1 \in L(x, \theta) \mid \theta, \CJ) \geq \frac{k_n}{2n_\CI} \Big) \\
    &\geq \BP\Big(\BP(X_1 \in L(x, \theta)) \geq \frac{N_L(x, \theta)}{n_\CI} - \frac{k_n}{2n_\CI} \geq \frac{k_n}{2n_\CI} \Big) \\
    &= \BP\Big((\BP(X_1 \in L(x, \theta) \geq \frac{N_L(x, \theta)}{n_\CI} - \frac{k_n}{2n_\CI}, N_L(x, \theta) \geq k_n\Big) \\
    &\geq a_n(x) + b_n(x) - 1 \geq a_n(x) - cn_\CI^{4d}e^{-k_n^2/2n_\CI} \to 1,
\end{align*}
and so Lemma \ref{lem:BCinfinity} (1) applies to show that $n_\CI \lambda(L(x, \theta)) \to \infty$ in probability. The almost sure statement follows from (2) of Lemma \ref{lem:BCinfinity} since
\begin{equation*}
    \sum_{n = 1}^\infty (1 - (a_n(x) - cn_\CI^{4d}e^{-k_n^2/2n_\CI})) = \sum_{n = 1}^\infty(1 - a_n(x)) + \sum_{n = 1}^\infty cn_\CI^{4d}e^{-k_n^2/2n_\CI} < \infty.
\end{equation*}
When $d = 1$, one should just replace $c n_\CI^{4d} e^{-2n_\CI\delta^2}$ in the proof by $c (n_\CI^{4} + 1) e^{-2n_\CI\delta^2}$  and the argument still goes through.
\end{proof}

\begin{proof}[Proof of Proposition~\ref{prop:regularIIstrong}]
We again assume $d > 1$, since the proof for $d = 1$ is very similar. In the proof of Proposition~\ref{prop:regularII}, it was established that
\begin{equation*}
    \BP\Big(n_\CI \lambda(L(x, \theta)) \geq \frac{k_n}{2C} \Big) \geq a_n(x) - cn_\CI^{4d}e^{-k_n^2/2n_\CI}.
\end{equation*}
Let $A_n(x) := \{n_\CI \lambda(L(x, \theta)) \geq k_n/2C\}$. We can now verify the conditions in Theorem~\ref{thm:strongconsistency} depending on the cases (i) and (ii). For (i), we get
\begin{equation*}
    \frac{1}{n_\CI^2 \lambda(L(x, \theta))^2} \leq \frac{4C^2}{k_n^2} \1_{A_n(x)} + \frac{1}{n_\CI^2 \lambda(L(x, \theta))^2} \1_{A_n(x)^c},
\end{equation*}
which is summable by the assumption on $k_n$ and since $\1_{A_n(x)}$ is eventually equal to one with probability one. Now consider the set of conditions in (ii) and let $\delta > 0$. Then
\begin{align*}
    \exp(-\delta n_\CI \lambda(L(x, \theta))) \leq \exp\Big(-\frac{\delta k_n}{2C} \Big)\1_{A_n(x)} + \1_{A_n(x)^c} \leq \exp\Big(-\frac{\delta k_n}{2C} \Big) + \1_{A_n(x)^c}.
\end{align*}
The proof is now complete by the summability assumptions stated in (ii).
\end{proof}

\begin{proof}[Proof of Corollary~\ref{cor:regularconsistent}]
The summability requirement in (ii) is clearly satisfied for both choices of $k_n$. For the summability requirement in Proposition~\ref{prop:regularII}, we have for (1) that
\begin{equation*}
    \frac{n_\CI^{4d}}{e^{k_n^2/2n_\CI}} \sim \frac{n_\CI^{4d}}{e^{n_\CI^{2\beta - 1}/2}},\quad \mbox{and}
\quad    \lim_{n \to \infty }\frac{\frac{n^{4d}}{e^{n^{2\beta - 1}/2}}}{1/n^2} = \lim_{n \to \infty} \frac{n^{4d + 2}}{e^{n^{2\beta - 1}/2}} = 0,
\end{equation*}
so a comparison test with the series $\sum_{n = 1}^\infty \frac{1}{n^2}$ yields that (1) satisfies the desired convergence statement. As for (2), now assuming that $k_n \sim \sqrt{n_\CI \log(n_\CI)^\beta}$ for $\beta > 1$,
\begin{equation*}
    \frac{n_\CI^{4d}}{e^{k_n^2/2n_\CI}} \sim n_\CI^{4d - \log(n_\CI)^{\beta - 1}/2} = o\Big(\frac{1}{n_\CI^2} \Big),
\end{equation*}
so the same comparison test as in (1) works. As for the summability of $k_n^{-2}$, this obviously holds for (1). For (2), we have
\begin{equation*}
    \int_2^\infty \frac{1}{y \log(y)^\beta}\mathrm{d} y = \int_{\log 2}^\infty v^{-\beta}\mathrm{d} v = \frac{\log(2)^{1 - \beta}}{\beta - 1} < \infty,
\end{equation*}
so the integral test shows that $k_n^{-2}$ is indeed summable.
\end{proof}

\begin{proof}[Proof of Proposition~\ref{prop:regularIuniform}]
We have
\begin{align*}
    &\BP\Big(\sup_{x \in [0, 1]^d}\Big|\frac{N_L^j(x, \theta)}{n_\CI} - p_j(\theta, \CJ) \Big| > \delta \mid \theta, \CJ \Big)\\
    &\leq \BP\Big(\sup_{A \in \mathcal{A}}\Big|\BP^{(n_\CI)}_j(A) - \BP(X_1^j \in A) \Big| > \delta \mid \theta, \CJ \Big)\leq c(n_\CI^4 + 1)e^{-2\delta n_\CI^2}
\end{align*}
where $\mathcal{A} = \{(a, b] : a, b \in [0, 1], a < b\}$ by Lemma~\ref{lem:Devroyebound}. This is summable, and applying Lemma~\ref{lem:convergenceconditioning}, we get
\begin{equation*}
    \sup_{x \in [0, 1]^d}\Big|\frac{N_L^j(x, \theta)}{n_\CI} - p_j(\theta, \CJ) \Big| \to 0 \quad \BP\text{-a.s.},
\end{equation*}
from which it follows that
\begin{equation*}
    \sup_{x \in [0, 1]^d} \frac{N_L^j(x, \theta)}{n_\CI} \overset{\BP}{\to} 0 \quad \Leftrightarrow \quad \sup_{x \in [0, 1]^d} p_j(\theta, \CJ) \overset{\BP}{\to} 0
\end{equation*}
and the same equivalence is of course true for convergence almost surely. The claim then holds since
\begin{equation*}
    \sup_{x \in [0, 1]^d} p_j(\theta, \CJ) \in \Big[\varepsilon \sup_{x \in [0, 1]^d} \{b_{n, j}^x - a_{n, j}^x\}, C\sup_{x \in [0, 1]^d}\{b_{n, j}^x - a_{n, j}^x\} \Big].
\end{equation*}
\end{proof}

\begin{proof}[Proof of Proposition~\ref{prop:regularIIstronguniform}]
We need to verify the assumptions of (ii) of Theorem~\ref{thm:uniformconsistency}. From the proof of Proposition~\ref{prop:regularII}, we have the uniform lower bound in $x$ given by
\begin{equation*}
    \BP\Big(\lambda(L(x, \theta)) \geq \frac{k_n}{2Cn_\CI} \Big) \geq a_n - cn_\CI^{4d} e^{-k_n^2/2n_\CI}.
\end{equation*}
To establish a similar bound for $\BP(\underline{v}_n \geq k_n/(2C n_\CI))$, we make the observation that there are at most $n_\CI$ leaves, which gives us the union bound
\begin{align*}
    \BP\Big(\underline{v}_n < \frac{k_n}{2Cn_\CI}\Big) &= \BP\Big(\lambda(L(x, \theta)) < \frac{k_n}{2Cn_\CI} \text{ for some } x \in [0, 1]^d \Big) \\
    &\leq n_\CI (1 - (a_n - cn_\CI^{4d}e^{-k_n^2/2n_\CI})) = n_\CI(1 - a_n) + cn_\CI^{4d + 1}e^{-k_n^2/2n_\CI}.
\end{align*}
The right hand side is summable and so if $A_n := \{\underline{v}_n \geq k_n/(2Cn_\CI)\}$, the indicator $\1_{A_n}$ goes to one almost surely. Choose $M_n =  n_\CI^\gamma$ for some $\gamma \in (0, \beta - 1/2)$. Then for any $t > 0$,
\begin{equation*}
    \frac{e^{-tM_n}}{\underline{v}_n} \leq 2C \frac{n_\CI}{k_n}e^{-tM_n}\1_{A_n} + \1_{A_n^c} \frac{e^{-tM_n}}{\underline{v}_n}.
\end{equation*}
The first term is on the order
\begin{equation*}
    n_\CI^{1 - \beta}e^{-tn_\CI^\gamma}
\end{equation*}
which is summable by the integral test. Since $\1_{A_n^c}$ goes to zero a.s., it is even eventually equal to zero with probability one, and thus the second term is also summable. As for the term involving the maximum, we have for any $\delta > 0$ that
\begin{equation*}
    e^{-\delta \underline{v}_n^2 n_\CI/M_n^2} \leq \1_{A_n} \exp\Big(-\frac{\delta}{4C^2} \cdot \frac{k_n^2}{n_\CI M_n^2}\Big) + \1_{A_n^c}.
\end{equation*}
Again, $\1_{A_n^c}$ is eventually zero a.s. and is thus summable. The first term is on the order
\begin{equation*}
    e^{-\delta' n_\CI^{2\beta - 1 - 2\gamma}}
\end{equation*}
for a constant $\delta' > 0$, and since $2\beta - 1 - 2\gamma > 0$ by the choice of $\gamma$, the term involving the maximum is summable. This completes the proof.
\end{proof}

\begin{proof}[Proof of Lemma~\ref{lem:numsplits}]
The proof of this result mirrors the first half of the proof of Lemma 2 in~\cite{WagerAthey}. We have
\begin{align*}
    &\BP(n_\CI \alpha^{s_n(x)} \leq N_L(x, \theta) \leq n_\CI(1 - \alpha)^{s_n(x)}, k_n \leq N_L(x, \theta) \leq 2k_n - 1) \leq \\
    &\BP(n_\CI \alpha^{s_n(x)} \leq N_L(x, \theta) \leq 2k_n - 1) \leq \BP\Big(s_n(x) \geq \frac{\log(n_\CI/(2k_n - 1))}{\log(\alpha^{-1})}\Big).
\end{align*}
The first probability is greater than or equal to $\overline{a}_n(x) + b_n(x) - 1$ which goes to one, and since $n_\CI/(2k_n - 1) \to \infty$, Lemma \ref{lem:BCinfinity} applies to show that $s_n(x) \to \infty$ in probability. Now assume that the tree is a random-split tree and let
\begin{equation*}
    B_n^j \sim \text{Bin}\Big(\Big\lfloor \frac{\log(n_\CI/(2k_n - 1))}{\log(\alpha^{-1})} \Big\rfloor, \overline{\pi}_j \Big).
\end{equation*}
Then for any real number $m > 0$,
\begin{align*}
    \BP(s_{n, j}(x) > m) &\geq \BP\Big(s_{n, j}(x) > m, s_n(x) \geq \frac{\log(n_\CI/(2k_n - 1))}{\log(\alpha^{-1})} \Big) \\
    &\geq \BP\Big(B_n^j > m, s_n(x) \geq \frac{\log(n_\CI/(2k_n - 1))}{\log(\alpha^{-1})} \Big) \\
    &\geq \BP(B_n^j > m) + \BP\Big(s_n(x) \geq \frac{\log(n_\CI/(2k_n - 1))}{\log(\alpha^{-1})} \Big) - 1.
\end{align*}
The second probability goes to one by the argument above. Now let $m$ depend on $n$ via
\begin{equation*}
    m_n = \delta \Big\lfloor \frac{\log(n_\CI/(2k_n - 1))}{\log(\alpha^{-1})} \Big\rfloor \overline{\pi}_j
\end{equation*}
where $\delta \in (0, 1)$. By the SLLN,
\begin{equation*}
    \frac{B_n^j}{m_n} = \frac{1}{\delta \overline{\pi}_j} \frac{B_n^j}{\Big\lfloor \frac{\log(n_\CI/(2k_n - 1))}{\log(\alpha^{-1})} \Big\rfloor} \to \frac{1}{\delta} > 1 \quad \BP\text{-a.s.}
\end{equation*}
and so, $\BP(B_n^j > m_n) \to 1$, and we can conclude by Lemma \ref{lem:BCinfinity} that $s_{n,j}(x) \to \infty$ in probability.
\end{proof}

\begin{proof}[Proof of Lemma~\ref{lem:regularsidelength}]
Fix $j \in \{1, \dots, d\}$ and $\delta > 0$. Then
\begin{align*}
    \BP\Big(\sup_{x \in [0, 1]^d} \frac{N_L^j(x, \theta)}{n_\CI} > \delta \Big) &\leq \BP\Big(\sup_{x \in [0, 1]^d }\frac{N_L^j(x, \theta)}{n_\CI} > \delta, \frac{N_L^j(x, \theta)}{n_\CI} \leq (1 - \alpha_j)^{s_{n,j}(x)} \Big) \\
    &+ \BP(N_L(x, \theta) > n_\CI (1 - \alpha_j)^{s_{n,j}(x)}) \\
    &\leq \BP((1 - \alpha_j)^{\inf_{x \in [0, 1]^d}s_{n, j}(x)} > \delta) + 1 - b_n(x) \to 0
\end{align*}
as claimed.
\end{proof}

\begin{proof}[Proof of Lemma~\ref{lem:uniformconsistency}]
Fix $j \in \{1, ..., d\}$. By the tower property,
\begin{align*}
    &\BE[b_{n, j}^x - a_{n, j}^x]\\
    &= \BE\Big[\BE\Big[\prod_{k = 1}^{s_{n, j}(\theta)}U_{k, j}(x) \ \Big| \ \theta \Big] \Big] = \BE\Big[\prod_{k = 1}^{s_{n, j}(\theta)}\BE[U_{k, j}(x)] \Big] = \BE\Big[\Big(\frac{1}{2} \Big)^{s_{n, j}(\theta)} \Big].
\end{align*}
It is easily verified that for $B \sim \text{Bin}(n, p)$, we have $\BE[a^B] = ((1 - (1 - a)p)^n$. Since $s_{n, j}(\theta) \sim \text{Bin}(s_n, 1/d)$, we get
\begin{equation*}
    \BE[b_{n, j}^x - a_{n, j}^x] = \Big(1 - \frac{1}{2d} \Big)^{s_n},
\end{equation*}
and an application of Markov's inequality yields that $b_{n, j}^x - a_{n, j}^x \overset{\BP}{\to} 0$.
\end{proof}

\begin{proof}[Proof of Proposition~\ref{prop:uniformconsistency}]
We have 
\begin{equation*}
    \log(n_\CI \lambda(L(x, \theta))) = \log(n_\CI) + \sum_{k = 1}^{s_n} \log(U_k(x)).
\end{equation*}
Dividing by $s_n$, we have by the SLLN that
\begin{equation*}
    \frac{1}{s_n} \sum_{k = 1}^{s_n} \log(U_k(x)) \to \BE[\log(U_1(x))] = -1 \quad \BP\text{-a.s.}
\end{equation*}
It follows that $n_\CI \lambda(L(x, \theta)) \to \infty$ $\BP$-a.s. under the assertion of the proposition.
\end{proof}

\begin{proof}[Proof of Proposition~\ref{prop:centerednonuniform}]
Let $M_k^j$ denote the minimum number of times feature $j$ has been split upon across all possible paths down a tree of depth $k$ and define the function
\begin{equation*}
    G_k(m) = \BP(M_k^j \geq m), \quad m = 0, 1, \dots 
\end{equation*}
Since the tree is assumed to be balanced, $M_{s_n}^j \overset{d}{=} \inf_{x \in [0, 1]^d} s_{n, j}(x)$. We determine a recursive relation for $G_k(m)$. Since the choice of feature to split upon is independent of every other split, we can consider a decision tree as being built up of independent subtrees. In particular, letting $I^j$ denote the indicator of whether we split on $j$ in the root, we may write
\begin{equation*}
    M_k^j \overset{d}{=} I^j + \min\{M_{k - 1}^{j, (1)}, M_{k - 1}^{j, (2)}\}
\end{equation*}
where $M_{k - 1}^{j, (1)}$ and $M_{k - 1}^{j, (2)}$ are the minimal number of times across all paths down the tree that $j$ is split upon in the the two subtrees created from splitting the root node. By construction, $M_{k - 1}^{j, (1)}$ and $M_{k - 1}^{j, (2)}$ are independent and have the same distribution as $M_{k - 1}^j$, which implies that
\begin{align*}
    G_k(m) &= p_j\BP(\min\{M_{k - 1}^{j, (1)}, M_{k - 1}^{j, (2)}\} + 1 \geq m)\\
    &\quad+ (1 - p_j)\BP(\min\{M_{k - 1}^{j, (1)}, M_{k - 1}^{j, (2)}\} \geq m) \\
    &= p_j G_{k - 1}(m - 1)^2 + (1 - p_j) G_{k - 1}(m)^2
\end{align*}
and in general, $G_k(0) = 1$. In particular, for $m = 1$, we have
\begin{equation*}
    G_k(1) = p_j + (1 - p_j)G_{k - 1}(m)^2, \quad G_1(1) = p_j.
\end{equation*}
Define
\begin{equation*}
    f(x) := p_j + (1 - p_j)x^2.
\end{equation*}
We solve for fixed points for $f$. We have $x = f(x)$ if and only if $x = 1$ or $x = p_j/(1 - p_j)$. We have $f'(x) = 2(1 - p_j)x$, and we see that
\begin{equation*}
    f'(1) = 2(1 - p_j) > 1 \quad \text{and} \quad f'\Big(\frac{p_j}{1 - p_j}\Big) = 2p_j < 1
\end{equation*}
by the assumption $p_j \in (0, 1/2)$ and $G_k(1)$ is obviously non-decreasing in $k$. It follows that
\begin{equation*}
    \BP\Big(\inf_{x \in [0, 1]^d}s_{n, j}(x) \geq 1\Big) = G_{s_n}(1) \to \frac{p_j}{1 - p_j} \quad \text{as} \quad s_n \to \infty,
\end{equation*}
from which the result follows.
\end{proof}

\begin{proof}[Proof of Theorem~\ref{thm:centereduniform}]
From the discussion above the theorem, we already have
\begin{equation*}
    \sup_{x \in [0, 1]^d} \{b_{n, j}^x - a_{n, j}^x\} \to 0 \quad \BP\text{-a.s.}
\end{equation*}
It remains to choose a suitable sequence $\{M_n\}$ such that the two summability requirements of Theorem~\ref{thm:uniformconsistency} hold. First of all, note that since the leaf volumes are all identical, $\underline{v}_n = 2^{-s_n}$. Now let $t, \delta > 0$ be arbitrary and choose $s_n = \lceil \log_2 n_\CI^{1 - \beta}\rceil$ for $\beta \in (1/2, 1)$ and let $M_n = s_n^\gamma$ for $\gamma > 1$. Then
\begin{equation*}
    \frac{e^{-tM_n}}{\underline{v}_n} = e^{-s_n (ts_n^{\gamma - 1} - \log 2)}.
\end{equation*}
Choose some $\epsilon \in (0, \gamma - 1)$. Then for large enough $n$,
\begin{equation*}
    \frac{e^{-tM_n}}{\underline{v}_n} \leq e^{-s_n t s_n^{\gamma - 1 - \epsilon}}  \leq e^{-t(\log_2 n_\CI^{1 - \beta})^{\gamma - \epsilon - 1} \log_2 n_\CI^{1 - \beta}} = {n_\CI}^{-t\left(\frac{(1 - \beta)}{\log 2}\right)^{\gamma - \epsilon} (\log n_\CI)^{\gamma - 1 - \epsilon}}
\end{equation*}
For large enough $n$, the exponent will be strictly less than one, and so the first summability requirement is satisfied. As for the second, note that
\begin{equation*}
    \frac{\underline{v}_n \sqrt{n_\CI}}{M_n} = 2^{-s_n} s_n^{-\gamma} \sqrt{n_\CI} \leq \frac{n_\CI^{-(1 - \beta) + 1/2} }{(\log_2 n_\CI^{1 - \beta})^\gamma} = \frac{n_\CI^{\beta - 1/2}}{(1 - \beta)^\gamma (\log_2 n_\CI)^\gamma}.
\end{equation*}
Since $\beta \in (1/2, 1)$, this goes to infinity, and for large enough $n$, we have the bound
\begin{equation*}
    \max\Big\{\Big(\frac{\underline{v}_n \sqrt{n_\CI}}{M_n}\Big)^{4d}, 1 \Big\} \leq \frac{n_\CI^{4d(\beta - 1/2)}}{(1 - \beta)^{4d\gamma}(\log_2 n_\CI)^{4d\gamma}}.
\end{equation*}
As for the exponential part,
\begin{align*}
    \exp\Big(-\frac{\delta \underline{v}_n^2 n_\CI}{M_n^2} \Big) &= \exp\Big(-\delta 4^{-s_n} n_\CI s_n^{-2\gamma}\Big)\\
    &\leq \exp\left(-\delta  \frac{4^{-(\log_2 n_\CI^{1 - \beta} + 1)} n_\CI}{(\log_2 n_\CI^{1 - \beta} + 1)^{2\gamma}} \right) = \exp\left(-\frac{\delta}{4}\frac{n_\CI^{2\beta - 1}}{(\log_2 n_\CI^{1 - \beta} + 1)^{2\gamma}} \right).
\end{align*}
Since $\beta > 1/2$, the terms
\begin{equation*}
    \max\Big\{\Big(\frac{\underline{v}_n \sqrt{n_\CI}}{M_n}\Big)^{4d}, 1 \Big\}e^{-\delta \underline{v}_n^2 n_\CI/M_n^2}
\end{equation*}
are summable by the integral test. This completes the proof. 
\end{proof}

\subsection{Proofs for Section \ref{sec:bootstrap}}

\begin{proof}[Proof of Lemma~\ref{lem:convergenceconditioningbootstrap}]
We can follow the idea from the proof of Lemma~\ref{lem:convergenceconditioning}. Define
\begin{align*}
    \CD &:= \sigma(\theta, (X_1, Y_1), \ldots, (X_n, Y_n), \boldsymbol{W}^\CI, \boldsymbol{W}^\CJ_\infty), \\
    \CG &:= \sigma(\theta, \CJ, \boldsymbol{W}^\CJ_ \infty), \\
    \CH &:= \sigma(\CJ_\infty \setminus \CJ).
\end{align*}
By definition, $\CG \subseteq \CD$ and $\CD \indep \CH$. The random variables
\begin{equation*}
    \1_{\{|\widehat{f}^*(x, \theta) - \BE[\widehat{f}^*(x, \theta) \mid \theta, \CJ, \boldsymbol{W}^\CJ]| > \delta\}} \quad \text{and} \quad \1_{\{|\widehat{m}^*(x, \theta) - \BE[\widehat{m}^*(x, \theta) \mid \theta, \CJ, \boldsymbol{W}^\CJ]| > \delta\}}
\end{equation*}
are both $\CD$-measurable. Since $\sigma(\theta, \CJ_\infty, \boldsymbol{W}^\CJ_\infty) = \CG \lor \CH$, the proof is complete by Lemma~\ref{lem:condmean}.
\end{proof}

\begin{proof}[Proof of Theorem~\ref{thm:bootstraptreeconsistency}]
One can essentially redo all the proofs of weak consistency from above. When showing that $\widehat{f}^*(x, \theta) \to  f(x)$ in probability, one just has to condition on $\theta, \CJ$ and $\boldsymbol{W}^\CJ$ instead of just $\theta, \CJ$. Then one obtains
\begin{equation*}
    \BE[\widehat{f}^*(x, \theta) \mid \theta, \CJ, \boldsymbol{W}^\CJ] = \frac{1}{\lambda(L^*(x, \theta))} \int_{L^*(x, \theta)} f(\widetilde{x})\mathrm{d}\widetilde{x} \overset{\BP}{\to} f(x)
\end{equation*}
in the same manner as before. As for the conditional variance, one gets
\begin{align*}
    \Var[\widehat{f}^*(x, \theta) \mid \theta, \CJ, \boldsymbol{W}^\CJ] &= \BE[\widehat{f}^*(x, \theta)^2 \mid \theta, \CJ, \boldsymbol{W}^\CJ] - \BE[\widehat{f}^*(x, \theta) \mid \theta, \CJ, \boldsymbol{W}^\CJ]^2.
\end{align*}
As we saw before, the second term converges to $f(x)^2$. Hence we need that the first term converges to $f(x)^2$ in order for the variance to vanish. We compute
\begin{align*}
    &\BE[\widehat{f}^*(x, \theta)^2  \mid \theta, \CJ, \boldsymbol{W}^\CJ]\\
    &= \sum_{i = 1}^{n_\CI} \sum_{j = 1}^{n_\CI} \BE\Big[\frac{W_i^\CI W_j^\CI \1_{\{X_i, X_j \in L^*(x, \theta)\}}}{n_\CI^2 \BE[W_1^\CI]^2 \lambda(L^*(x, \theta))^2} \mid \theta, \CJ, \boldsymbol{W}^\CJ \Big] \\
    &= \frac{n_\CI \BE[(W_1^\CI)^2] \int_{L^*(x, \theta)} f(\widetilde{x})\mathrm{d}\widetilde{x} + n_\CI(n_\CI - 1)\BE[W_1^\CI W_2^\CI] \Big(\int_{L^*(x, \theta)} f(\widetilde{x})\mathrm{d}\widetilde{x} \Big)^2}{n_\CI^2 \BE[(W_1^\CI)]^2 \lambda(L^*(x, \theta))^2} \\
    &= \frac{1}{n_\CI \BE[W_1^\CI] \lambda(L^*(x, \theta))} \frac{\BE[(W_1^\CI)^2]}{\BE[W_1^\CI]} \frac{1}{\lambda(L^*(x, \theta))} \int_{L^*(x, \theta)} f(\widetilde{x})\mathrm{d}\widetilde{x} \\
    &\quad+ \frac{n_\CI - 1}{n_\CI} \frac{\BE[W_1^\CI W_2^\CI]}{\BE[W_1^\CI]^2} \Big(\frac{1}{\lambda(L(x, \theta))}\int_{L^*(x, \theta)}f(\widetilde{x})\mathrm{d}\widetilde{x} \Big)^2.
\end{align*}
The first term converges to $0 \cdot L_{2,1} f(x) = 0$ in probability, and the second term converges to $1 \cdot 1 \cdot f(x)^2 = f(x)^2$ in probability. We conclude the proof of $\widehat{f}^*(x, \theta) \overset{\BP}{\to} f(x)$ in the same way as for Proposition~\ref{prop:densityestimatorconsistent}. The exact same computation shows that $\widehat{m}^*(x, \theta) \overset{\BP}{\to} m(x)$, and the continuous mapping theorem completes the proof.
\end{proof}

\begin{proof}[Proof of Proposition~\ref{prop:regularIbootstrap}]
Let $p_j(\theta, \CJ, \boldsymbol{W}^\CJ) := \BP(X_1^j \in L^*(x, \theta) \mid \theta, \CJ, \boldsymbol{W}^\CJ). $ We claim that
\begin{equation*}
    \frac{N_L^{j,*}(x, \theta)}{n_\CI \BE[W_1^\CI]} - p_j(\theta, \CJ, \boldsymbol{W}^\CJ) \overset{\BP}{\to} 0.
\end{equation*}
The conditional mean of $\frac{N_L^{j, *}(x, \theta)}{n_\CI \BE[W_1^\CI]} - p_j(\theta, \CJ, \boldsymbol{W}^\CJ)$ given $\theta, \CJ, \boldsymbol{W}^\CJ$ is clearly zero, while the conditional variance satisfies 
\begin{align*}
    &\Var\Big[\frac{N_L^{j,*}(x, \theta)}{n_\CI \BE[W_1^\CI]} - p_j(\theta, \CJ, \boldsymbol{W}^\CJ) \mid \theta, \CJ, \boldsymbol{W}^\CJ \Big]\\
    &= \BE\Big[\frac{N_L^{j,*}(x, \theta)^2}{n_\CI^2 \BE[W_1^\CI]^2} \mid \theta, \CJ, \boldsymbol{W}^\CJ \Big] - p_j(\theta, \CJ, \boldsymbol{W}^\CJ)^2
\end{align*}
and
\begin{align*}
    &\BE\Big[\frac{N_L^{j,*}(x, \theta)^2}{n_\CI^2 \BE[W_1^\CI]^2} \mid \theta, \CJ, \boldsymbol{W}^\CJ \Big]\\
    &= \frac{1}{n_\CI^2 \BE[W_1^\CI]^2}\sum_{i = 1}^{n_\CI}\sum_{k = 1}^{n_\CI} \BE[W_i^\CI W_k^\CI]\BE[\1_{\{X_i^j, X_k^j \in (a_{n, j}^{x,*}, b_{n, j}^{x, *}]\}} \mid \theta, \CJ, \boldsymbol{W}^\CJ] \\
    &= \frac{n_\CI \BE[(W_1^\CI)^2]p_j(\theta, \CJ, \boldsymbol{W}^\CJ) + n_\CI(n_\CI - 1)\BE[W_1^\CI W_2^\CI] p_j(\theta, \CJ, \boldsymbol{W}^\CJ)^2}{n_\CI^2 \BE[W_1^\CI]^2} \\
    &= \frac{\BE[(W_1^\CI)^2]}{\BE[W_1^\CI]} \frac{1}{n_\CI \BE[W_1^\CI]} p_j(\theta, \CJ, \boldsymbol{W}^\CJ)+ \frac{n_\CI - 1}{n_\CI} \frac{\BE[W_1^\CI W_2^\CI]}{\BE[W_1^\CI]^2} p_j(\theta, \CJ, \boldsymbol{W}^\CJ)^2,
\end{align*}
so
\begin{align*}
    &\Var\Big[\frac{N_L^{j,*}(x, \theta)}{n_\CI \BE[W_1^\CI]} - p_j(\theta, \CJ, \boldsymbol{W}^\CJ) \mid \theta, \CJ, \boldsymbol{W}^\CJ \Big] = \frac{\BE[(W_1^\CI)^2]}{\BE[W_1^\CI]^2} \frac{p_j(\theta, \CJ, \boldsymbol{W}^\CJ)}{n_\CI \BE[W_1^\CI]} \\
    &+ p_j(\theta, \CJ, \boldsymbol{W}^\CJ)^2\Big(\frac{\BE[W_1^\CI W_2^\CI]}{\BE[(W_1^\CI)^2]} - 1 \Big) -\frac{1}{n_\CI} \frac{\BE[W_1^\CI W_2^\CI]}{\BE[(W_1^\CI)^2]}p_j(\theta, \CJ, \boldsymbol{W}^\CJ)^2
\end{align*}

which goes to zero by Assumption~\ref{asm:bootstrap} (iii) and the condition $n_\CI \BE[W_1^\CI] \to \infty$. Furthermore, by the fact that $p_j(\theta, \CJ, \boldsymbol{W}^\CJ) \leq 1$ and the tower property, we have $\frac{N_L^{j,*}(x, \theta)}{n_\CI \BE[W_1^\CI]} - p_j(\theta, \CJ, \boldsymbol{W}^\CJ) \overset{\BP}{\to} 0$. We conclude that $p_j(\theta, \CJ, \boldsymbol{W}^\CJ) \overset{\BP}{\to} 0$ if and only if $N_L^{j, *}(x, \theta)/n_\CI \BE[W_1^\CI] \overset{\BP}{\to} 0$. As before,
\begin{equation*}
    p_j(\theta, \CJ, \boldsymbol{W}^\CJ) = \int_{a_{n, j}^{x, *}}^{b_{n, j}^{x, *}} f^j(\widetilde{x})\mathrm{d}\widetilde{x} \in [\varepsilon(b_{n, j}^{x, *} - a_{n, j}^{x, *}), C(b_{n, j}^{x, *} - a_{n, j}^{x, *})]
\end{equation*}
which proves the proposition. 
\end{proof}

\begin{proof}[Proof of Proposition~\ref{prop:regularIIbootstrap}]
The proof is very similar to the one for Proposition~\ref{prop:regularII}. The key difference is in the way Lemma~\ref{lem:Devroyebound} is applied. Define
\begin{equation*}
    \widetilde{N}_L^*(x, \theta) := \sum_{i = 1}^{n_\CI} \1_{\{X_i \in L^*(x, \theta)\}}, \quad p(\theta, \CJ, \boldsymbol{W}^\CJ) := \BP(X_1 \in L^*(x, \theta) \mid \theta, \CJ, \boldsymbol{W}^\CJ).
\end{equation*}
For any $a, b \geq 0$, it holds that $a + b > \delta$ implies $a > \delta/2$ or $b > \delta/2$. From this observation, the triangle inequality and a union bound, we get for any $\delta > 0$ that
\begin{align*}
    \BP\Big(\Big|\frac{N_L^*(x, \theta)}{n_\CI \BE[W_1^\CI]} - p(\theta, \CJ, \boldsymbol{W}^\CJ) \Big|& > \delta \mid \theta, \CJ, \boldsymbol{W}^\CJ \Big) \leq \\
    & \quad \BP\Big(\Big|\frac{N_L^*(x, \theta)}{\BE[W_1^\CI]} - \widetilde{N}_L^*(x, \theta) \Big| > \frac{\delta n_\CI}{2} \mid \theta, \CJ, \boldsymbol{W}^\CJ \Big) \\
    &+ \BP\Big(\Big|\frac{\widetilde{N}_L^*(x, \theta)}{n_\CI} - p(\theta, \CJ, \boldsymbol{W}^\CJ) \Big| > \frac{\delta}{2} \mid \theta, \CJ, \boldsymbol{W}^\CJ \Big).
\end{align*}
Since, conditional on $\theta, \CJ, \boldsymbol{W}^\CJ$, $\widetilde{N}_L^*(x, \theta)/n_\CI$ is the empirical measure of $p(\theta, \CJ, \boldsymbol{W}^\CJ)$, we can again apply Lemma~\ref{lem:Devroyebound} to bound the second term by $cn_\CI^{4d}e^{-n_\CI \delta^2/2}$. As for the first term, we notice that
\begin{equation*}
    \frac{N_L^*(x, \theta)}{\BE[W_1^\CI]} - \widetilde{N}_L^*(x, \theta) = \sum_{i = 1}^{n_\CI} (\widetilde{W}_i^\CI - 1)\1_{\{X_i \in L^*(x, \theta)\}}.
\end{equation*}
This has mean zero, and the conditional second moment is
\begin{align*}
    &\BE\Big[\Big(\sum_{i = 1}^{n_\CI} (\widetilde{W}_i^\CI - 1)\1_{\{X_i \in L^*(x, \theta)\}} \Big)^2 \mid \theta, \CJ, \boldsymbol{W}^\CJ \Big] = n_\CI \BE[(\widetilde{W}_1^\CI - 1)^2]p(\theta, \CJ, \boldsymbol{W}^\CJ) \\
    &+ n_\CI(n_\CI - 1)\BE[(\widetilde{W}_1^\CI - 1)(\widetilde{W}_2^\CI - 1)]p(\theta, \CJ, \boldsymbol{W}^\CJ)^2
\end{align*}
Using the assumptions on the mixed moment, we see that the second term either negative, in which case we can simply disregard it, or it is $O(n_\CI)$. As for the first term,
\begin{equation*}
    \BE[(\widetilde{W}_1^\CI - 1)^2] = \frac{\BE[(W_1^\CI)^2]}{\BE[W_1^\CI]^2 } - 1 \leq \frac{1}{\BE[W_1^\CI]} \frac{\BE[(W_1^\CI)^2]}{\BE[W_1^\CI]}.
\end{equation*}
By a second order Markov inequality,
\begin{align*}
    &\BP\Big(\Big|\frac{N_L^*(x, \theta)}{\BE[W_1^\CI]} - \widetilde{N}_L^*(x, \theta) \Big| > \frac{\delta n_\CI}{2} \mid \theta, \CJ, \boldsymbol{W}^\CJ \Big)\\
    &\leq \frac{4}{\delta^2 n_\CI \BE[W_1^\CI]} \frac{\BE[(W_1^\CI)^2]}{\BE[W_1^\CI]} + \frac{C'}{\delta^2 n_\CI} \\
    &= \frac{1}{\delta^2 n_\CI \BE[W_1^\CI]}\Big(\frac{4\BE[(W_1^\CI)^2]}{\BE[W_1^\CI]} + C'\BE[W_1^\CI] \Big) \\
    &\leq \frac{C''}{\delta^2 n_\CI \BE[W_1^\CI]}
\end{align*}
for some deterministic constants $C'$ and $C''$ (for $n$ large enough) where we have used Assumption~\ref{asm:bootstrap} (iii). The bounds on both terms in the above union bound also hold unconditionally by the tower property. We conclude that with probability at least $1 - cn_\CI^{4d} e^{-n_\CI \delta^2/2} - C''/\delta^2n_\CI\BE[W_1^\CI]$, we have
\begin{equation*}
    p(\theta, \CJ, \boldsymbol{W}^\CJ) \geq \frac{N_L^*(x, \theta)}{n_\CI \BE[W_1^\CI]} - \delta.
\end{equation*}
Choose $\delta = k_n/2n_\CI\BE[W_1^\CI]$. Then
\begin{align*}
    \BP\Big(n_\CI \BE[W_1^\CI]\lambda(L^*(x, \theta)) \geq \frac{k_n}{2C}\Big) &\geq \BP\Big(p(\theta, \CJ, \boldsymbol{W}^\CJ) \geq \frac{k_n}{2n_\CI \BE[W_1^\CI]}\Big) \\
    &\geq \BP\Big(p(\theta, \CJ, \boldsymbol{W}^\CJ) \geq \frac{N_L^*(x, \theta)}{n_\CI \BE[W_1^\CI]} - \frac{k_n}{2n_\CI \BE[W_1^\CI]} \geq \frac{k_n}{2n_\CI \BE[W_1^\CI]}\Big) \\
    &= \BP\Big(p(\theta, \CJ, \boldsymbol{W}^\CJ) \geq \frac{N_L^*(x, \theta)}{n_\CI \BE[W_1^\CI]} - \frac{k_n}{2n_\CI \BE[W_1^\CI]}, N_L^*(x, \theta) \geq k_n\Big) \\
    &\geq a_n^*(x) - cn_\CI^{4d}e^{-k_n^2/8n_\CI\BE[W_1^\CI]^2} - \frac{C''}{k_n^2/4n_\CI\BE[W_1^\CI]} \to 1,
\end{align*}
so by Lemma \ref{lem:BCinfinity}, the proof is complete.
\end{proof}

\begin{proof}[Proof of Proposition~\ref{prop:densityestimatorstrongconvergencebootstrap}]
The proof is identical to the one for Theorem~\ref{thm:strongconsistencybootstrap}, just let $Y \equiv 1$. The only thing one has to note is that the condition for the moment generating function of $W_1^\CI Y_1$ collapses to
\begin{equation*}
    \frac{\kappa_{W_1^\CI}(t) - 1}{\BE[W_1^\CI]} \frac{1}{\lambda(L^*(x, \theta))} \int_{L^*(x, \theta)} f(\widetilde{x})\mathrm{d}\widetilde{x} = \frac{\kappa_{W_1^\CI}(t) - 1}{\BE[W_1^\CI]}f(x) + o(1) \quad \BP\text{-a.s.}
\end{equation*}
which is satisfied by the boundedness assumption provided in the proposition.
\end{proof}

\begin{proof}[Proof of Theorem~\ref{thm:strongconsistencybootstrap}]
We show strong convergence of $\widehat{m}^*(x, \theta)$. We have
\begin{equation*}
    \BE[\widehat{m}^*(x, \theta) \mid \theta, \CJ , \boldsymbol{W}^\CJ] = \frac{1}{\lambda(L^*(x, \theta))} \int_{L^*(x, \theta)} m(\widetilde{x})\mathrm{d}\widetilde{x} \to m(x) \quad \BP\text{-a.s.},
\end{equation*}
so we need only show $\BE[\widehat{m}^*(x, \theta) \mid \theta, \CJ , \boldsymbol{W}^\CJ] - \widehat{m}^*(x, \theta) \to 0$ a.s. We consider the cases (i) and (ii) separately, starting with (i). To simplify the calculations and ease the notation, we generalise to the setup of iid $Z_1, \dots, Z_n$ with finite fourth moment independent of the exchangeable weights $W_1, \dots, W_n$ satisfying $W_1 + \cdots + W_n = n$ and $\BE[W_1^4] < \infty$. We may write
\begin{equation*}
    \sum_{i = 1}^n (W_i Y_i - \BE[Y_1]) = \sum_{i = 1}^n W_i(Y_i - \BE[Y_1]) + \BE[Y_1]\sum_{i = 1}^n (W_i - 1)
\end{equation*}
where the latter sum equals zero by assumption. Conditional on the weights, the first sum is a sum of independent mean zero variables, and the fourth moment calculation thus simplifies to
\begin{align*}
    \BE\Big[\Big(\sum_{i = 1}^n (W_i Y_i - \BE[Y_1])\Big)^4 \Big] &= \BE\Big[\BE\Big[ \Big(\sum_{i = 1}^n W_i(Y_i - \BE[Y_1]) \Big)^4 \mid W_1, \dots, W_n \Big] \Big] \\
    &= \BE\Big[\BE[(Y_1 - \BE[Y_1])^4]\sum_{i = 1}^n W_i^4 + 3\BE[(Y_1 - \BE[Y_1])^2]^2\sum_{i \neq j}W_i^2W_j^2 \Big] \\
    &= n\BE[(Y_1 - \BE[Y_1])^4] \BE[W_1^4] + 3n(n - 1)\Var(Y_1)^2 \BE[W_1^2W_2^2].
\end{align*}
We can now apply this formula to $\widehat{m}^*(x, \theta) - \BE[\widehat{m}^*(x, \theta) \mid \theta, \CJ, \boldsymbol{W}^\CJ]$ by noting that
\begin{equation*}
    \widehat{m}^*(x, \theta) - \BE[\widehat{m}^*(x, \theta) \mid \theta, \CJ, \boldsymbol{W}^\CJ] = \frac{1}{n_\CI \lambda(L^*(x, \theta))}\sum_{i = 1}^{n_\CI}(\widetilde{W}_i^\CI \1_{\{X_i \in L^*(x, \theta)\}}Y_i - \BE[\1_{\{X_1 \in L^*(x, \theta)\}}Y_1 \mid \theta, \CJ, \boldsymbol{W}^\CJ]).
\end{equation*}
This yields
\begin{align*}
    \BE[(\widehat{m}^*(x, \theta) &- \BE[\widehat{m}^*(x, \theta) \mid \theta, \CJ, \boldsymbol{W}^\CJ])^4 \mid \theta, \CJ, \boldsymbol{W}^\CJ] = \\
    &\frac{\BE[(\widetilde{W}_1^\CI)^4] \BE[(\1_{\{X_1 \in L^*(x, \theta)\}}Y_1 - \BE[1_{\{X_1 \in L^*(x, \theta)\}}Y_1 \mid \theta, \CJ, \boldsymbol{W}^\CJ])^4 \mid \theta, \CJ, \boldsymbol{W}^\CJ]}{n_\CI^3 \lambda(L^*(x, \theta))^4} + \\
    &3\frac{n_\CI - 1}{n_\CI} \frac{\BE[(\widetilde{W}_1^\CI)^2(\widetilde{W}_2^\CI)^2] \Var(\1_{\{X_1 \in L^*(x, \theta)\}}Y_1 \mid \theta, \CJ, \boldsymbol{W}^\CJ)^2}{n_\CI^2 \lambda(L^*(x, \theta))^4} = \\
    &\frac{\BE[(\widetilde{W}_1^\CI)^4](\BE[Y^4 \mid X = x] + o_\BP(1))}{n_\CI^3 \lambda(L^*(x, \theta))^3} + 3\frac{n_\CI - 1}{n_\CI} \frac{\BE[(\widetilde{W}_1^\CI)^2(\widetilde{W}_2^\CI)^2](\BE[Y^2 \mid  X= x]^2 + o_\BP(1))}{n_\CI^2 \lambda(L^*(x, \theta))^2}
\end{align*}
which is summable by assumption. By a fourth order Markov inequality, we conclude that $\widehat{m}^*(x, \theta) - \BE[\widehat{m}^*(x, \theta) \mid \theta, \CJ, \boldsymbol{W}^\CJ] \to 0$ $\BP$-a.s. as desired. By repeating the proof for $Y \equiv 1$, it holds that $\widehat{f}^*(x, \theta) \to f(x)$ $\BP$-a.s. and the continuous mapping theorem completes the proof. \\

Now for the set of assumptions in (ii). As in the non-bootstrap case, we may combine a union bound and a Chernoff bound to get
\begin{align*}
    &\BP(|\widehat{m}^*(x, \theta)  - \BE[\widehat{m}^*(x, \theta) \mid \theta, \CJ , \boldsymbol{W}^\CJ]| > \delta \mid \theta, \CJ, \boldsymbol{W}^\CJ) \leq \\
    &e^{-\delta t n_\CI \BE[W_1^\CI] \lambda(L^*(x, \theta)) + \Lambda_{n_\CI}^*(t \mid \theta, \CJ, \boldsymbol{W}^\CJ)} + e^{-\delta t n_\CI \BE[W_1^\CI] \lambda(L^*(x, \theta)) + \Lambda_{n_\CI}^*(-t \mid \theta, \CJ, \boldsymbol{W}^\CJ)}
\end{align*}
where 
\begin{align*}
    &\Lambda_{n_\CI}(t \mid \theta, \CJ, \boldsymbol{W}^\CJ):= \log \BE\Big[\exp\Big(t\sum_{i = 1}^{n_\CI} (\1_{\{X_i \in L^*(x, \theta)\}}Y_i W_i^\CI - \BE[\1_{\{X_1 \in L^*(x, \theta)\}}Y_1 W_1^\CI \mid \theta, \CJ, \boldsymbol{W}^\CJ]) \Big) \mid \theta, \CJ, \boldsymbol{W}^\CJ \Big].
\end{align*}
Redoing the computations in the proof of Theorem~\ref{thm:strongconsistency} with the assumption on the moment generating function of $X_1^\CI Y_1 \mid X_1 = x$ yields that a.s.,
\begin{align*}
    &\Lambda_n(t \mid \theta, \CJ, \boldsymbol{W}^\CJ) = -tn_\CI \BE[W_1^\CI] \lambda(L^*(x, \theta))(m(x) + o(1)) \\
    &+ n_\CI \log\Big(1 - p(\theta, \CJ, \boldsymbol{W}^\CJ) + \BE\Big[\1_{\{X_1 \in L^*(x, \theta)\}}e^{tY_1 W_1^\CJ} \mid \theta, \CJ, \boldsymbol{W}^\CJ \Big] \Big) \\
    &= -tn_\CI \BE[W_1^\CI] \lambda(L^*(x, \theta))(m(x) + o(1)) \\
    &+ n_\CI\log\Big(1 + \lambda(L^*(x, \theta))\frac{1}{\lambda(L^*(x, \theta))}\int_{L^*(x, \theta)} (\BE[e^{tY_1W_1^\CI} \mid X_1 = \widetilde{x}] - 1)f(\widetilde{x})\mathrm{d}\widetilde{x} \Big) \\
    &= -tn_\CI \BE[W_1^\CI] \lambda(L^*(x, \theta))(m(x) + o(1)) \\
    &+ n_\CI\log\Big(1 + \lambda(L^*(x, \theta))\BE[W_1^\CI]\Big(\frac{\BE[e^{tW_1^\CI Y_1} \mid X_1 = x] - 1}{\BE[W_1^\CI]}f(x) + o(1) \Big) \Big),
\end{align*}
where we used the assumption on the moment generating function of $Y_1W_1^\CI$. Now apply a Taylor expansion to get
\begin{align*}
    \Lambda_n&(t \mid \theta, \CJ, \boldsymbol{W}^\CJ) = -tn_\CI \BE[W_1^\CI] \lambda(L^*(x, \theta))(m(x) + o(1)) \\
    &+ n_\CI \BE[W_1^\CI] \lambda(L^*(x, \theta))\Big(\frac{\BE[e^{tW_1^\CI Y_1} \mid X_1 = x] - 1}{\BE[W_1^\CI]}f(x) + o(1) \Big).
\end{align*}
Disregarding $n_\CI \BE[W_1^\CI]\lambda(L^*(x, \theta))$ and the $o(1)$ terms, the two exponents in the union bound are
\begin{align*}
    E_1(t) &:= -\delta t - tm(x) + \frac{\BE[e^{tW_1^\CI Y_1} \mid X_1 = x] - 1}{\BE[W_1^\CI]}f(x) \\
    E_2(t) &:= -\delta t + tm(x) + \frac{\BE[e^{-tW_1^\CI Y_1} \mid X_1 = x] - 1}{\BE[W_1^\CI]}f(x).
\end{align*}
We have $E_1(0) = E_2(0) = 0$ and by independence of $W_1^\CI$ and $Y_1$ given $X_1$,
\begin{equation*}
    E_1'(0) = -\delta - m(x) + \frac{\BE[W_1^\CI Y_1 \mid X_1 = x]}{\BE[W_1^\CI]}f(x) = -\delta - m(x) + m(x) = -\delta < 0
\end{equation*}
and similarly, $E_2'(0) = -\delta$. The proof for almost sure convergence of $\widehat{m}^*(x, \theta)$ is now complete by the same argument as in the proof of Theorem~\ref{thm:strongconsistency}. Since we also have $\widehat{f}^*(x, \theta) \to f(x)$ a.s. by Proposition~\ref{prop:densityestimatorstrongconvergencebootstrap}, an application of the continuous mapping theorem concludes the proof.
\end{proof}

\begin{proof}[Proof of Lemma~\ref{lem:multinomialMGF}]
We have
\begin{equation*}
    \kappa_{W_1^\CI}(t) = \frac{1}{n_\CI^{m_n}}(e^t + n_\CI - 1)^{m_n}.
\end{equation*}
so elementary calculations yield
\begin{equation*}
    \kappa_{W_1^\CI}''(t) = \frac{m_n}{n_\CI^{m_n}}e^t (e^t + n_\CI - 1)^{m_n - 2} \Big(e^t + n_\CI - 1 + (m_n - 1)e^t \Big).
\end{equation*}
Now apply a second order Taylor expansion with remainder term around zero. For some $\xi$ between $0$ and $t$, it holds that
\begin{equation*}
    \kappa_{W_1^\CI}(t) = 1 + \BE[W_1^\CI]t + \frac{t^2}{2} \kappa_{W_1^\CI}''(\xi)
\end{equation*}
so that
\begin{align*}
    \frac{\kappa_{W_1^\CI}(t) - 1}{\BE[W_1^\CI]} &= t + \frac{t^2}{2} \frac{1}{n_\CI^{m_n - 1}}e^\xi (e^\xi + n_\CI - 1)^{m_n - 2}(n_\CI - 1 + m_ne^\xi) \\
    &= t + \frac{t^2}{2} e^\xi \Big(1 + \frac{e^\xi - 1}{n_\CI} \Big)^{m_n - 2}\Big(1 + \frac{m_ne^\xi - 1}{n_\CI}\Big).
\end{align*}
In any case, $\xi < |t|$ and $m_n \leq n_\CI$, which implies
\begin{align*}
    \frac{\kappa_{W_1^\CI}(t) - 1}{\BE[W_1^\CI]} &\leq t + \frac{t^2}{2}e^{|t|} \Big(1 + \frac{e^{|t|} - 1}{n_\CI} \Big)^{n_\CI}\Big(1 + e^{|t|} - \frac{1}{n_\CI} \Big) \\
    &\to t + \frac{t^2}{2}e^{|t| + e^{|t|} - 1}(1 + e^{|t|})
\end{align*}
and so the requirement in Proposition~\ref{prop:densityestimatorstrongconvergencebootstrap} is satisfied. 
\end{proof}

\begin{proof}[Proof of Lemma~\ref{lem:regularstrongconsistencyI}]
Recall that we defined
\begin{equation*}
    N_L^{j, *}(x, \theta) = \sum_{i = 1}^{n_\CI} \1_{\{X_i^j \in (a_{n, j}^{x, *}, b_{n, j}^{x, *}]\}}W_i^\CI, \quad p_j(\theta, \CJ, \boldsymbol{W}^\CJ) = \BP(X_1^j \in (a_{n, j}^{x, *}, b_{n, j}^{x, *}] \mid \theta, \CJ, \boldsymbol{W}^\CJ).
\end{equation*}
We show that
\begin{equation*}
    \frac{N_L^{j, *}(x, \theta)}{n_\CI \BE[W_1^\CI]} - p_j(\theta, \CJ, \boldsymbol{W}^\CJ) \to 0 \quad \BP\text{-a.s.}
\end{equation*}
We have
\begin{equation*}
    \frac{N_L^{j, *}(x, \theta)}{n_\CI \BE[W_1^\CI]} - p_j(\theta, \CJ, \boldsymbol{W}^\CJ) = \frac{1}{n_\CI}\sum_{i = 1}^{n_\CI}\Big(\1_{\{X_i^j \in (a_{n, j}^{x, *}, b_{n, j}^{x, *}]\}}\widetilde{W}_i^\CI - p_j(\theta, \CJ, \boldsymbol{W}^\CJ) \Big)
\end{equation*}
and applying a fourth order Markov inequality yields for any $\delta > 0$ that
\begin{align*}
    \BP\Big(&\Big|\frac{N_L^{j, *}(x, \theta)}{n_\CI \BE[W_1^\CI]} - p_j(\theta, \CJ, \boldsymbol{W}^\CJ) \Big| > \delta \mid \theta, \CJ, \boldsymbol{W}^\CJ \Big) \leq \\
    &\frac{\BE\Big[\Big(\sum_{i = 1}^{n_\CI} \Big(\1_{\{X_i^j \in (a_{n, j}^{x, *}, b_{n, j}^{x, *}]\}} \widetilde{W}_i^\CI - p_j(\theta, \CJ, \boldsymbol{W}^\CJ)\Big) \Big)^4 \mid \theta, \CJ, \boldsymbol{W}^\CJ \Big]}{\delta^4 n_\CI^4}.
\end{align*}
The terms in the sum are conditionally iid given $\theta, \CJ, \boldsymbol{W}^\CJ$, so we have
\begin{align*}
    \BE&\Big[\Big(\sum_{i = 1}^{n_\CI} \Big(\1_{\{X_i^j \in (a_{n, j}^{x, *}, b_{n, j}^{x, *}]\}} \widetilde{W}_i^\CI - p_j(\theta, \CJ, \boldsymbol{W}^\CJ)\Big) \Big)^4 \mid \theta, \CJ, \boldsymbol{W}^\CJ \Big] = \\
    &n_\CI \BE\Big[\Big(\1_{\{X_1^j \in (a_{n, j}^{x, *}, b_{n, j}^{x, *}]\}} \widetilde{W}_1^\CI - p_j(\theta, \CJ, \boldsymbol{W}^\CJ)\Big)^4\Big]\\
    &\quad+ 3n_\CI(n_\CI - 1)\Var\Big[\1_{X_1^j \in (a_{n, j}^{x, *}, b_{n, j}^{x, *}]}\widetilde{W}_1^\CI \mid \theta, \CJ, \boldsymbol{W}^\CJ \Big]^2.
\end{align*}
We upper bound the fourth moment by first applying Minkowski's inequality to obtain
\begin{align*}
    \BE&\Big[\Big(\1_{\{X_1^j \in (a_{n, j}^{x, *}, b_{n, j}^{x, *}]\}} \widetilde{W}_1^\CI - p_j(\theta, \CJ, \boldsymbol{W}^\CJ)\Big)^4\Big]^{1/4} \leq \\
    &\BE\Big[\1_{\{X_1^j \in (a_{n, j}^{x, *}, b_{n, j}^{x, *}]\}}(\widetilde{W}_1^\CI)^4 \mid \theta, \CJ, \boldsymbol{W}^\CJ\Big]^{1/4} + p_j(\theta, \CJ, \boldsymbol{W}^\CJ) \leq \\
    &p_j(\theta, \CJ, \boldsymbol{W}^\CJ)^{1/4} \BE[(\widetilde{W}_1^\CI)^4]^{1/4} + p_j(\theta, \CJ, \boldsymbol{W}^\CJ).
\end{align*}
Since $p_j(\theta, \CJ, \boldsymbol{W}^\CJ) \leq 1$, $p_j(\theta, \CJ, \boldsymbol{W}^\CJ) \leq p_j(\theta, \CJ, \boldsymbol{W}^\CJ)^{1/4}$, and since $L^p$-norms are increasing in $p$, $\BE[(\widetilde{W}_1^\CI)^4] \geq \BE[\widetilde{W}_1^\CI]^4 = 1$, so
\begin{align*}
    \BE&\Big[\Big(\1_{\{X_1^j \in (a_{n, j}^{x, *}, b_{n, j}^{x, *}]\}} \widetilde{W}_1^\CI - p_j(\theta, \CJ, \boldsymbol{W}^\CJ)\Big)^4\Big] \leq 16 p_j(\theta, \CJ, \boldsymbol{W}^\CJ)\BE[(\widetilde{W}_1^\CI)^4].
\end{align*}
As for the variance term,
\begin{align*}
    \Var\Big[\1_{\{X_1^j \in (a_{n, j}^{x, *}, b_{n, j}^{x, *}]\}}\widetilde{W}_1^\CI \mid \theta, \CJ, \boldsymbol{W}^\CJ \Big] &= p_j(\theta, \CJ, \boldsymbol{W}^\CJ)\BE[(\widetilde{W}_1^\CI)^2]\\ &\quad- p_j(\theta, \CJ, \boldsymbol{W}^\CJ)^2 \BE[W_1^\CI]^2
\end{align*}
implying
\begin{align*}
    \BP\Big(&\Big|\frac{N_L^{j, *}(x, \theta)}{n_\CI \BE[W_1^\CI]} - p_j(\theta, \CJ, \boldsymbol{W}^\CJ) \Big| > \delta \mid \theta, \CJ, \boldsymbol{W}^\CJ \Big) \leq \\
    &\frac{16p_j(\theta, \CJ, \boldsymbol{W}^\CJ) \BE[(\widetilde{W}_1^\CI)^4]}{\delta^4 n_\CI^3} + \frac{3p_j(\theta, \CJ, \boldsymbol{W}^\CJ)^2 \BE[(\widetilde{W}_1^\CI)^2]^2}{\delta^4 n_\CI^2}.
\end{align*}
By applying the bound $p_j(\theta, \CJ, \boldsymbol{W}^\CJ) \leq 1$, we may remove the conditioning by the tower property and obtain the desired statement by the summability assumptions in the lemma. The rest of the proof is the same as the final arguments in the proof of Proposition~\ref{prop:regularIbootstrap}. 
\end{proof}

\begin{proof}[Proof of Lemma~\ref{lem:regularstrongconsistencyII}]
To see that (II$^*$) holds almost surely, we use the same notation and setup as in the proof of Proposition~\ref{prop:regularIIbootstrap}. From that proof, we have for any $\delta > 0$ that
\begin{align*}
    &\BP\Big(\Big|\frac{N_L^*(x, \theta)}{n_\CI \BE[W_1^\CI]} - p(\theta, \CJ, \boldsymbol{W}^\CJ) \Big| > \delta \mid \theta, \CJ, \boldsymbol{W}^\CJ \Big) \leq \\
    &\BP\Big(\Big|\sum_{i = 1}^{n_\CI} (\widetilde{W}_i^\CI - 1)\1_{\{X_i \in L^*(x, \theta)\}} \Big| > \frac{\delta n_\CI}{2} \mid \theta, \CJ, \boldsymbol{W}^\CJ \Big) + cn_\CI^{4d}e^{-n_\CI \delta^2/2}.
\end{align*}
We now bound the first term by a fourth order Markov inequality instead of a second order one. By very similar calculations as those above, we get
\begin{align*}
    \BP\Big(&\Big|\sum_{i = 1}^{n_\CI} (\widetilde{W}_i^\CI - 1)\1_{\{X_i \in L^*(x, \theta)\}} \Big| > \frac{\delta n_\CI}{2} \mid \theta, \CJ, \boldsymbol{W}^\CJ \Big) \leq \\
    & \frac{n_\CI p(\theta, \CJ, \boldsymbol{W}^\CJ) \BE[(\widetilde{W}_1^\CI - 1)^4] + 3n_\CI(n_\CI - 1)\BE[(\widetilde{W}_1^\CI - 1)^2]^2 p(\theta, \CJ, \boldsymbol{W}^\CJ)^2}{\delta^4 n_\CI^4/16} = \\
    & 16 p(\theta, \CJ, \boldsymbol{W}^\CJ)\frac{\BE[(\widetilde{W}_1^\CI - 1)^4]}{\delta^4 n_\CI^3} + 48 p(\theta, \CJ, \boldsymbol{W}^\CJ) \frac{(n_\CI - 1)\BE[(\widetilde{W}_1^\CI - 1)^2]^2}{\delta^4 n_\CI^3},
\end{align*}
and so by the tower property,
\begin{align*}
    \BP\Big(\Big|\frac{N_L^*(x, \theta)}{n_\CI \BE[W_1^\CI]} - p(\theta, \CJ, \boldsymbol{W}^\CJ) \Big| \leq \delta \Big) &\geq 1 - cn_\CI^{4d}e^{-n_\CI \delta^2/2} - 16\frac{\BE[(\widetilde{W}_1^\CI - 1)^4]}{\delta^4 n_\CI^3} \\ &-48 \frac{\BE[(\widetilde{W}_1^\CI - 1)^2]^2}{\delta^4 n_\CI^2}.
\end{align*}
Again, let $\delta = k_n/(2n_\CI \BE[W_1^\CI])$ and copy the calculations in the proof of Proposition \ref{prop:regularIIbootstrap} to get
\begin{align*}
    \BP\Big(n_\CI \BE[W_1^\CI]\lambda(L^*(x, \theta)) \geq \frac{k_n}{2C} \Big) &\ge a_n^*(x) - cn_\CI^{4d}e^{-k_n^2/8n_\CI \BE[W_1^\CI]^2}\\
    &- 256\frac{\BE[(W_1^\CI - \BE[W_1^\CI])^4]}{k_n^4}n_\CI -768 \frac{\BE[(W_1^\CI - \BE[W_1^\CI])^2]^2}{k_n^4}n_\CI^2. 
\end{align*}
Now simply apply the summability assumptions in the lemma along with Lemma~\ref{lem:BCinfinity} to conclude the proof.
\end{proof}

\begin{proof}[Proof of Theorem~\ref{thm:regularstrongconsistency}]
By the two lemmata above the theorem, we have that conditions (I$^*$) and (II$^*$) of Assumption~\ref{asm:bootstrap} both hold almost surely. It only remains to consider the two sets of assumptions (i) and (ii). From the proof of Lemma~\ref{lem:regularstrongconsistencyII}, we have that
\begin{align*}
    \BP\Big(n_\CI \BE[W_1^\CI]\lambda(L^*(x, \theta)) \geq \frac{k_n}{2C} \Big) &\geq a_n^*(x) - cn_\CI^{4d}e^{-k_n^2/8n_\CI \BE[W_1^\CI]^2}\\
    &- 256\frac{\BE[(W_1^\CI - \BE[W_1^\CI])^4]}{k_n^4}n_\CI -768 \frac{\BE[(W_1^\CI - \BE[W_1^\CI])^2]^2}{k_n^4}n_\CI^2. 
\end{align*}
To ease notation, define the set
\begin{equation*}
    A_n := \Big\{n_\CI \BE[W_1^\CI]\lambda(L^*(x, \theta)) \geq \frac{k_n}{2C} \Big\}.
\end{equation*}
Now suppose the requirements of (i) hold. Then 
\begin{align*}
    \frac{\BE[(\widetilde{W}_1^\CI)^4]}{n_\CI^3 \lambda(L^*(x, \theta))^3} + \frac{\BE[(\widetilde{W}_1^\CI)^2(\widetilde{W}_2^\CI)^2]}{n_\CI^2 \lambda(L^*(x, \theta))^2} &\leq \1_{A_n(x)}\left(\frac{8C^3 \BE[(W_1^\CI)^4]}{\BE[W_1^\CI]k_n^3} + \frac{4C^2 \BE[(W_1^\CI)^2(W_2^\CI)^2]}{k_n^2 \BE[W_1^\CI]^2}\right) \\
    &+ \1_{A_n(x)^c}\left(\frac{\BE[(\widetilde{W}_1^\CI)^4]}{n_\CI^3 \lambda(L^*(x, \theta))^3} + \frac{\BE[(\widetilde{W}_1^\CI)^2(\widetilde{W}_2^\CI)^2]}{n_\CI^2 \lambda(L^*(x, \theta))^2} \right).
\end{align*}
Since $\1_{A_n(x)}$ is eventually one with probability one, the summability requirement of Theorem \ref{thm:strongconsistencybootstrap} (i) holds, and the bootstrap tree is strongly consistent. Now for (ii). For all $\delta > 0$,
\begin{equation*}
    \sum_{n_\CI = 1}^\infty \exp\Big(-\delta n_\CI \BE[W_1^\CI]\lambda(L^*(x, \theta)) \Big) < \infty \quad \BP\text{-a.s.}
\end{equation*}

then
\begin{equation*}
    \exp\Big(-\delta n_\CI \BE[W_1^\CI]\lambda(L^*(x, \theta)) \Big) \leq \exp\Big(-\frac{\delta k_n}{2C} \Big) \1_{A_n} + \1_{A_n^c} \leq \exp\Big(-\frac{\delta k_n}{2C}\Big) + \1_{A_n^c}. 
\end{equation*}
Again by the summability requirements from the two previous lemmata and the additional summability requirement for $k_n$, the proof is complete by Theorem~\ref{thm:strongconsistencybootstrap}.
\end{proof}

\section{Discussion and further work}\label{sec:disc}

In this paper, we have studied consistency of honest trees and random forests for both general splitting criteria and some concrete splitting schemes found in the literature. The analysis focuses on individual trees using elementary and transparent methods. Since random forest predictions work by first computing the estimator of interest for each tree and then aggregating, this opens up many possibilities for analysing more complicated estimators. Examples include quantiles (see e.g. \cite{Meinshausen} and \cite{Zhangetal} for studies in quantile regression forests and their applications in computing prediction intervals) and survival forests (see \cite{RSF} for a description of this algorithm and \cite{RSFcon} for a proof of consistency for discrete variables). Some extensions are of interest, in particular conditions for uniform consistency and $L^p$-consistency of honest bootstrap trees. Many other avenues remain to be explored in the above presented framework. First of all, it is of interest to analyse more concrete splitting schemes presented in the literature, a noteworthy example being the quantile splitting scheme as presented in e.g. \cite{RFSubsampling}. This scheme is data-adaptive while still being sufficiently simple to allow for mathematical analysis, making it relevant to study in our setup. One of the schemes presented in this work, namely the uniform splitting scheme, is still of interest to study further, more precisely, to find under which conditions (if any) one can conclude strong or uniform consistency. An interesting and practically relevant route to take from here is to consider asymptotic normality, taking inspiration from the classical proofs for kernel estimators as presented in e.g. \cite{Ghosh}. This avenue has been partially explored by the authors, although this obvious starting point does not seem to yield fully satisfactory conclusions at the time of writing. Many works have already investigated asymptotic normality for random forests, see e.g. \cite{WagerAthey} and \cite{PengColemanMentch}, but none of them have, as far as we know, attempted to tackle the problem on an individual tree basis. 

%

\begin{funding}
The research of both authors was supported by the Carlsberg Foundation, grant CF23-1096. The research of both authors was carried out within the project framework InterAct.
\end{funding}



\newpage
\bibliographystyle{imsart-number} 
\bibliography{main.bib}       

@article{BreimanRF,
    author = {L. Breiman},
    title = {Random Forests},
    journal = {Machine Learning},
    volume = {45},
    date = {2001},
    pages = {5-32}
}

@book{CART,
    author = {L. Breiman and J. Friedman and R. A. Olshen and C. J. Stone},
    title = {Classification and Regression Trees},
    publisher = {Chapman and Hall/CRC},
    date = {1984},
    isbn = {9781315139470},
    edition = {1}
}

@article{BreimanBagging,
    author = {L. Breiman},
    title = {Bagging Predictors},
    journal = {Machine Learning},
    volume = {24},
    date = {1996},
    pages = {123-140}
}

@article{CattaneoChandakKlusowski,
    author = {M. D. Cattaneo and R. Chandak and J. M. Klusowski},
    title = {Convergence Rates of Oblique Regression Trees for Flexible Function Libraries},
    journal = {The Annals of Statistics},
    volume = {52},
    number = {2},
    date = {2024},
    pages = {466-490}
}

@article{Efron,
    author = {B. Efron},
    title = {Bootstrap Methods: Another Look at the Jackknife},
    journal = {The Annals of Statistics},
    volume = {7},
    date = {1979},
    pages = {1-26}
}

@article{Klusowski,
    author = {J. M. Klusowski},
    title = {Sharp Analysis of a Simple Model for Random Forests},
    journal = {Proceedings of The 24th International Conference on Artificial Intelligence and Statistics},
    volume = {130},
    date = {2021},
    pages = {757-765}
}

@article{ScornetKernel,
    author = {E. Scornet},
    title = {Random forests and kernel methods},
    journal = {IEEE Transactions on Information Theory},
    volume = {62},
    number = {3},
    date = {2016},
    pages = {1485-1500}
}

@article{ScornetAsymp,
    author = {E. Scornet},
    title = {On the asymptotics of random forests},
    journal = {Journal of Multivariate Analysis},
    volume = {146},
    date = {2016},
    pages = {72-83}
}

@book{Ghosh,
    author = {S. Ghosh},
    title = {Kernel Smoothing: Principles, methods and applications},
    publisher = {John Wiley \& Sons Inc.},
    date = {2018},
    isbn = {978-1-118-45605-7},
    edition = {1}
}

@article{RSF,
    author = {H. Ishwaran and U. B. Kogalur and E. H. Blackstone and M. S. Lauer},
    title = {Random Survival Forests},
    journal = {The Annals of Applied Statistics},
    volume = {2},
    number = {3},
    date = {2008},
    pages = {841-860}
}

@article{RSFcon,
    author = {H. Ishwaran and U. B. Kogalur},
    title = {Consistency of random survival forests},
    date = {2010},
    journal = {Statistics and Probability Letters},
    number = {80},
    pages = {1056-1064}
}

@article{WagerAthey,
    author = {S. Wager and S. Athey},
    title = {Estimation and Inference of Heterogeneous Treatment Effects using Random Forests},
    journal = {Journal of the American Statistical Association},
    volume = {113},
    number = {523},
    date = {2018},
    pages = {1228–1242}
}

@book{Rudin,
    author = {W. Rudin},
    title = {Real and complex analysis},
    publisher = {McGraw-Hill Book Co.},
    date = {1987},
    isbn = {0-07-054234-1},
    edition = {3}
}

@book{Bogachev,
    author = {V. I. Bogachev},
    title = {Measure Theory},
    publisher = {Springer Berlin, Heidelberg},
    date = {2007},
    isbn = {978-3-540-34514-5},
    volume = {1},
    edition = {1}
}

@book{VaartWellner,
    author = {A. W. van der Vaart and J. A. Wellner},
    title = {Weak Convergence and Empirical Processes},
    publisher = {Springer},
    date = {2022},
    isbn = {978-3-031-29038-1}
}

@book{Pattern,
    author = {L. Devroye and L. Györfi and G. Lugosi},
    title = {A Probabilistic Theory of Pattern Recognition},
    publisher = {Springer-Verlag New York Inc.},
    date = {1997},
    isbn = {0-387-94618-7}
}

@article{BiauDevroyeLugosi2008,
  author = {G. Biau and L. Devroye and G. Lugosi},
  title = {Consistency of Random Forests and Other Averaging Classifiers},
  journal = {Journal of Machine Learning Research},
  volume = {9},
  date = {2008},
  pages = {2015--2033}
}

@article{Devroye1982,
    author = {L. Devroye},
    title = {Bounds for the uniform deviation of empirical measures},
    journal = {Journal of Multivariate Analysis},
    number = {12},
    date = {1982},
    pages = {72-79}
}

@article{RFSubsampling,
    author = {R. Duroux and E. Scornet},
    title = {Impact of subsampling and tree depth on random forests},
    journal = {ESAIM: Probability and Statistics},
    volume = {22},
    date = {2018},
    pages = {96-128}
}

@article{Wu,
    author = {C. F. J. Wu},
    title = {Jackknife, Bootstrap and Other Resampling Methods in Regression Analysis},
    journal = {The Annals of Statistics},
    volume = {14},
    number = {4},
    date = {1986},
    pages = {1261-1295}
}

@article{Ouimet,
    author = {Ouimet, Frédéric},
    title = {General Formulas for the Central and Non-Central Moments of the Multinomial Distribution},
    journal = {Stats},
    volume = {4},
    year = {2021},
    number = {1},
    pages = {18--27},
    doi = {10.3390/stats4010002}
}

@article{Meinshausen,
    author = {N. Meinshausen},
    title = {Quantile Regression Forests},
    journal = {Journal of Machine Learning Research},
    volume = {7},
    number = {},
    year = {2006},
    pages = {983-999}
}

@article{Zhangetal,
    author = {H. Zhang and J. Zimmerman and D. Nettleton and D. J. Nordman},
    title = {Random Forest Prediction Intervals},
    journal = {The American Statistician},
    volume = {74},
    number = {4},
    year = {2020},
    pages = {392-406}
}

@article{LinJeon2006,
  author = {Y. Lin and Y. Jeon},
  title = {Random Forests and Adaptive Nearest Neighbors},
  journal = {Journal of the American Statistical Association},
  volume = {101},
  number = {474},
  date = {2006},
  pages = {578-590}
}

@article{PengColemanMentch,
    author = {W. Peng and T. Coleman and L. Mentch},
    title = {Rates of convergence for random forests
via generalized U-statistics},
    journal = {Electronic Journal of Statistics},
    volume = {16},
    number = {},
    year = {2022},
    pages = {232-292}
}

@article{Genuer2012,
  author = {R. Genuer},
  title = {Variance Reduction in Purely Random Forests},
  journal = {Journal of Nonparametric Statistics},
  volume = {24},
  number = {3},
  date = {2012},
  pages = {543--562}
}

@article{Biau2012,
    author = {G. Biau},
    title = {Analysis of a Random Forest Model},
    journal = {Journal of Machine Learning Research},
    volume = {13},
    number = {},
    date = {2012},
    pages = {1063-1095}
}

@article{Arlot2014, 
author = {S. Arlot and R. Genuer}, 
title = {Analysis of purely random forests bias}, journal = {arXiv:1407.3939}, 
volume = {}, number = {}, date = {2014}, pages = {}, url = {https://arxiv.org/abs/1407.3939} }

@article{AtheyImbens2019,
    author = {S. Athey and G. Imbens},
    title = {Recursive partitioning for heterogeneous causal effects},
    journal = {Proceedings of the National Academy of Sciences},
    volume = {113},
    number = {27},
    date = {2016},
    pages = {7353-7360}
}

@article{MentchHooker,
    author = {L. Mentch and G. Hooker},
    title = { Quantifying Uncertainty in Random Forests via Confidence Intervals and Hypothesis Tests },
    journal = {Journal of Machine Learning Research},
    volume = {17},
    number = {26},
    date = {2016},
    pages = {1-41}
}

@article{WagerHastieEfron,
    author = {S. Wager and T. Hastie and B. Efron},
    title = {Confidence Intervals for Random Forests: The Jackknife and the Infinitesimal Jackknife},
    journal = {Journal of Machine Learning Research},
    volume = {15},
    number = {},
    date = {2014},
    pages = {1625-1651}
}

@article{Talagrand1994,
    author = {M. Talagrand},
    title = {Sharper Bounds for Gaussian and Empirical Processes},
    journal = {Annals of Probability},
    volume = {22},
    number = {1},
    date = {1994},
    pages = {28-76}
}


\newpage

\appendix

\section{Additional results}\label{sec:moreresults}

\subsection{A concentration inequality}

This appendix contains additional helpful results not included in the main text. Let $(\Omega, \CF, \BP)$ denote a fixed probability space. \\

The following lemma is used to show that the volume of a leaf decreases slower than one over the number of $\CI$-samples for an honest tree. We let $\mathcal{A}$ denote the set of half-open rectangles in $[0, 1]^d$,
\begin{equation*}
    \mathcal{A} = \{(a_1, b_1] \times \cdots (a_d, b_d] : a_i, b_i \in [0, 1], a_i < b_i\}.
\end{equation*}
We remind the reader of the following quantities. For a class of measurable sets $\mathcal{C}$, $s(\mathcal{C}, n)$ is the \emph{shattering coefficient} of the class $\mathcal{C}$, that is, the maximal number of subsets that can be picked out by $\mathcal{C}$ of a set of size $n$. To be precise, if 
\begin{equation*}
    N_\mathcal{C}(x_1, ..., x_n) = \# \{\{x_1, ..., x_n\} \cap C : C \in \mathcal{C}\},
\end{equation*}
then $s(\mathcal{C}, n) := \max_{x_1, ..., x_n} N_\mathcal{C}(x_1, ..., x_n)$. The largest number $n$ such that $s(\mathcal{C}, n) = 2^n$ is called the \emph{VC-dimension} of $\mathcal{C}$ and is denoted by $\mathcal{V}(\mathcal{C})$. If $\mathcal{V}(\mathcal{C}) < \infty$, $\mathcal{C}$ is called a \emph{VC-class}. We refer the reader to chapters 12 and 13 of~\cite{Pattern} for more background. 

\begin{lem}\label{lem:Devroyebound}
Let $X_1, ..., X_n$ denote an iid sample in $[0, 1]^d$ with $d > 1$ and $\BP^{(n)}$ the corresponding empirical measure. For the class of sets $\mathcal{A}$, it holds for $\varepsilon \in (0, 1]$ that
\begin{equation*}
    \BP\Big(\sup_{A \in \mathcal{A}}|\BP^{(n)}(A) - \BP(X_1 \in A)| > \varepsilon\Big) \leq c n^{4d} e^{-2n\varepsilon^2}
\end{equation*}
for some constant $c \leq e^{4\varepsilon(1 + \varepsilon)}$. If $d = 1$, the result holds with the right hand side replaced by $c(n^4 + 1)e^{-2n\varepsilon^2}$.
\end{lem}
\begin{proof}
In~\cite{Devroye1982} it is shown that
\begin{equation*}
    \BP\Big(\sup_{A \in \mathcal{A}}|\BP^{(n)}(A) - \BP(X_1 \in A)| > \varepsilon\Big) \leq c s(\mathcal{A}, n^2) e^{-2n\varepsilon^2}
\end{equation*}
with $c$ as specified above. In Theorem 13.8 of~\cite{Pattern}, it is shown that $\mathcal{V}(\mathcal{A}) = 2d$, and if $d > 1$, Theorem 13.3 applies to show that $s(\mathcal{A}, n^2) \leq n^{4d}$. The same theorem covers the case $d = 1$. This completes the proof. 
\end{proof}

\subsection{General convergence results and conditioning}

The following result is valuable in showing convergence of a random sequence to infinity, both in probability and almost surely. In the following results, the notation for a sequence of subsets $\{A_n\} \subseteq \CF$
\begin{equation*}
    \{A_n \text{ i.o.}\} := \bigcap_{N = 1}^\infty \bigcup_{n = N}^\infty A_n \quad \text{and} \quad \{A_n \text{ evt.}\} := \bigcup_{N = 1}^\infty \bigcap_{n = N}^\infty A_n
\end{equation*}
is extensively used. Note that $\{A_n \text{ i.o.}\}^c = \{A_n^c \text{ evt.}\}$. 

\begin{lem}\label{lem:BCinfinity}
Let $\{X_n\}_n$ be a sequence of random variables satisfying $\BP(X_n > a_n) \geq b_n$ where $a_n \to \infty$.
\begin{enumerate}
    \item[(1)] If $b_n \to 1$, we have $X_n \overset{\BP}{\to} \infty$.
    \item[(2)] If furthermore $\sum_{n = 1}^\infty(1 - b_n) < \infty$, we have $X_n \to \infty$ $\BP$-a.s.
\end{enumerate}
\end{lem}
\begin{proof}
As for (1), we have to show that for every $c > 0$, $\BP(X_n \leq c) \to 0$ as $n \to \infty$. Let $c > 0$ and $\varepsilon > 0$ be given and choose $N$ such that for $n \geq N$, it holds that $a_n > c$ and $b_n > 1 - \varepsilon$. Then
\begin{equation*}
    \BP(X_n \leq c) = 1 - \BP(X_n > c) \leq 1 - \BP(X_n > a_n) \leq 1 - b_n < \varepsilon
\end{equation*}
as desired. As for (2), note that
\begin{equation*}
    \BP\Big(\lim_{n \to \infty} X_n = \infty \Big) = \BP\Big(\bigcap_{M \in \BN} \{X_n > M \text{ evt.}\} \Big),
\end{equation*}
so it suffices to show that for every $M \in \BN$, we have $\BP(X_n > M \text{ evt.}) = 1$, which is equivalent to $\BP(X_n \leq M \text{ i.o.}) = 0$. Choose $N \in \BN$ so large such that $a_n \geq M$ for all $n \geq N$. Then for $n \geq N$, it holds that
\begin{equation*}
    \BP(X_n \leq M) = 1 - \BP(X_n > M) \leq 1 - \BP(X_n > a_n) \leq 1 - b_n.
\end{equation*}
Thus,
\begin{equation*}
    \sum_{n = 1}^\infty \BP(X_n \leq M) \leq N - 1 + \sum_{n = N}^\infty (1 - b_n) < \infty,
\end{equation*}
and the proof is complete by the Borel--Cantelli lemma.
\end{proof}

Throughout the article, almost sure convergence is typically established by first conditioning on a fixed $\sigma$-algebra. The following lemma shows that the Borel--Cantelli lemma applies as one would expect in this context.

\begin{lem}[\textbf{Conditional Borel--Cantelli lemma}]\label{lem:conditionalBC}
Let $\{A_n\} \subseteq \CF$ be a sequence of sets and $\CG$ a fixed sub-$\sigma$-algebra of $\CF$. If
\begin{equation*}
    \sum_{n = 1}^\infty \BP(A_n \mid \CG) < \infty \quad \BP\text{-a.s.},
\end{equation*}
then $\BP(A_n \operatorname{ i.o.} \mid \CG) = 0$ $\BP$-a.s. and in particular, $\BP(A_n \operatorname{ i.o.}) = 0$. 
\end{lem}
\begin{proof}
The proof is essentially the same as the classical Borel--Cantelli lemma. Define the variable $Y := \sum_{n = 1}^\infty \1_{A_n}$. Then $\{Y = \infty\} = \{A_n \text{ i.o.}\}$ and so
\begin{equation*}
    \BE[Y \mid \CG] = \sum_{n = 1}^\infty \BP(A_n \mid \CG) < \infty
\end{equation*}
by assumption, which implies $\BP(A_n \text{ i.o.} \mid \CG) = \BP(Y = \infty \mid \CG) = 0$. The final assertion is just the tower property.
\end{proof}

The proof of the following corollary is unchanged from the unconditional case.

\begin{cor}\label{cor:conditionalBC}
Let $\{X_n\}$ be a sequence of random variables, $X$ a random variable and $\CG$ a sub-$\sigma$-algebra. If for any $\varepsilon > 0$,
\begin{equation*}
    \sum_{n = 1}^\infty \BP(|X_n - X| > \varepsilon \mid \CG) < \infty \quad \BP\text{-a.s.},
\end{equation*}
then $X_n \to X$ $\BP$-a.s.
\end{cor}

Throughout the paper, we often show strong convergence of a sequence $\{X_n\}$ to $X$ by showing that probabilities of the form
\begin{equation}\label{eq:condtail}
    \BP(|X_n - X| > \varepsilon \mid \CG_n)
\end{equation}
are summable, where $\{\CG_n\}$ is a filtration, but one must be careful in concluding anything about the convergence behaviour of $\{X_n\}$ based on this. Our way of getting around the issue with a growing $\sigma$-algebra in the conditioning is to consider
\begin{equation*}
    \CG_\infty := \sigma\Big(\bigcup_{n = 1}^\infty \CG_n \Big)
\end{equation*}
and verifying that the bound on~\eqref{eq:condtail} is still valid when conditioning on $\mathcal{G}_\infty$. The following lemma is an indispensible tool for this.

\begin{lem}\label{lem:condmean}
Let $\CD, \CG$ and $\CH$ be sub-$\sigma$-algebras with $\CG \subseteq \CD$ and $\CD \indep \CH$. If $Z$ is $\CD$-measurable with $\BE[|Z|] < \infty$, then it holds that
\begin{equation*}
    \BE[Z \mid \CG] = \BE[Z \mid \CG \lor \CH]
\end{equation*}
where $\CG \lor \CH := \sigma(\CG, \CH)$.
\end{lem}
\begin{proof}
By definition, $\BE[Z \mid \CG]$ is $\CG$-measurable and thus, a fortiori, $\CG \lor \CH$-measurable. Let $G \in \CG$ and $H \in \CH$. Since sets of the form $G \cap H$ constitute a $\cap$-stable generator for $\CG \lor \CH$, we are done once we verify that
\begin{equation*}
    \int_{G \cap H} \BE[Z \mid \CG]\mathrm{d}\BP = \int_{G \cap H} Z \mathrm{d}\BP.
\end{equation*}
This follows from the calculation
\begin{align*}
    \int_{G \cap H} \BE[Z \mid \CG]\mathrm{d}\BP &= \int \BE[Z \mid \CG] \1_G \1_H \mathrm{d}\BP\\
    &= \BP(H) \int_G \BE[Z \mid \CG]\mathrm{d}\BP= \BP(H) \int_G Z \mathrm{d}\BP = \int_{G \cap H} Z \mathrm{d}\BP.
\end{align*}
In the second equality, we used independence of $\CG$ and $\CH$, in the third we used the definition of a conditional expectation, and in the last, we used independence of $\CD$ and $\CH$. 
\end{proof}

\end{document}